\documentclass[a4paper,11pt,reqno]{amsart}

\usepackage[headings]{fullpage}
\usepackage[utf8]{inputenc}
\usepackage{mathrsfs}
\usepackage{dsfont}
\usepackage{hyperref}
\hypersetup{bookmarksdepth=3}
\usepackage{amsmath}
\newcommand\numberthis{\addtocounter{equation}{1}\tag{\theequation}}
\allowdisplaybreaks[4]
\usepackage{amssymb}
\usepackage{amsthm}
\usepackage{amsfonts}
\usepackage{amstext}
\usepackage{amsopn}
\usepackage{amsxtra}
\usepackage{mathrsfs}
\usepackage{dsfont}
\usepackage{esint}
\usepackage{enumitem}
\usepackage[capitalize]{cleveref}
\usepackage{braket}
\usepackage{mathtools}

\usepackage{subcaption}

\usepackage{tikz}
\usetikzlibrary{arrows}
\usepackage{tkz-tab}
\usetikzlibrary{math}
\usepackage{pgfplots}
\pgfplotsset{compat=1.18}

\usepackage[dvipsnames]{xcolor}
 
\usepackage{csquotes}
\usepackage[backend=bibtex,sorting=nty,citestyle=numeric-comp,bibstyle=ieee,maxbibnames=99,doi=false,isbn=false,url=false,eprint=false,dashed=true]{biblatex}
\AtEveryBibitem{\clearfield{month}}
\AtEveryBibitem{\clearfield{day}}
\theoremstyle{plain}
\newtheorem{theorem}{Theorem}[section]
\newtheorem{lemma}[theorem]{Lemma}
\newtheorem{proposition}[theorem]{Proposition}
\newtheorem{corollary}[theorem]{Corollary}

\theoremstyle{definition}

\newtheorem{assumption}[theorem]{Assumption}

\theoremstyle{remark}
\newtheorem{remark}[theorem]{Remark}

\numberwithin{equation}{section}

\newcommand{\bbB}{\mathbb{B}}
\newcommand{\C}{\mathbb{C}}
\newcommand{\bbH}{\mathbb{H}}
\newcommand{\R}{\mathbb{R}}

\newcommand{\bbV}{\mathbb{V}}
\newcommand{\Z}{\mathbb{Z}}

\newcommand{\ii}{\infty}
\newcommand\1{{\ensuremath {\mathds 1} }}
\renewcommand\phi{\varphi}

\newcommand{\cQ}{\mathcal{Q}}
\newcommand{\cV}{\mathcal{V}}

\newcommand{\cM}{\mathcal{M}}

\newcommand{\cE}{\mathcal{E}}
\newcommand{\cF}{\mathcal{F}}
\newcommand{\cN}{\mathcal{N}}

\newcommand{\cH}{\mathcal{H}}
\newcommand{\cI}{\mathcal{I}}
\newcommand{\cJ}{\mathcal{J}}
\newcommand{\cK}{\mathcal{K}}
\newcommand{\cL}{\mathcal{L}}

\newcommand{\tr}{{\rm Tr}\,}

\renewcommand{\geq}{\geqslant}
\renewcommand{\leq}{\leqslant}

\renewcommand{\tilde}{\widetilde}

\newcommand{\eps}{\varepsilon}
\usepackage{color}

\newcommand{\bx}{\mathbf{x}}

\newcommand{\by}{\mathbf{y}}

\newcommand{\loc}{\textrm{\normalfont loc}}

\newcommand{\dd}{\mathrm{d}}

\newcommand{\Supp}{\mathrm{Supp}}

\newcommand{\op}{\mathrm{op}}

\newcommand{\hc}{\mathrm{h.c.}}
\newcommand{\ui}{\mathbf{i}}
\newcommand{\diver}{\mathrm{div}}
\newcommand{\ren}{\textrm{\normalfont ren}}
\newcommand{\GP}{\textrm{\normalfont GP}}
\newcommand{\otimessym}{\otimes_{\textrm{s}}}

\renewcommand{\d}[1]{\ensuremath{\operatorname{d}\!{#1}}}
\newcommand{\sym}{\textmd{\normalfont s}}
\newcommand{\ext}{\textmd{\normalfont ext}}
\newcommand{\trap}{\textmd{\normalfont trap}}

\newcommand{\twoB}{\textmd{\normalfont 2B}}
\newcommand{\threeB}{\textmd{\normalfont 3B}}
\newcommand{\fourB}{\textmd{\normalfont 4B}}
\newcommand{\fiveB}{\textmd{\normalfont 5B}}

\newcommand{\mfH}{\mathfrak{H}}

\author[A. Triay]{Arnaud Triay}
\address{Department of Mathematics, LMU Munich, Theresienstrasse 39, 80333 Munich, Germany}
\email{triay@math.lmu.de}
\author[F. L. A. Visconti]{François L. A. Visconti}
\address{Department of Mathematics, LMU Munich, Theresienstrasse 39, 80333 Munich, Germany}
\email{visconti@math.lmu.de}

\title[Derivation 3D quintic Gross--Pitaevskii equation]{Derivation of the 3D quintic Gross--Pitaevskii equation}

\date{}

\begin{document}
 
\maketitle
 
\begin{abstract}
	We study the time evolution of Bose--Einstein condensates with three-body interactions in the Gross--Pitaevskii regime. We show that Bose--Einstein condensation is preserved under many-body evolution and that the condensate wavefunction evolves according to the quintic Gross--Pitaevskii equation in $\R^3$, which is energy critical. In particular, we show that the effective coupling constant is universal and depends only on a three-body scattering hypervolume.
\end{abstract}
 
\tableofcontents

\section{Introduction and main result}

	In this paper, we derive the quintic Gross--Pitaevskii equation from the many-body Schrödinger equation with three-body interactions. We consider the limit of a large number of bosons in the Gross--Pitaevskii (GP) regime and prove that Bose--Einstein condensation (BEC) is preserved in time. In particular, we establish universality in the sense that the coupling constant of the Gross--Pitaevskii equation depends only on the three-body scattering hypervolume of the interaction potential and not on its microscopic details.

	The robustness of BEC in the presence of interactions is an archetypal example of propagation of chaos in quantum systems. Hepp first studied this in the context of classical limits of quantum fields \cite{Hepp74},Ginibre and Velo extended it to singular interactions \cite{GinibreVel79,GinibreVel79b}, and Spohn further developed it for more general systems \cite{Spohn80}. For particle systems, Bardos, Golse and Mauser \cite{BardosGolMau00} proved that, in the weak coupling regime, the evolution of factorised initial data is governed by the non-linear Schrödinger equation with a Hartree-type nonlinearity. Their method, based on ideas from \cite{Spohn80}, consists of studying the so-called BBGKY hierarchy, and was later extended to the singular case of Coulomb interactions by Erdös and Yau \cite{ErdoesYau01}.
	
	In dilute regimes---which closely reflect experimental setups---particle collisions are rare but strong, and interactions become delta-like, making the mathematical treatment intricate. The dilute limit of systems with two-body interactions was studied in seminal works by Erdös, Schlein and Yau \cite{ErdoesSchYau07,ErdoesSchYau07b,ErdoesSchYau09,ErdoesSchYau10}, where they proved the convergence of the many-body dynamics to the cubic Gross--Pitaevskii equation; the argument was later simplified in \cite{ChenHaiPavSei15}. Remarkably, they showed that the coupling constant depends only on the scattering length of the interaction potential. Until the present work, no such universality result had been proved for systems with three-body interactions.
	
	A lot is known about the dynamics of Bose gases with two-body interactions: we refer to \cite{BenediktOliSch15,BrenneckKro24,BrenneckSch19,Pickl11,RodniansSch09} for results on the rate of convergence of the density matrices, to \cite{CaraciOldSch24,GrillakiMacMar10,GrillakiMacMar11,LewinNamSch15,MitrouskPetPic19} for norm approximations, and to \cite{ChenLeeSch11,FroehlicKnoSch09,KnowlesPic10,Pickl15} for related results.
	
	\subsection{The model}
	
	We consider the dynamics of a gas of $N$ three-dimensional bosons interacting via a repulsive three-body potential $V_N$ of the form
	\begin{equation}
		\label{eq:gross_pitaevskii_scaling}
		V_N(x,y,z) = NV(N^{1/2}(x - y, x - z)),
	\end{equation}
	for some fixed function $V$. The system is described by the Hamiltonian
	\begin{equation}
		\label{eq:hamiltonian}
		H_N = \sum_{i=1}^{N} -\Delta_i + \sum_{1\leq i<j<k\leq N}V_N(x_i,x_j,x_k)
	\end{equation}
	acting on $\mfH^N$, the $N$-fold symmetric tensor product of the one-body Hilbert space $\mfH = L^2(\R^3)$.
	
	We study the evolution of initial data $\Psi_{N,0}$ exhibiting Bose--Einstein condensation in some one-body wavefunction $\varphi_0$, meaning that
	\begin{equation*}
		\label{eq:inital_data_bec}
		\gamma_N^{(1)} \xrightarrow[N\rightarrow\ii]{} \vert\varphi_0\rangle\langle\varphi_0\vert
	\end{equation*}
	in the trace-norm topology. Here, $\gamma_N^{(k)}$ denotes the $k$-particle reduced density matrix $\Psi_{N,0}$, which is defined as the integral operator whose kernel is
	\begin{equation}
		\label{eq:1_prdm_def}
		\gamma_N^{(k)}(\bx_k;\by_k) = \int_{\R^{3(N-k)}}\dd{\mathbf{z}_{N-k}}\Psi_{N,0}(\bx_k,\mathbf{z}_{N-k})\overline{\Psi_{N,0}(\by_k,\mathbf{z}_{N-k})}.
	\end{equation}	
	We also require that $\Psi_{N,0}$ satisfies the energy compatibility relation
	\begin{equation}
		\label{eq:inital_data_energy_compatibility}
		\langle\Psi_{N,0},H_N\Psi_{N,0}\rangle = N\cE^\GP[\varphi_0] + o(N),
	\end{equation}
	where $\cE^\GP$ denotes the \textit{Gross--Pitaevskii energy functional}
	\begin{equation}
		\label{eq:gross_pitaevskii_energy}
		\mathcal{E}^{\GP}[\varphi] = \int_{\R^3}\vert\nabla \varphi\vert^2 + \dfrac{b(V)}{6}\int_{\R^3}\vert \varphi\vert^6.
	\end{equation}
	Here, $b(V)$ is the \textit{three-body scattering hypervolume} of $V$ and it is defined by
	\begin{equation}
		\label{eq:three_body_scattering_hypervolume}
		b(V) = \int_{\R^6}Vf,
	\end{equation}
	where $f$ denotes the solution to the \textit{zero-energy scattering equation}
	\begin{equation}
		\label{eq:zero_energy_scattering_equation}
		(-\Delta_x - \Delta_y - \Delta_z)f(x - y,x - z) + (Vf)(x - y,x - z) = 0
	\end{equation}
	with boundary condition $f(\bx) \rightarrow 1$ as $\vert\bx\vert \rightarrow \ii$. See \cite{NamRicTri23} for more information. Note that the state $\varphi_0^{\otimes N}$ does not satisfy \eqref{eq:inital_data_energy_compatibility} as it lacks the appropriate short-range correlation structure; its energy is given, to the leading order, by $N\cE^\GP[\varphi_0]$ with $b(V)$ replaced by $\int V$. The validity of \eqref{eq:inital_data_energy_compatibility} for the ground state energy of trapped bosons was established in \cite{NamRicTri22,NamRicTri22b,NamRicTri23}. Moreover, the correction to the leading order was derived recently by Brooks \cite{Brooks25} and is of order $\sqrt{N}$. We refer to \cite{JungeVis26b,NguyenRic24,NguyenRic25,Visconti26} for related works in the stationary setting.
	
	We show that the solution $\Psi_{N,t}$ to the many-body Schrödinger equation
	\begin{equation}
		\label{eq:schroedinger_eq_time_dependent}
		\ui\partial_t\Psi_{N,t} = H_N\Psi_{N,t}
	\end{equation}
	with initial data $\Psi_{N,0}$ exhibits Bose--Einstein condensation in the solution $\varphi_t$ to the quintic Gross--Pitaevskii equation
	\begin{equation}
		\label{eq:gross_pitaevskii_equation}
		\ui\partial_t\varphi_t = \Big(-\Delta + \dfrac{b(V)}{2}\vert\varphi_t\vert^4 - \mu_t\Big)\varphi_t
	\end{equation}
	with initial data $\varphi_0$. Said differently, we prove that the 1-particle reduced density matrix $\gamma_{N,t}^{(1)}$ of $\Psi_{N,t}$ satisfies
	\begin{equation*}
		\gamma_{N,t}^{(1)} \xrightarrow[N\rightarrow\infty]{} \vert\varphi_t\rangle\langle\varphi_t\vert.
	\end{equation*}
	As proved by Lieb and Seiringer in \cite{LiebSei02}, this implies, for all $k \geq 1$,
	\begin{equation*}
		\gamma_{N,t}^{(k)} \xrightarrow[N\rightarrow\infty]{} \vert\varphi_t^{\otimes k}\rangle\langle\varphi_t^{\otimes k}\vert.
	\end{equation*}
	The gauge parameter $\mu_t$ in \eqref{eq:gross_pitaevskii_equation} can be chosen freely and for convenience we take it to be
	\begin{equation*}
		\mu_t = \dfrac{b(V)}{3}\int_{\R^3}\vert\varphi_t\vert^6,
	\end{equation*}
	for which we have the compatibility of the energies
	\begin{equation}
		\label{eq:gross_pitaevskii_drift_term}
		\langle\ui\partial_t\varphi_t,\varphi_t\rangle = \int_{\R^3}\vert\nabla\varphi_t\vert^2 + \dfrac{b(V)}{2}\int_{\R^3}\vert\varphi_t\vert^6 - \mu_t = \cE^\GP[\varphi_t].
	\end{equation}
	The well-posedness of \eqref{eq:gross_pitaevskii_equation} in $\dot{H}^1(\R^3)$ was first proved for radial data by Bourgain \cite{Bourgain99} and Grillakis \cite{Grillaki00}, and later extended to general data by Colliander, Keel, Staffilani, Takaoka and Tao \cite{ColliandKeeStaTak08}; we collect important properties in Appendix \ref{sec:prop_gp_equation}.
	
	The dynamics of Bose gases with three-body interactions have already been studied in the weak-coupling regime, where the interaction potential is given by
	\begin{equation}
		\label{eq:weak_coupling_regime}
		V_{N,\beta}(x,y,z) = N^{6\beta - 2}V(N^\beta(x - y,x - z)),
	\end{equation}
	for some parameter $0 < \beta < 1/2$. In that case, a simple adaptation of our proof shows that the limiting equation is also a quintic nonlinear Schrödinger equation like \eqref{eq:gross_pitaevskii_equation}, but where the coupling constant $b(V)$ is replaced by the non-physical parameter $\int V$. This was first proved by Chen and Pavlović \cite{ChenPav11} for $0 < \beta < (4(d + 1))^{-1}$ in dimensions $d = 1,2$, and later by Chen and Holmer \cite{ChenHol19} for $0 < \beta < 1/9$ in dimension $d = 3$ with $H^1$ initial data. The dynamics of fluctuations around the condensate were studied by Chen \cite{Chen11} in the mean-field regime $\beta = 0$, and this was later extended to $0 < \beta < 1/6$ by Nam and Salzmann \cite{NamSal19}. For further works, we refer to \cite{AdamiLee25,BabbRou26,Lee21,LiYao21,RoutSoh25,Xie15,Yuan15}.
	
	In the present paper, we study the critical case $\beta = 1/2$---corresponding to \eqref{eq:gross_pitaevskii_scaling}---which is known as the Gross--Pitaevskii regime. Our goal is to derive the Gross--Pitaevskii equation \eqref{eq:gross_pitaevskii_equation}, and to show that the coupling constant is universal in the sense that it depends only on the three-body scattering hypervolume $b(V)$ of $V$ defined in \eqref{eq:three_body_scattering_hypervolume}. The proof follows the overall strategy of \cite{BenediktOliSch15,BrenneckSch19} from the two-body case, and uses simplifying ideas from \cite{BrenneckKro24,Brooks24}. However, when adapting the method to the three-body case, we face important difficulties that have no counterpart in the two-body case. Specifically, we need to control some effective two-body interactions that arise when two excited particles collide with a particle in the condensate. It turns out that controlling these short-range pairwise interactions in terms of short-range three-body interactions is particularly challenging. In fact, by considering many pairs far away from each other, it is clear that three-body interactions alone are not sufficient. To solve this problem, we developed a new multi-scale localisation method, which we believe is interesting on its own; see Section~\ref{sec:effective_interaction_potentials_non_renormalised}. In particular, this method allows us to prove, for any compactly supported potential $0\leq W\in L^\ii(\R^3)$, the following stability of the second kind result:
	\begin{equation}
		-C_\varepsilon M \leq \varepsilon\sum_{i = 1}^M-\Delta_i - \sum_{1\leq i<j\leq M}N^{1/2}W(N^{1/2}(x_i - x_j)) + \varepsilon\sum_{1\leq i<j<k\leq M}V_N(x_i,x_j,x_k),
	\end{equation}
	for any $\varepsilon > 0$. More precise statements are made in Propositions \ref{prop:general_two_body_bound_renormalised} and \ref{prop:general_three_and_four_body_bound_renormalised}.

\subsection{Main result}

	Suppose that the interaction potential $V$ in \eqref{eq:gross_pitaevskii_scaling} satisfies the following assumptions.
	\begin{assumption}
		\label{ass:potential}
		The potential $V:\R^3\times\R^3\rightarrow \R$ is nonnegative, has compact support and satisfies the following \textit{three-body symmetry} properties
		\begin{equation}
			\label{eq:three_body_symmetry_properties}
			V(x,y) = V(y,x) \quad \textmd{and} \quad V(x-y,x-z) = V(y-x,y-z) = V(z-x,z-y).
		\end{equation}
		Moreover, $V$ is bounded from above by a radially symmetric and decreasing function of class $L^{2 + \eta}(\R^6),$ for some $\eta > 0$. In addition, there exists $R > 0$ and $C > 0$ such that
		\begin{equation}
			\label{eq:three_body_potential_bounded_below_assumption}
			V \geq C \quad \textmd{on $B(0,R)$.}
		\end{equation}
	\end{assumption}
	
	Under Assumption \ref{ass:potential}, the Hamiltonian $H_N$ leaves the bosonic space $\mfH^N$ invariant.

\begin{theorem}\label{th:main_result}
	Let $V$ satisfy Assumption~\ref{ass:potential}. Let $(\Psi_{N,0})_N$ be a sequence of normalised wavefunctions in $\mfH^N$, and let $\gamma_N^{(1)}$ denote the $1$-particle reduced density matrix of $\Psi_{N,0}$. Let $\varphi_0\in H^6(\R^3)$ be normalised and define
	\begin{equation*}
		\eps_{1,N} = \tr\big\vert\gamma_N^{(1)} - \vert\varphi_0\rangle\langle\varphi_0\vert\big\vert
	\end{equation*}
	and
	\begin{equation*}
		\eps_{2,N} = \big| N^{-1} \braket{\Psi_{N,0},H_N \Psi_{N,0}} - \mathcal E^{\GP}[\varphi_0] \big|.
	\end{equation*}
	
	Let $\Psi_{N,t} = e^{-\ui tH_N}\Psi_{N,0}$ be the solution to the Schrödinger equation \eqref{eq:schroedinger_eq_time_dependent} with initial data $\Psi_{N,0}$, and let $\gamma_{N,t}^{(1)}$ denote its $1$-particle density matrix. Let $\varphi_t$ be the solution to the time-dependent Gross--Pitaevskii equation \eqref{eq:gross_pitaevskii_equation} with initial data $\varphi_0$. Then for all $t\in\R$,
	\begin{equation}
		\label{eq:bec_main_bound}
		1 - \langle\varphi_t,\gamma_{N,t}^{(1)}\varphi_t\rangle \leq C\big(\varepsilon_{1,N} + \varepsilon_{2,N} + N^{-1/2}\big)\exp(c\vert t\vert),
	\end{equation} 
	for some constants $C,c > 0$ depending only on $\Vert\varphi_0\Vert_{H^6}$ and $V$.
\end{theorem}

Here are some remarks on the result.
\begin{enumerate}
	\item Supposing that $\varepsilon_{1,N},\varepsilon_{2,N} = \mathcal{O}(N^{-1/2})$, we get
	\begin{equation*}
		1 - \langle\varphi_t,\gamma_{N,t}^{(1)}\varphi_t\rangle \leq CN^{-1/2}\exp(c\vert t\vert),
	\end{equation*}
	meaning that Bose--Einstein condensation holds with rate $N^{-1/2}$. It was shown in \cite{Brooks25} that for the ground state of a tapped Bose gas both $\varepsilon_{1,N}$ and $\varepsilon_{2,N}$ are of order $N^{-1/2}$, meaning that we cannot hope to get a better rate of convergence than the present one without resolving the energy up to the next order.
	\item The assumption that $V$ is bounded from above by a radially symmetric and decreasing function of class $L^{2+\eta}$ is technical and ensures that the function $x\in\R^3 \mapsto \Vert V(x,\cdot)\Vert_{L^1}$ is bounded. This property is used only in the proofs of Propositions \ref{prop:general_two_body_bound_renormalised} and \ref{prop:general_three_and_four_body_bound_renormalised}, and it is stated more generally in Lemma \ref{lemma:mixed_norm_regularity_estimate}.
	\item Our proof can be easily adapted---following for example \cite{BrenneckKro24,BrenneckSch19}---to deal with systems prepared in a trap that is turned off at $t = 0$. Namely, we consider initial data $\Psi_{N,0}$ that are approximate ground states of the Hamiltonian
	\begin{equation}
		\label{eq:hamiltonian_trapped}
		H_N^\trap = \sum_{i = 1}^N-\Delta_i + V^\ext(x_i) + \sum_{1 \leq i < j < k \leq N}V_N(x_i,x_j,x_k),
	\end{equation}
	for some external potential $0\leq V^\ext\in L_\loc^\ii(\R^3;\R)$ growing polynomially at infinity, and we consider initial data for the Gross--Pitaevskii equation \eqref{eq:gross_pitaevskii_equation} that are minimisers of the functional
	\begin{equation}
		\label{eq:gross_pitaevskii_equation_trapped}
		\cE_\GP^\trap[\varphi] = \int\vert\nabla\varphi\vert^2 + V^\ext\vert\varphi\vert^2 + \dfrac{b(V)}{6}\int\vert\varphi\vert^6.
	\end{equation}
	Then, Theorem \ref{th:main_result} remains valid if $\varepsilon_{1,N}$ and $\varepsilon_{2,N}$ are replaced with
		\begin{equation*}
		\tilde{\eps}_{1,N} = 1 - \langle\varphi_\GP,\gamma_N^{(1)}\varphi_\GP\rangle
	\end{equation*}
	and
	\begin{equation*}
		\tilde{\eps}_{2,N} = \big| N^{-1} \braket{\Psi_{N,0},H_N^\trap \Psi_{N,0}} - \mathcal E_\GP^\trap[\varphi_\GP] \big|,
	\end{equation*}
	where $\varphi_\GP$ is the ground state of $\cE_\GP^\trap$. It was proved in \cite{NamRicTri23} that approximate ground states of \eqref{eq:hamiltonian_trapped} satisfy $\tilde{\varepsilon}_{1,N},\tilde{\varepsilon}_{2,N} \rightarrow 0$. Moreover, under sufficient regularity and growth assumptions on $V^\ext$, the minimiser $\varphi_\GP$ belongs to $H^6(\mathbb{R}^{3})$.
\end{enumerate}
	
	\bigskip
	\noindent{\bf Organisation of the paper.} In Section \ref{sec:strategy_and_proof_main_theorem}, we set some notation, explain the main strategy and give the proof of Theorem \ref{th:main_result} assuming the validity of Proposition \ref{prop:gronwall_bound}, which is the main estimate used in the Grönwall argument. Section \ref{sec:ana_fluct} is devoted to the proof of Proposition~\ref{prop:gronwall_bound}. In Section \ref{sec:effective_interaction_potentials_non_renormalised}, we show the many-body estimates which are the new and central ingredients in our analysis. Finally, in Section \ref{sec:effective_interaction_potentials_renormalised}, we show that these inequalities also hold if we replace the kinetic and interaction terms by their renormalised counterparts. In Appendix \ref{sec:prop_gp_equation}, we recall some known properties of the solution to the Gross--Pitaevskii equation \eqref{eq:gross_pitaevskii_equation}. In Appendix \ref{app:op_bounds} we prove some operator bounds involving the correlation kernels. Lastly, in Appendix \ref{app:theta_coefficients} we give some technical estimates on the prefactors appearing in the reformulation of dynamics in the excitation Fock space.
	
	\bigskip
	\noindent{\bf Acknowledgements.} The authors thank C. Brennecke, J. Chong and P. T. Nam for valuable discussions. This work was partly funded by the Deutsche Forschungsgemeinschaft (DFG, German Research Foundation) through the TRR 352 Project ID	47090307. F. L. A. V. acknowledges partial support from the European Research Council through the ERC CoG RAMBAS Project Nr. 101044249. 

\section{Notation, strategy and proof of Theorem~\ref{th:main_result}}

\label{sec:strategy_and_proof_main_theorem}

In this section, we outline the proof of Theorem~\ref{th:main_result}. First, we rewrite the problem in terms of fluctuations around the condensate wavefunction. Then, we carefully analyse the generator of the fluctuation dynamics to prove a Grönwall estimate that allows us to bound the time evolution of the number of excitations. In order to be able to close this Grönwall bound, we need to remove correlations from the ground state of the system, which we do using a cubic transformation. To avoid the use of operator exponential expansions, we approximate this transformation using algebraic ideas.

\subsection{Notation and preliminaries}

We fix here some notation and refer to \cite{LewinNamSch15,LewinNamSerSol15} for precise definitions regarding the Fock space formalism.

Denote by $\Vert\cdot\Vert_p$ the usual norm on $L^p(\R^d)$ when there is no possible confusion on the dimension $d$. Moreover, denote by $\Vert\cdot\Vert_{L^pL^\infty}$ the norm defined by
\begin{equation*}
	\Vert f\Vert_{L^pL^\infty} = \sup_{x\in\R^3}\Vert f(x,\cdot)\Vert_{L^p},
\end{equation*}
for any function $f:\R^6\rightarrow \C$. Also, when there is no ambiguity we often omit integration variables, e.g. we write $\int f$ instead of $\int\d{}xf(x)$.

Given a function $F$ defined on $\R^6$ that satisfies the three-body symmetry \eqref{eq:three_body_symmetry_properties}, we will often abuse notation and write $F(x,y,x)$ to denote $F(x - y,x - z)$.

Let
\begin{equation*}
	\cF = \bigoplus_{k = 0}^\ii\mfH^k
\end{equation*}
be the Fock space constructed from the one-body space $\mfH$. Let $a^*(f)$ and $a(f)$ be the usual creation and annihilation operators on $\cF$, and let $a_x^*$ and $a_x$ denote the corresponding operator-valued distributions. These operators satisfy the canonical commutation relations (CCR)
\begin{equation}
	\label{eq:creation_annihilation_op_ccr}
	\begin{aligned}
		[a(f),a^*(g)] = \langle f,g\rangle, \quad &[a(f),a(g)] = [a^*(f),a^*(g)] = 0,\\
		[a_x,a_y^*] = \delta_{xy}, \quad &[a_x,a_y] = [a_x^*,a_y^*] = 0,
	\end{aligned}
\end{equation}
for all $f,g\in\mfH$ and $x,y\in\R^3$. Let $\cN$ denote the number operator on $\cF$. It can be written in terms of $a_x$ and $a_x^*$ as
\begin{equation*}
	\cN = \int_{\R^3}\d{}xa_x^*a_x \quad \textmd{on $\cF$}.
\end{equation*}
Let $\Omega$ denote the vacuum state $(1,0,\dots)\in\cF$. The second quantisations of the operators $-\Delta$ and $V_N$ are given by
\begin{equation}
	\label{eq:second_quantised_kinetic_and_interaction}
	\d{}\Gamma(-\Delta) = \int_{\R^3}a_x^*(-\Delta_x)a_x \quad \textmd{and} \quad \bbV_N = \int_{\R^9}V_N(x,y,z)a_x^*a_y^*a_z^*a_xa_ya_z.
\end{equation}

Let $\varphi_t$ be the solution of the Gross--Pitaevskii equation \eqref{eq:gross_pitaevskii_equation} and define the projection $Q_t = 1 - \vert\varphi_t\rangle\langle\varphi_t\vert$. Define also the space $\mfH_+(t)$ of one-body wavefunctions orthogonal to $\varphi_t$. In terms of the projection $Q_t$, this space is given by $\mfH_+(t) = Q_t\mfH$. Denote by $\mfH_+^k(t)$ the $k$-fold symmetric tensor product of $\mfH_+(t)$. Define the Fock space of excitations $\cF_+(t)$ and the truncated Fock space of excitations $\cF_+^{\leq N}(t)$ by
\begin{equation*}
	\cF_+(t) = \bigoplus_{k = 0}^\ii\mfH_+^k(t) \quad  \textmd{and} \quad \cF_+^{\leq N}(t) = \bigoplus_{k = 0}^N\mfH_+^k(t).
\end{equation*}
Since $\mfH_+(t)$ is a subset of $\mfH$, we have the inclusions
\begin{equation}
	\label{eq:fock_spaces_inclusion}
	\cF_+^{\leq N}(t) \subset \cF_+(t) \subset \cF.
\end{equation}
Let $\cN_+$ be the number operator on $\cF_+(t)$. This operator coincides with the restriction of $\cN$ to $\cF_+(t)$, meaning that
\begin{equation*}
	\cN_+ = \int_{\R^3}a_x^*a_x \quad \textmd{on $\cF_+(t)$.}
\end{equation*}
Moreover, due to the inclusion \eqref{eq:fock_spaces_inclusion}, the operator $\cN_+$ can also be viewed as an operator acting on $\cF$ defined by
\begin{equation*}
	\cN_+ = \cN - a^*(\varphi_t)a(\varphi_t) \quad \textmd{on $\cF$}.
\end{equation*}

\subsection{Factoring out the condensate}

Again, let $\varphi_t$ be the solution of the Gross--Pitaevskii equation \eqref{eq:gross_pitaevskii_equation}. As observed in \cite{LewinNamSch15,LewinNamSerSol15}, any wavefunction $\Psi\in \mfH^N$ can be uniquely represented as
\begin{equation*}
	\Psi = \sum_{k = 0}^N\psi_k\otimessym \varphi_t^{\otimes N - k},
\end{equation*}
for some family $(\psi_k)_{k = 0}^N$ of wavefunctions with $\psi_k\in\mfH_+^k(t)$. Here $\otimes_\sym$ denotes the symmetric tensor product. This observation allows us to define the unitary map
\begin{align*}
	U_N: \mfH^N &\rightarrow \mathcal{F}_+^{\leq N}(t)\\
	\Psi &\mapsto \psi_0\oplus\psi_1\oplus\cdots\oplus\psi_N.
\end{align*}
This map factors out the condensate described by the one-body wavefunction $\varphi_t$ and lets us focus on the excitations orthogonal to the condensate. As was shown in \cite{LewinNamSch15}, the action of $U_N$ on creation and annihilation operators can be computed explicitly and is given by
\begin{equation}
	\label{eq:excitation_map_action_creation_annihilation_op}
	\begin{aligned}[c]
		U_N a^*(\varphi_t)a(\varphi_t) U_N^* &= N - \mathcal{N}_+, \quad & U_N a^*(\varphi_t)a(g) U_N^* &= \sqrt{N - \mathcal{N}_+}a(g),\\
		U_Na^*(f)a(g) U_N^* &= a^*(f) a(g), \quad & U_N a^*(f)a(\varphi_t) U_N^* &= a^*(f) \sqrt{N - \mathcal{N}_+},
	\end{aligned}
\end{equation}
for all $f,g \in \mathfrak{H}_+(t)$. In view of these relations, it is natural to introduce the modified creation and annihilation operators
\begin{equation*}
	b^*(f) = a^*(Q_tf) \quad \textmd{and} \quad b(f) = a(Q_tf),
\end{equation*}
for all $f\in\mfH$, as well as the corresponding operator-value distributions $b_x^*$ and $b_x$ defined by
\begin{equation*}
	b^*(f) = \int_{\R^3}f(x)b^*_x \quad \textmd{and} \quad b(f) \int_{\R^3}\overline{f(x)}b_x.
\end{equation*}
These operators satisfy the canonical commutation relations
\begin{equation}
	\label{eq:creation_annihilation_op_modified_ccr}
	\begin{aligned}
		[b(f),b^*(g)] = \langle f,Q_tg\rangle, \quad &[b(f),b(g)] = [b^*(f),b^*(g)] = 0,\\
		[b_x,b_y^*] = Q_t(x,y), \quad &[b_x,b_y] = [b_x^*,b_y^*] = 0,
	\end{aligned}
\end{equation}
for all $f,g\in\mfH$ and $x,y\in\R^3$. Here, $Q_t(x,y) = \delta_{xy} - \overline{\varphi_t(y)}\varphi_t(x)$ denotes the integral kernel of $Q_t$. In terms of $b_x^*$ and $b_x$, the operator $\cN_+$ is given by
\begin{equation}
	\label{eq:number_excitation_op_int_form}
	\cN_+ = \int_{\R^3}b_x^*b_x \quad \textmd{on $\cF$}.
\end{equation}

Lastly, define the second quantised operators
\begin{equation}
	\label{eq:second_quantised_kinetic_and_interaction_mod}
	\cK = \int_{\R^3}b_x^*(-\Delta_x)b_x \quad \textmd{and} \quad \cV_N = \int_{\R^9}V_N(x,y,z)b_x^*b_y^*b_z^*b_xb_yb_z.
\end{equation}
As quadratic forms on $\cF_+^{\leq N}(t)$, they coincide with $\d{}\Gamma(-\Delta)$ and $\bbV_N$.

\subsection{Fluctuations around the condensate.} 

Given a solution $\Psi_{N,t}$ of the time-dependent Schrödinger equation \eqref{eq:schroedinger_eq_time_dependent}, it is natural to consider the fluctuation vector $\Phi_{N,t}$ defined by
\begin{equation*}
	\Phi_{N,t} = U_N\Psi_{N,t}.
\end{equation*}
This vector solves the Hamiltonian equation
\begin{equation}
	\label{eq:fluctuation_vector_equation}
	\ui\partial_t\Phi_{N,t} = \mathcal{H}_N\Phi_{N,t},
\end{equation}
where the generator $\cH_N$ of the fluctuation dynamics is given by
\begin{equation}
	\label{eq:generator_fluctuations}
	\mathcal{H}_N = (\ui\partial_t U_N)U_N^* + U_NH_NU_N^*.
\end{equation}
Moreover, the fluctuation vector $\Phi_{N,t}$ is such that
\begin{equation}
	\label{eq:number_excitation_op_rewritten_rdm}
	1 - \langle\varphi_t,\gamma_{N,t}^{(1)}\varphi_t\rangle = N^{-1}\langle\cN_+\rangle_{\Psi_{N,t}} = N^{-1}\langle\cN_+\rangle_{\Phi_{N,t}}.
\end{equation}
This results from
\begin{equation*}
	U_N\cN_+U_N^* = \cN_+,
\end{equation*}
which is itself a consequence of \eqref{eq:excitation_map_action_creation_annihilation_op}.

A naive approach to prove Theorem~\ref{th:main_result} would be to control the growth of the number of excitations $\cN_+$ using the Grönwall lemma, as was done in \cite{LewinNamSch15} for systems with two-body interactions in the mean-field regime. However, when evaluating the time derivative
\begin{equation*}
	\partial_t\langle\cN_+\rangle_{\Phi_{N,t}} = \langle\ui\left[\cH_N,\cN_+\right]\rangle_{\Phi_{N,t}} = \langle\ui\left[H_N,\cN_+\right]\rangle_{\Psi_{N,t}} - 2\Re\langle a^*(\partial_t\varphi_t)a(\varphi_t)\rangle_{\Psi_{N,t}},
\end{equation*}
we see that it contains several contributions of order $N$. These contributions result from short-range correlations contained in $\Psi_{N,t}$, as we now explain. Heuristically, we can think of $\Psi_{N,t}$ as having the form
\begin{equation*}
	\Psi_{N,t}  \approx \prod_{1 \leq i<j<k\leq N}f_{\lambda,N}(x_i - x_j,x_i - x_k)\prod_{i = 1}^N\varphi_t(x_i),
\end{equation*}
where $f_{\lambda,N}$ describes three-body correlations up to some distance $\lambda$. More precisely, the function $f_{\lambda,N}$ is chosen such that $f_{\lambda,N}(\bx) = f(N^{1/2}\bx)$ if $\vert \bx\vert \leq \lambda/2$ with $f$ solving \eqref{eq:zero_energy_scattering_equation}, and $f_{\lambda,N}(\bx) = 1$ if $\vert\bx\vert \geq \lambda$. See \eqref{eq:scattering_solution_def} for a precise definition. Then, basic computations show that, due to the short-range correlations $f_{\lambda,N}$, the orthogonal excitations of such states carry a large energy of order $\mathcal{O}(N)$, thus preventing us from controlling $\cN_+$ using a Grönwall argument. Hence, we would like to factor out these correlations to ensure that the state stays close to the uncorrelated state $\varphi_t^{\otimes N}$. In practice, we remove these correlations using an appropriate Bogoliubov transformation.

\subsection{Removing correlations} 

\label{subsec:removing_correlations}

Define the cubic kernel $k_t$ by
\begin{equation}
	\label{eq:cubic_kernel_def}
	k_t(x,y,z) = N^{3/2}\omega_{\lambda,N}(x,y,z)\varphi_t(x)\varphi_t(y)\varphi_t(z) \quad \textmd{with} \quad \omega_{\lambda,N} = 1 - f_{\lambda,N}.
\end{equation}
The three-body symmetry \eqref{eq:three_body_symmetry_properties} ensures that $k_t$ is symmetric. See Section~\ref{subsec:scattering_energy} for precise definitions and important properties of $f_{\lambda,N}$ and $\omega_{\lambda,N}$.

To figure out how to remove the correlations from the wavefunction $\Psi_{N,t}$, we approximate it by
\begin{align*}
	\Psi_{N,t} &\approx \prod_{1 \leq i< j< k\leq N}(1 - (1 - f_{\lambda,N}))(x_i,x_j,x_k)\prod_{i = 1}^N\varphi_t(x_i)\\
	&\approx \Big(1 - \sum_{1 \leq i< j<k\leq N}\omega_{\lambda,N}(x_i,x_j,x_k) + \dots\Big)\prod_{i = 1}^N\varphi_t(x_i)\\
	&\approx \exp\Big(-\dfrac{1}{6}\int_{\R^9}\d{}x\d{}y\d{}z\omega_{\lambda,N}(x,y,z)a_x^*a_y^*a_z^*a_xa_ya_z\Big)\prod_{i = 1}^N\varphi_t(x_i).
\end{align*}
The analysis of \cite{Brooks25,NamRicTri23} shows that, at the leading order, the correlation energy originates from collisions in which three condensate particles become excited (and vice versa). This indicates that the correlation structure can be incorporated in the product state $\varphi^{\otimes N}$ by applying the cubic Bogoliubov transformation
\begin{equation*}
	\exp\Big(-\dfrac{1}{6}\int_{\R^9}k_t(x,y,z)b_x^*b_y^*b_z^*\dfrac{a(\varphi_t)}{\sqrt{N}}\dfrac{a(\varphi_t)}{\sqrt{N}}\dfrac{a(\varphi_t)}{\sqrt{N}} - \hc\Big).
\end{equation*}
Let us rephrase this in terms of the fluctuation vector $\Phi_{N,t}$. Conjugating the operator inside the previous exponential by the excitation map $U_N$ and using the relations \eqref{eq:excitation_map_action_creation_annihilation_op} gives rise to the cubic operator $K_t:\cF_+^{\leq N}(t) \rightarrow\cF_+^{\leq N}(t)$ given by
\begin{equation*}
	K_t = -\dfrac{1}{6}\int_{\R^9}\big(k_t(x,y,z)b_x^*b_y^*b_z^*\Theta_N - \hc\big).
\end{equation*}
with $\Theta_N$ given by
\begin{equation}
	\label{eq:theta_coeff_cubic_transform}
	\Theta_N = \sqrt{1 - \dfrac{\cN_+}{N}}\sqrt{1 - \dfrac{\cN_+ + 1}{N}}\sqrt{1 - \dfrac{\cN_+ + 2}{N}}.
\end{equation}
Hence, the above heuristics translate to
\begin{equation}
	\label{eq:fluctuation_vector_approx_bog_vacuum}
	\Phi_{N,t} \approx e^{K_t}\Omega
\end{equation}
where $\Omega = (1,0,0,\dots)$ is the vacuum state on $\cF_+$.

The approximation \eqref{eq:fluctuation_vector_approx_bog_vacuum} suggests that we can remove correlations from $\Phi_{N,t}$ by applying the unitary transformation $e^{-K_t}$. In other words, we expect that the fluctuation vector $e^{-K_t}\Phi_{N,t}$ behaves essentially like the vacuum state $\Omega$, and that we might therefore be able to control its number of excitations. For technical reasons, however, we need to control not only the number of excitations but also the energy fluctuations of $e^{-K_t}\Phi_{N,t}$, which we describe next. The fluctuation dynamics of $e^{-K_t}\Phi_{N,t}$ is governed by
\begin{equation*}
	\ui\partial_t\left(e^{-K_t}\Phi_{N,t}\right) = \left(e^{-K_t}\cH_Ne^{K_t} + (\ui\partial_t e^{-K_t})e^{K_t}\right)e^{-K_t}\Phi_{N,t}.
\end{equation*}
Moreover, some heuristic computations show that, on states with few excitations,
\begin{equation}
	\label{eq:hamiltonian_main_contribution_extracted}
	e^{-K_t}\cH_Ne^{K_t} \approx \cK + \cV_N + N\cE^\GP[\varphi_t] - N\langle\varphi_t,\ui\partial_t\varphi_t\rangle,
\end{equation}
where $\cK$ and $\cV_N$ are the second quantised operators defined in \eqref{eq:second_quantised_kinetic_and_interaction_mod}, which both vanish when tested on $\Omega$. Hence, in view of \eqref{eq:fluctuation_vector_approx_bog_vacuum}, it is natural to try to control the fluctuations of the energy
\begin{equation*}
	\langle e^{-K_t}\cH_Ne^{K_t} + (\ui\partial_t e^{-K_t})e^{K_t}\rangle_{e^{-K_t}\Phi_{N,t}} - N\cE^\GP[\varphi_t] + N\langle\varphi_t,\ui\partial_t\varphi_t\rangle,
\end{equation*}
which simplifies to $\langle\cH_N + e^{K_t}(\ui\partial_t e^{-K_t})\rangle_{\Phi_{N,t}}$ due to \eqref{eq:gross_pitaevskii_drift_term}. Said differently, we would like to prove the Grönwall bound
\begin{equation*}
	\partial_t\langle\cH_N + e^{K_t}(\ui\partial_t e^{-K_t}) + e^{K_t}\cN_+e^{-K_t}\rangle_{\Phi_{N,t}} \lesssim \langle\cH_N + e^{K_t}(\ui\partial_t e^{-K_t}) + e^{K_t}\cN_+e^{-K_t}\rangle_{\Phi_{N,t}}.
\end{equation*}

\subsection{Approximating the cubic transformation.}

One of the main difficulties we face when trying to implement the above strategy is that the action of $e^{-K_t}(\cdot)e^{K_t}$ is not explicit and tedious to compute. The standard approach is to expand it using Duhamel's formula sufficiently many times so that the error becomes controllable in terms of $\cN_+, \cK$ and $\cV_N$, which usually requires computing many commutators. To solve this problem, we use ideas from \cite{BrenneckKro24,BrenneckBroCarOld24,Brooks24} to extract only the relevant terms without needing to use operator exponential expansions.

More precisely, to compute $e^{K_t}\cN_+e^{-K_t}$, we apply the Duhamel formula twice to get
\begin{equation*}
	e^{K_t}\cN_+e^{-K_t} = \cN_+ + [K_t,\cN_+] + \int_0^1\d{}s\int_0^s\d{}ue^{uK_t}[K_t,[K_t,\cN_+]]e^{-uK_t}.
\end{equation*}
Then, we use the CCR \eqref{eq:creation_annihilation_op_modified_ccr} and the properties of the cubic kernel $k_t$ to find
\begin{equation*}
	e^{K_t}\cN_+e^{-K_t} \approx \mathcal N_+ + \frac{1}{2} \int_{\R^{9}} k_t(x,y,z) b^*_x b^*_y b^*_z\Theta_N + \hc
\end{equation*}
up to a small error. Similarly, for $e^{K_t}\left(\ui\partial_te^{-K_t}\right)$ we get
\begin{equation*}
	e^{K_t}\left(\ui\partial_te^{-K_t}\right) = \int_0^1\d{}se^{sK_t}\ui\partial_tK_te^{-sK_t} \approx \ui\partial K_t.
\end{equation*}
This motivates the definition of the renormalised operators $\cN^\ren$ and $\cK^\ren$ by
\begin{equation}
	\label{eq:renormalised_number_operator_def}
	\mathcal N^\ren = \mathcal N_+ + \frac{1}{2} \int_{\R^{9}}k_t(x,y,z) b_x^* b_y^* b_z^*\Theta_N + \hc
\end{equation}
and
\begin{equation}
	\label{eq:renormalised_drift_operator_def}
	\mathcal Q^\ren = \dfrac{1}{6}\int_{\R^{9}} \ui (\partial_tk_t)(x,y,z) b^*_x b^*_y b^*_z\Theta_N + \hc
\end{equation}
Then, instead of proving a Grönwall bound on 
\begin{equation*}
	\langle\cH_N + e^{K_t}(\ui\partial_te^{-K_t}) + e^{K_t}\cN_+e^{-K_t}\rangle_{\Phi_{N,t}},
\end{equation*}
the key idea of \cite{BrenneckKro24} is to prove a Grönwall bound on
\begin{equation*}
	\langle\cH_N + \cQ^\ren + \cN^\ren\rangle_{\Phi_{N,t}}.
\end{equation*}
This is the content of the following proposition. The proof is given in Section~\ref{sec:ana_fluct}, and the subsequent sections are devoted to proving of technical results required for that proof.

\begin{proposition}[Main Grönwall estimate]
	\label{prop:gronwall_bound}
	Let $V$ satisfy Assumption~\ref{ass:potential} and take $\lambda = 2RN^{-1/3}$ in \eqref{eq:cubic_kernel_def}. Let $\Psi_{N,t}$ and $\varphi_t$ be solutions respectively of the Schrödinger equation \eqref{eq:schroedinger_eq_time_dependent} and of the time-dependent Gross--Pitaevskii equation \eqref{eq:gross_pitaevskii_equation} for some normalised initial data $\Psi_{N,0} \in \mfH^N$ and $\varphi_0\in H^6(\R^3)$. Then, there exist some constants $c$ and $C$ depending only on $V$ and $\cE^\GP[\varphi_0]$ such that, for all $t \in \R$,
	\begin{equation}
		\label{eq:number_excitations_bound}
		\cN_+ \leq \cH_N + \cQ^\ren + C(\cN^\ren + N^{1/2})
	\end{equation}
	on $\cF_+^{\leq N}(t)$, and
	\begin{equation}
		\label{eq:gronwall_estimate_main}
		\partial_t\langle\cH_N + \cQ^\ren + C(\cN^\ren + N^{1/2})\rangle_{\Phi_{N,t}} \leq c\langle\cH_N + \cQ^\ren + C(\cN^\ren + N^{1/2})\rangle_{\Phi_{N,t}}.
	\end{equation}
\end{proposition}

Before proving Theorem~\ref{th:main_result}, we also introduce the following result, which relates the renormalised quantities $\cN^\ren$ and $\cQ^\ren$ to $\cN_+$.

\begin{lemma}[Properties of $\mathcal N^\ren$ and $\mathcal Q^\ren$] \label{lemma:number_operator_renormalised}
	Assume the same hypotheses as in Proposition~\ref{prop:gronwall_bound}. Define $\cN^\ren$ and $\cQ^\ren$ as in \eqref{eq:renormalised_number_operator_def} and \eqref{eq:renormalised_drift_operator_def}, and suppose that the parameter $\lambda$ in \eqref{eq:cubic_kernel_def} is sufficiently small. Then there exists a constant $C > 0$ depending only on $V$ and $\Vert\varphi_0\Vert_{H^6}$ such that, for all $t\in\R$,
	\begin{equation}
		\label{eq:number_operator_renormalised_estimate}
		C^{-1}(\cN_+ + 1) \leq \cN^\ren \leq C(\cN_+ + 1) \quad \textmd{and} \quad \pm\cQ^\ren \leq C(\cN^\ren + 1)
	\end{equation}
	on $\cF_+^{\leq N}(t)$.
\end{lemma}

\subsection{Proof of Theorem~\ref{th:main_result}}

Using Proposition~\ref{prop:gronwall_bound} and the Grönwall lemma, we obtain
\begin{equation*}
	\langle\cN_+\rangle_{\Phi_{N,t}} \leq ce^{c\vert t\vert}\langle\left(\cH_N + \cQ^\ren + C\cN^\ren\right)_{\vert t=0}\rangle_{\Phi_N(0)} + Ce^{c\vert t\vert}N^{1/2}.
\end{equation*}
On the one hand, thanks to Lemma~\ref{lemma:number_operator_renormalised} and the estimate
\begin{equation*}
	1 - \langle\varphi_0,\gamma_N^{(1)}\varphi_0\rangle \leq \varepsilon_{1,N},
\end{equation*}
we have
\begin{equation*}
	\langle\cN^\ren_{\vert t = 0}\rangle_{\Phi_N(0)} \leq C(1 + N\varepsilon_{1,N}) \quad \textmd{and} \quad \langle\cQ^\ren_{\vert t = 0}\rangle_{\Phi_N(0)} \leq C(1 + N\varepsilon_{1,N}).
\end{equation*}
On the other hand, using the identities
\begin{equation*}
	\langle(\cH_N)_{\vert t = 0}\rangle_{\Phi_N(0)} = \langle H_N\rangle_{\Psi_N(0)} + \langle((\ui\partial_t U_N)U_N^*)_{\vert t  = 0}\rangle_{\Phi_N(0)}
\end{equation*}
and
\begin{equation*}
	(\ui\partial_t U_N)U_N^* = a^*(\varphi_t)b(\ui\partial_t \varphi_t) - \big(\sqrt{N-\mathcal N_+} b(\ui\partial_t \varphi_t) + \hc\big) - \braket{\varphi_t,\ui\partial_t\varphi_t} (N-\mathcal N_+),
\end{equation*}
as well as \eqref{eq:number_excitation_op_rewritten_rdm} and $a(\varphi_t)\Phi_{N,t} = 0$, we deduce
\begin{multline*}
	\langle(\cH_N)_{\vert t = 0}\rangle_{\Phi_N(0)} = \langle H_N\rangle_{\Psi_N(0)} - N\langle\varphi_t,\ui\partial_t\varphi_t\rangle_{\vert t = 0} + N\langle\varphi_t,\ui\partial_t\varphi_t\rangle_{\vert t = 0}(1 - \langle\varphi_0,\gamma_N^{(1)}\varphi_0)\rangle\\
	-2N\Re\langle\varphi_0,\gamma_N^{(1)}Q_0(\ui\partial_t\varphi_t)_{\vert t = 0}\rangle.
\end{multline*}
As a result of \eqref{eq:gross_pitaevskii_drift_term}, the sum of the first two terms can be bounded by
\begin{equation*}
	\langle H_N\rangle_{\Psi_N(0)} - N\langle\varphi_t,\ui\partial_t\varphi_t\rangle_{\vert t = 0} = \langle H_N\rangle_{\Psi_N(0)} - \cE^\GP[\varphi_0] \leq N\varepsilon_{2,N},
\end{equation*}
and the third by
\begin{equation*}
	N\langle\varphi_t,\ui\partial_t\varphi_t\rangle_{\vert t = 0}(1 - \langle\varphi_0,\gamma_N^{(1)}\varphi_0\rangle) \leq CN\varepsilon_{1,N}.
\end{equation*}
Moreover, the fourth term satisfies
\begin{align*}
	\langle\varphi_0,\gamma_N^{(1)}Q_0(\ui\partial_t\varphi_t)_{\vert t = 0}\rangle &= \langle\varphi_0,\gamma_N^{(1)}(\ui\partial_t\varphi_t)_{\vert t = 0}\rangle - \langle\varphi_0,\gamma_N^{(1)}\varphi_0\rangle\langle\varphi_0,(\ui\partial_t\varphi_t)_{\vert t = 0}\rangle\\
	&= \langle\varphi_0,(\gamma_N^{(1)} - \vert\varphi_0\rangle\langle\varphi_0\vert)(\ui\partial_t\varphi_t)_{\vert t = 0}\rangle + (1 - \langle\varphi_0,\gamma_N^{(1)}\varphi_0\rangle)\langle\varphi_0,(\ui\partial_t\varphi_t)_{\vert t = 0}\rangle,
\end{align*}
which implies
\begin{equation*}
	\left\vert\langle\varphi_0,\gamma_N^{(1)}Q_0(\ui\partial_t\varphi_t)_{\vert t = 0}\rangle\right\vert \leq C\varepsilon_{1,N}.
\end{equation*}
Gathering the previous estimates yields
\begin{equation*}
	1 - \langle\varphi_t,\gamma_{N,t}^{(1)}\varphi_t\rangle \leq Ce^{c|t|}(\varepsilon_{1,N} + \varepsilon_{2,N} + N^{-1/2}).
\end{equation*}
This concludes the proof of Theorem~\ref{th:main_result}.

\section{Analysis of fluctuations}
\label{sec:ana_fluct}

This section is devoted to the proof of Proposition~\ref{prop:gronwall_bound} and Lemma~\ref{lemma:number_operator_renormalised}.

\subsection{Three-body scattering solution}

\label{subsec:scattering_energy}

Consider the three-body scattering hypervolume $b(V)$ defined in \eqref{eq:three_body_scattering_hypervolume}. It is convenient to rewrite the zero-energy scattering equation \eqref{eq:zero_energy_scattering_equation} in terms of $\omega = 1 - f$. Define the modified Laplacian
\begin{equation}
	\label{eq:modified_laplacian_def}
	-\Delta_\mathcal{M} = |\mathcal{M}\nabla_{\R^6}|^2 = -\diver\left(\mathcal{M}^2\nabla_{\R^6}\right),
\end{equation}
where the matrix $\cM:\R^3\times\R^3\rightarrow\R^3\times\R^3$ s given by
\begin{equation}
	\label{eq:modified_laplacian_matrix}
	\cM =
	\left(
	\dfrac{1}{2}
	\begin{pmatrix}
		2 & 1\\
		1 & 2
	\end{pmatrix}
	\right)^{1/2}
	=
	\dfrac{1}{2\sqrt{2}}
	\begin{pmatrix}
		\sqrt{3} + 1 & \sqrt{3} - 1\\
		\sqrt{3} - 1 & \sqrt{3} + 1
	\end{pmatrix}.
\end{equation}
Then, by changing the absolute coordinates $(x,y,z)$ to the relative coordinates $(x-y,x-z,(x+y+z)/3)$ and getting rid of the centre of mass $(x+y+z)/3$, we find that $\omega$ solves
\begin{equation}
	\label{eq:three_body_scattering_equation}
	-2\Delta_\cM\omega + V(\omega - 1) = 0
\end{equation}
on $\R^6$.

It was shown in \cite[Theorem 8]{NamRicTri23} that the function $\omega$ is unique, satisfies the three-body symmetry \eqref{eq:three_body_symmetry_properties}, and obeys the pointwise estimates
\begin{equation}
	\label{eq:three_body_scattering_solution_pointwise_estimate}
	0 \leq \omega(\bx)\leq 1, \quad \omega(\bx) \leq \dfrac{C}{1 + |\bx|^4} \quad \textmd{and} \quad |\nabla\omega(\bx)|\leq \dfrac{C}{1 + |\bx|^5},
\end{equation}
for all $\bx\in\R^6$, and for some constant $C$ that depends only on $V$.

For technical reasons, we need to truncate the solution $\omega$ of \eqref{eq:three_body_scattering_equation}. Let $\tilde{\chi}$ be a smooth function on $\R^3$ that satisfies $0 \leq \tilde{\chi} \leq 1$ with $\tilde{\chi}(x) = 1$ if $\vert x\vert \leq 1/4$ and $\tilde{\chi}(x)$ if $\vert x\vert \geq 1$. Define
\begin{equation*}
	\chi(x,y) = (\tilde{\chi}(x)\tilde{\chi}(y) + \tilde{\chi}(-x)\tilde{\chi}(y - x) + \tilde{\chi}(-y)\tilde{\chi}(x - y))/3.
\end{equation*}
Then, $\chi$ satisfies the three-body symmetry properties \eqref{eq:three_body_symmetry_properties} and is such that $0 \leq \chi\leq 1$ with $\chi(\bx) = 1$ if $\vert \bx\vert \leq 1/2$ and $\chi(\bx) = 0$ if $\vert \bx\vert \geq 1$. Define also
\begin{equation}
	\label{eq:scattering_solution_def}
	\begin{array}{c}
		\omega_N = \omega(N^{1/2}\cdot), \quad f_N = 1 - \omega_N, \quad \chi_\lambda = \chi(\lambda^{-1}\cdot),\\
		\omega_{\lambda,N} = \chi_\lambda\omega_N \quad \textmd{and} \quad f_{\lambda,N} = 1 - \omega_{\lambda,N}.
	\end{array}
\end{equation}
Then, for $\lambda \geq 2RN^{-1/2}$, the truncated scattering solution $\omega_{\lambda,N}$ solves, for almost all $x,y,z\in\R^3$, the \textit{modified zero-energy scattering equation}
\begin{equation}
	\label{eq:truncated_scattering_equation_no_modified_laplacian}
	\left(-\Delta_x - \Delta_y - \Delta_z\right)\omega_{\lambda,N}(x,y,z) = (V_Nf_N)(x,y,z) - N^{-2}\varepsilon_\lambda(x,y,z),
\end{equation}
with $\varepsilon_\lambda$ given by
\begin{equation}
	\label{eq:truncated_scattering_equation_error}
	\varepsilon_\lambda = 4N^{2}\mathcal{M}\nabla\omega_N\cdot\mathcal{M}\nabla\chi_\lambda + 2N^{2} \omega_N\Delta_\mathcal{M}\chi_\lambda.
\end{equation}
Lastly, again for technical reasons, we define
\begin{equation}
	\label{eq:truncated_scattering_equation_error_green_function}
	u_\lambda = (-2\Delta_\cM)^{-1}\varepsilon_\lambda.
\end{equation}
By definition of the modified Laplacian $-\Delta_\cM$, the potential $u_\lambda$ solves, for almost all $x,y,z\in\R^3$, the equation
\begin{equation}
	\label{eq:truncated_scattering_equation_error_green_function_solves}
	(-\Delta_x -\Delta_y - \Delta_z)u_\lambda(x,y,z) = \varepsilon_\lambda(x,y,z).
\end{equation}

\begin{lemma}[Properties of $\omega_{\lambda,N}, \varepsilon_\lambda$ and $u_\lambda$]
	\label{lemma:truncated_three_body_scattering_solution}
	Let $V$ satisfy Assumption~\ref{ass:potential}. Define $\omega_{\lambda,N},\varepsilon_\lambda$ and $u_\lambda$ as in \eqref{eq:scattering_solution_def}, \eqref{eq:truncated_scattering_equation_error} and \eqref{eq:truncated_scattering_equation_error_green_function}. Then, there exists a constant $C$ depending only on $V$ such that, for all $\bx\in\R^6$,
	\begin{equation}
		\label{eq:truncated_scattering_equation_pointwise_estimate}
		0\leq \omega_{\lambda,N}(\bx) \leq C\dfrac{\mathds{1}(\vert\bx\vert \leq \lambda)}{1 + |N^{1/2}\bx|^4}, \quad |\nabla\omega_{\lambda,N}(\bx)| \leq CN^{1/2}\dfrac{\mathds{1}(|\bx| \leq \lambda)}{1 + |N^{1/2}\bx|^5},
	\end{equation}
	\begin{equation}
		\label{eq:truncated_scattering_equation_error_pointwise_estimate}
		|\varepsilon_\lambda(\bx)| \leq C\lambda^{-2}\dfrac{\mathds{1}(\lambda/2\leq|\bx|\leq\lambda)}{|\bx|^4},
	\end{equation}
	\begin{equation}
		\label{eq:truncated_scattering_equation_error_green_function_pointwise_estimate}
		\vert u_\lambda(\bx)\vert \leq \dfrac{C}{\lambda^4 + \vert\bx\vert^4} \quad \textmd{and} \quad \vert\nabla u_\lambda(\bx)\vert \leq \dfrac{C}{\lambda^5 + \vert\bx\vert^5}.
	\end{equation}
	Consequently, the following estimates hold:
	\begin{equation}
		\label{eq:truncated_scattering_equation_Lp_estimates}
		\begin{array}{c}
			\Vert\omega_{\lambda,N}\Vert_{L^1} \leq C\lambda^2N^{-2}, \quad \Vert\omega_{\lambda,N}\Vert_{L^2} \leq CN^{-3/2},\\
			\Vert\nabla\omega_{\lambda,N}\Vert_{L^1} \leq C\lambda N^{-2}, \quad \Vert\nabla\omega_{\lambda,N}\Vert_{L^2} \leq CN^{-1},
		\end{array}
	\end{equation}
	\begin{equation}
		\label{eq:truncated_scattering_equation_LpLinf_estimates}
		\begin{array}{c}
			\Vert\omega_{\lambda,N}\Vert_{L^1L^\infty} \leq CN^{-3/2}, \quad \Vert\omega_{\lambda,N}\Vert_{L^{3/2}L^\infty} \leq CN^{-1},\\
			\Vert\omega_{\lambda,N}\Vert_{L^2L^\infty} \leq CN^{-3/4}, \quad \Vert\nabla\omega_{\lambda,N}\Vert_{L^1L^\ii} \leq CN^{-1},
		\end{array}
	\end{equation}
	\begin{equation}
		\label{eq:truncated_scattering_equation_error_Lp_estimates}
		\Vert\varepsilon_\lambda\Vert_{L^2} \leq C\lambda^{-3} \quad \textmd{and} \quad \Vert \nabla u_\lambda\Vert_{L^2} \leq C\lambda^{-2},
	\end{equation}
\end{lemma}

\begin{proof}[Proof of Lemma~\ref{lemma:truncated_three_body_scattering_solution}]
	\sloppy That $\omega_{\lambda,N}$ solves \eqref{eq:truncated_scattering_equation_no_modified_laplacian} follows directly from the zero-energy scattering equation \eqref{eq:three_body_scattering_equation}, the definition of the modified Laplacian $-\Delta_\cM$ \eqref{eq:modified_laplacian_def}, and the fact that $\chi_\lambda = 1$ on $\Supp V(N^{1/2}\cdot)$. Moreover, the pointwise estimates \eqref{eq:truncated_scattering_equation_pointwise_estimate} result directly from $\Supp\chi_\lambda\subset B(0,\lambda)$, $|\nabla\chi_\lambda|\leq C\lambda^{-1}\mathds{1}(\lambda/2\leq |\bx|\leq\lambda)$ and \eqref{eq:three_body_scattering_solution_pointwise_estimate}. Similarly, the estimate \eqref{eq:truncated_scattering_equation_error_pointwise_estimate} results from the bounds $|\nabla\chi_\lambda|\leq C\lambda^{-1}\mathds{1}(\lambda/2\leq |\bx|\leq\lambda)$, $|\Delta\chi_\lambda|\leq C\lambda^{-2}\mathds{1}(\lambda/2\leq |\bx|\leq\lambda)$ and \eqref{eq:three_body_scattering_solution_pointwise_estimate}.
	
	We now prove the estimates \eqref{eq:truncated_scattering_equation_Lp_estimates}~and~\eqref{eq:truncated_scattering_equation_LpLinf_estimates}. The pointwise bound \eqref{eq:truncated_scattering_equation_pointwise_estimate} implies
	\begin{equation*}
		\Vert\omega_{\lambda,N}\Vert_{L^1} \leq CN^{-2}\int_{\R^6}\dfrac{\mathds{1}(\bx|\leq\lambda)}{|\bx|^4}\d{}\bx \leq C\lambda^2N^{-2},
	\end{equation*}
	which is precisely the first bound in \eqref{eq:truncated_scattering_equation_Lp_estimates}. The other three estimates in \eqref{eq:truncated_scattering_equation_Lp_estimates} follow similarly using additionally \eqref{eq:truncated_scattering_equation_error_pointwise_estimate}. Furthermore, the bound \eqref{eq:truncated_scattering_equation_pointwise_estimate} implies
	\begin{equation*}
		\int_{\R^3}\d{}y\omega_{\lambda,N}(x,y) \leq C\int_{\R^3}\dfrac{\d{}y}{1 + \vert N^{1/2}y\vert^4} \leq CN^{-3/2},
	\end{equation*}
	for all $x \in \R^3$, which yields the first bound in \eqref{eq:truncated_scattering_equation_LpLinf_estimates}. The remaining two bounds are proven analogously.
	
	It follows from the analysis of \cite[Section 2.2.]{NamRicTri23} that $u_\lambda$ can be written explicitly as
	\begin{equation*}
		u_\lambda(\bx) = \dfrac{\det\cM^{-1}}{8\vert\mathbb{S}^5\vert}\int_{\R^6}\d{}\by\dfrac{\varepsilon_\lambda(\by)}{\vert\cM^{-1}(\bx - \by)\vert^4},
	\end{equation*}
	where $\vert\mathbb{S}^5\vert$ denotes the surface area of the unit sphere in dimension $6$. Thus, $\nabla u_\lambda$ is given by
	\begin{equation*}
		\nabla u_\lambda(\bx) = -\dfrac{\det\cM^{-1}}{2\vert\mathbb{S}^5\vert}\int_{\R^6}\d{}\by\dfrac{\varepsilon_\lambda(\by)\cM^{-1}(\bx - \by)}{\vert\cM^{-1}(\bx - \by)\vert^6}.
	\end{equation*}
	The estimates \eqref{eq:truncated_scattering_equation_error_green_function_pointwise_estimate} follow from \eqref{eq:truncated_scattering_equation_error_pointwise_estimate}.
\end{proof}

\subsection{Analysis of $\cN^\ren$ and $\cQ^\ren$: proof of Lemma~\ref{lemma:number_operator_renormalised}}

By definition of $\cN^\ren$, we can write
\begin{equation*}
	\cN^\ren - \cN_+ = \dfrac{1}{2}\int_{\R^3}b_x^*\Big(\int_{\R^6}k_t(x,y,z)b_y^*b_z^*\Big)\Theta_N + \hc
\end{equation*}
Then, we use the Cauchy--Schwarz inequality, the estimate $0 \leq \Theta_N \leq 1$, the fact that $\Vert Q_t\Vert_\op =1$ and the relations \eqref{eq:creation_annihilation_op_modified_ccr} and \eqref{eq:number_excitation_op_int_form}, to deduce
\begin{align*}
	\pm(\cN^\ren - \cN_+) &\leq \varepsilon C\int_{\R^3}b_x^*b_x + \varepsilon^{-1}C\int_{\R^3}\Big(\int_{\R^6}\overline{k_t(x,y,z)}b_yb_z\Big)\Big(\int_{\R^6}k_t(x,y',z')b_{y'}^*b_{z'}^*\Big)\\
	&\leq
	\begin{multlined}[t]
		\varepsilon C\cN_+  + \varepsilon^{-1}C\Vert k_t\Vert_{L^2(\R^9)} + \varepsilon^{-1}C\int_{\R^{12}}\overline{k_t(x,y,z)}k_t(x,y,z')b_{z'}^*b_z\\
		+ \varepsilon^{-1}C\int_{\R^{15}}\overline{k_t(x,y,z)}k_t(x,y',z')b_{y'}^*b_{z'}^*b_yb_z,
	\end{multlined}
\end{align*}
for all $\varepsilon > 0$. Thanks to Hölder's inequality, and the estimates \eqref{eq:gross_pitaevskii_solution_estimate_H1_L6} and \eqref{eq:truncated_scattering_equation_Lp_estimates}, the second term can be bounded by
\begin{equation}
	\label{eq:cubic_kernel_estimate_L2}
	\Vert k_t\Vert_{L^2} \leq N^{3/2}\Vert \omega_{\lambda,N}\Vert_{L^2}\Vert\varphi_t\Vert_{L^6}^6 \leq C.
\end{equation}
Using Lemma~\ref{lemma:three_body_kernel_second_quantisation_estimates} to bound the remaining two terms, we obtain
\begin{equation*}
	\pm(\cN^\ren - \cN_+) \leq \varepsilon \cN_+ + \varepsilon^{-1}C(1 + \lambda\cN_+),
\end{equation*}
for all $\varepsilon > 0$. For $\lambda$ small enough, this yields the estimate on $\cN^\ren$ in \eqref{eq:number_operator_renormalised_estimate}. The bound for $\cQ^\ren$ follows similarly. This concludes the proof of Lemma~\ref{lemma:number_operator_renormalised}.

\subsection{Analysis of the generator of the fluctuations}

Conjugating \eqref{eq:hamiltonian_main_contribution_extracted} by $e^{K_t}(\cdot)e^{-K_t}$ and using \eqref{eq:gross_pitaevskii_drift_term}, we find
\begin{equation*}
	\cH_N \approx e^{K_t}\cK e^{-K_t} + e^{K_t}\cV_N e^{-K_t}.
\end{equation*}
This roughly says that $\cH_N$ is equal to the second quantisation of the Hamiltonian $H_N$, from which we removed the leading-order contributions. Just as for $\cN^\ren$ and $\cQ^\ren$, we use ideas from \cite{BrenneckKro24,BrenneckBroCarOld24,Brooks24} to extract the relevant terms from $e^{K_t}\cK e^{-K_t}$ and $e^{K_t}\cV_N e^{-K_t}$ in order to avoid using operator exponential expansions. Namely, we write
\begin{equation*}
	e^{K_t}\cK e^{-K_t} = \int_{\R^3}\d{}xe^{K_t}b_x^*e^{-K_t}(-\Delta_x)e^{K_t}b_xe^{-K_t},
\end{equation*}
and approximate $e^{K_t}b_xe^{-K_t}$ by
\begin{equation*}
	e^{K_t}b_xe^{-K_t} \approx b_x + [K_t,b_x] \approx b_x + \dfrac{1}{2}\int_{\R^6}\d{}y\d{}zk_t(x,y,z)b_y^*b_z^*,
\end{equation*}
and likewise for $b_x^*$. For $e^{K_t}\cV_Ne^{-K_t}$, we approximate $e^{K_t}b_xb_yb_ze^{-K_t}$ by
\begin{equation*}
	e^{K_t}b_xb_yb_ze^{-K_t} \approx b_xb_yb_z + [K_t,b_xb_yb_z] \approx b_xb_yb_z + k_t(x,y,z).
\end{equation*}
Rigorously, this corresponds to defining the renormalised annihilation operators
\begin{equation}
	\label{eq:renormalised_annihilation_op_one}
	c_x = b_x + \dfrac{1}{2}\int_{\R^6}\d{}y\d{}zk_t(x,y,z)b_y^*b_z^*
\end{equation}
and
\begin{equation}
	\label{eq:renormalised_annihilation_op_three}
	d_{xyz} = b_xb_yb_z + k_t(x,y,z).
\end{equation}
These are associated to the renormalised operators
\begin{equation}
	\label{eq:renormalised_hamiltonian_def}
	\cK^\ren = \int_{\R^3}c_x^*(-\Delta_x)c_x \quad \textmd{and} \quad \cV^\ren = \dfrac{1}{6}\int_{\R^9}V_N(x,y,z)d_{xyz}^*d_{xyz}.
\end{equation}
Notice that $\cK^\ren \geq 0$ and $\cV^\ren \geq 0$.
\begin{lemma}
	\label{lemma:generator_estimate}
	Assume the same hypotheses as in Proposition~\ref{prop:gronwall_bound}. Define $\cH_N$ as in \eqref{eq:generator_fluctuations}, and $\cK^\ren$ and $\cV^\ren$ as in \eqref{eq:renormalised_hamiltonian_def}. Then
	\begin{equation}
		\label{eq:generator_estimate}
		\mathcal H_N = \mathcal K^\ren + \mathcal V^\ren + \mathcal E,
	\end{equation}
	with $\mathcal E$ satisfying
	\begin{equation}
		\label{eq:generator_estimate_error}
		\pm \mathcal E \leq \varepsilon \mathcal{K}^\ren + \varepsilon\mathcal{V}^\ren + C_{\varepsilon}\mathcal{N}^\ren + C_\varepsilon N^{1/2},
	\end{equation}
	for all $\varepsilon>0$. The constant $C_\varepsilon$ depends only on $\varepsilon, V$ and $\Vert\varphi_0\Vert_{H^6}$.
\end{lemma}

\begin{lemma}\label{lemma:generator_commutator_estimates}
	Assume the same hypotheses as in Proposition~\ref{prop:gronwall_bound}. Define $\cH_N$ as in \eqref{eq:generator_fluctuations}, and $\cN^\ren$ and $\cQ^\ren$ as in \eqref{eq:renormalised_number_operator_def}. Then
	\begin{equation}
		\label{eq:generator_commutator_Qren_estimate}
		\pm\partial_t \langle \mathcal H_N + \mathcal Q^\ren \rangle_{\Phi_{N,t}} \leq C\langle \mathcal H_N + \mathcal N^\ren\rangle_{\Phi_{N,t}} + CN^{1/2},
	\end{equation}
	and
	\begin{equation}
		\label{eq:generator_commutator_Nren_estimate}
		\pm \ui[\cN^\ren, \cH_N] \leq C \cH_N + C\cN^\ren + CN^{1/2},
	\end{equation}
	for some constant $C$ that depends only on $V$ and $\Vert\varphi_0\Vert_{H^6}$.
\end{lemma}

The first step in the proofs of both Lemmas~\ref{lemma:generator_estimate}~and~\ref{lemma:generator_commutator_estimates} is to rewrite $\cH_N$, which we recall is given by
\begin{equation*}
	\mathcal{H}_N = (\ui\partial_t U_N)U_N^* + U_NH_NU_N^*.
\end{equation*}
On the one hand, it was shown in \cite{LewinNamSch15} that
\begin{equation}
	\label{eq:excitation_map_drift}
	(\ui\partial_t U_N)U_N^* = a^*(\varphi_t)b(\ui\partial_t \varphi_t) - \big(\sqrt{N-\mathcal N_+} b(\ui\partial_t \varphi_t) + \hc\big) - \braket{\varphi_t,\ui\partial_t\varphi_t} (N-\mathcal N_+).
\end{equation}
On the other hand, a tedious but straightforward computation relying on the relations \eqref{eq:excitation_map_action_creation_annihilation_op} yields
\begin{equation}
	U_N H_N U_N^* = \cK + \cV_N + \big(\sqrt{N-\mathcal N_+} b(-\Delta \varphi_t) + \hc\big) + \Vert\nabla \varphi_t\Vert_2^2 (N - \mathcal N_+) +  \sum_{j=0}^5 \mathcal L_j, \label{eq:hamiltonian_conjugation_excitation_map}
\end{equation}
with $\cK$ and $\cV_N$ defined in \eqref{eq:second_quantised_kinetic_and_interaction_mod}, and $\cL_0,\dots,\cL_5$ given by
\begin{align*}
	\mathcal{L}_0 &= \Theta_0\dfrac{N^3}{6} \int_{\R^9} V_N (x,y,z) |\varphi_t(x)|^2 |\varphi_t(y)|^2 |\varphi_t(z)|^2,\\
	\mathcal L_1 &=  \Theta_1\dfrac{N^{5/2}}{2}\int_{\R^9} V_N(x,y,z)|\varphi_t(y)|^2 |\varphi_t(z)|^2 \overline{\varphi_t(x)}b_x + \hc,\\
	\mathcal{L}_2 &= \mathcal{L}_2^{(1)} + \mathcal{L}_2^{(2)} + \mathcal{L}_2^{(3)},\\
	\mathcal{L}_2^{(1)} &= \Theta_2^{(1)}\dfrac{N^2}{2}\int_{\R^9} V_N(x,y,z)|\varphi_t(z)|^2 \overline{\varphi_t(x) \varphi_t(y)} b_xb_y + \hc,\\
	\mathcal{L}_2^{(2)} &= \Theta_2^{(2)}\dfrac{N^2}{2}\int_{\R^9} V_N(x,y,z) |\varphi_t(y)|^2|\varphi_t(z)|^2 b_x^* b_x, \\
	\mathcal{L}_2^{(3)} &= \Theta_2^{(3)}N^2\int_{\R^9} V_N(x,y,z) \varphi_t(x)\overline{\varphi_t(y)}|\varphi_t(z)|^2 b^*_x b_y, \\
	\mathcal L_3 &= \mathcal L_3^{(1)} + \mathcal L_3^{(2)} + \mathcal L_3^{(3)},\\
	\mathcal L_3^{(1)} &= \Theta_3^{(1)}\dfrac{N^{3/2}}{6}\int_{\R^9} V_N (x,y,z) \overline{\varphi_t(x)\varphi_t(y) \varphi_t(z)} b_xb_yb_z + \hc, \\
	\mathcal L_3^{(2)} &= \Theta_3^{(2)}N^{3/2}\int_{\R^9} V_N(x,y,z) |\varphi_t(z)|^2 \overline{\varphi_t(y)} b^*_x b_x b_y + \hc, \\
	\mathcal L_3^{(3)} &= \Theta_{3}^{(3)}\dfrac{N^{3/2}}{2}\int_{\R^9}  V_N(x,y,z) \varphi_t(x)\overline{\varphi_t(y)}\overline{\varphi_t(z)} b^*_xb_y b_z + \hc,\\
	\mathcal L_4 &= \mathcal L_4^{(1)} + \mathcal L_4^{(2)} + \mathcal L_4^{(3)},\\
	\mathcal L_4^{(1)} &= \Theta_4^{(1)}\frac{N}{2}\int_{\R^9} V_N(x,y,z) \overline{\varphi_t(y) \varphi_t(z)} b^*_xb_xb_yb_z + \hc, \\
	\mathcal L_4^{(2)} &= \Theta_4^{(2)}\frac{N}{2}\int_{\R^9} V_N(x,y,z) |\varphi_t(z)|^2 b^*_xb^*_y b_xb_y,\\
	\mathcal L_4^{(3)} &= \Theta_4^{(3)}N\int_{\R^9} V_N(x,y,z) \varphi_t(y)\overline{\varphi_t(z)} b^*_xb^*_y b_xb_z, \\
	\mathcal L_5 &= \Theta_5\dfrac{N^{1/2}}{2}\int_{\R^9} V_N(x,y,z) \overline{\varphi_t(z)} b^*_xb^*_yb_xb_yb_z + \hc
\end{align*}
We refer to \cite[Appendix B]{LewinNamSch15} for more detail on the derivation of \eqref{eq:hamiltonian_conjugation_excitation_map}. Though the coefficients $\Theta_j$ and $\Theta_j^{(i)}$ in the previous expressions are functions of $\cN_+$, we omit this dependence from the notation for readability. Moreover, these coefficients can be determined explicitly, but since the exact expressions do not play a role in the analysis, we give them only in Appendix~\ref{app:theta_coefficients}. The only relevant properties of these coefficients is that they all satisfy
\begin{equation}
	\label{eq:theta_coef_estimate_almost_1}
	\vert\Theta_j^{(i)} - 1\vert \leq C(\cN_+ + 1)/N,
\end{equation}
and
\begin{equation}
	\label{eq:theta_coef_estimate_difference}
	\Theta_N\vert\Theta_j^{(i)}(\cN_+ + p) - \Theta_j^{(i)}(\cN_+)\vert \leq C_p/N.
\end{equation}
for any nonnegative integer $p$ (and likewise for $\Theta_0, \Theta_1$ and $\Theta_5$). Recall that $\Theta_N$ was defined in \eqref{eq:theta_coeff_cubic_transform}. Moreover, the coefficient $\Theta_N$ is such that
\begin{equation}
	\label{eq:theta_coef_cubic_transform_almost_1}
	\vert\Theta_N - 1\vert \leq C(\cN_+ + 1)/N
\end{equation}
and
\begin{equation}
	\label{eq:theta_coef_cubic_transform_difference}
	\vert\Theta_N(\cN_+ + p) - \Theta_N(\cN_+)\vert \leq C_p/N,
\end{equation}
for any nonnegative integer $p$.

Before giving the proofs of Lemmas~\ref{lemma:generator_estimate}~and~\ref{lemma:generator_commutator_estimates}, we say a few words about \eqref{eq:hamiltonian_conjugation_excitation_map} and we introduce important estimates. Although the last sum in \eqref{eq:hamiltonian_conjugation_excitation_map} contains many terms, most of term do not play a role at the main order. In fact, the only terms for which we need to identify leading-order cancellations are $\cL_0, \cL_1, \cL_3^{(1)}$ and $\cL_4^{(1)}$. While no cancellations are required to control the remaining terms, they still present a significant challenge. Notably, $\cL_4^{(2)}$ is an effective two-body interaction term whose control in terms of $\cN_+, \cK$ and $\cV_N$ is particularly subtle. Moreover, when evaluating the commutators $\left[\cQ^\ren,\cH_N\right]$ and $\left[\cN^\ren,\cH_N\right]$, some effective three-body and four-body interaction terms appear, which are again nontrivially controlled by $\cN_+, \cK$ and $\cV_N$. Both of these difficulties are specific to the three-body case - there is no counterpart in the two-body setting. Controlling these terms is the subject of the following two propositions, which is one of the key novelties of this work. The proof of Propositions \ref{prop:general_two_body_bound_renormalised} and \ref{prop:general_three_and_four_body_bound_renormalised} occupies the entirety of Sections~\ref{sec:effective_interaction_potentials_non_renormalised}~and~\ref{sec:effective_interaction_potentials_renormalised}.

\begin{proposition}[Estimate on effective two-body potentials]
	\label{prop:general_two_body_bound_renormalised}
	Assume the same hypotheses as in Proposition~\ref{prop:gronwall_bound}. Define $\tilde{V}:x\in\R^3\mapsto\Vert V(x,\cdot)\Vert_{L^1(\R^3)}$ and
	\begin{equation*}
		\cV_\twoB = \int_{\R^6}N^{1/2} \tilde{V}(N^{1/2}(x - y))b^*_xb^*_y b_xb_y.
	\end{equation*}
	Then, for any $\varepsilon > 0$, the estimates 
	\begin{equation}
		\label{eq:general_two_body_bound_non_renormalised}
		\cV_\twoB \leq \varepsilon\cK + \varepsilon\cV_N + C_\varepsilon\cN_+
	\end{equation}
	and
	\begin{equation}
		\label{eq:general_two_body_bound_renormalised}
		\cV_\twoB \leq \varepsilon\cK^\ren + \varepsilon\cV^\ren + C_\varepsilon(\cN^\ren + N^{1/2})
	\end{equation}
	hold on $\cF_+^{\leq N}(t)$ for some constant $C_\varepsilon$ depending only on $\varepsilon$ and $V$, and additionally on $\Vert\varphi_0\Vert_{H^6}$ for \eqref{eq:general_two_body_bound_renormalised}.
\end{proposition}

\begin{proposition}[Further estimates on effective interaction potentials]
	\label{prop:general_three_and_four_body_bound_renormalised}
	Assume the same hypotheses as in Proposition~\ref{prop:gronwall_bound}. Define $\tilde{V}:x\in\R^3\mapsto\Vert V(x,\cdot)\Vert_{L^1(\R^3)}$, as well as
	\begin{equation*}
		\tilde{\cV}_\twoB = \int_{\R^{6}}\mathds{1}(\vert x - y\vert \leq \lambda)b_x^*b_y^*b_xb_y,
	\end{equation*}
	\begin{equation*}
		\cV_\threeB = \int_{\R^{9}}\tilde{V}(N^{1/2}(x - y))\mathds{1}(\vert x - z\vert \leq \lambda)b_x^*b_y^*b_z^*b_xb_yb_z
	\end{equation*}
	and
	\begin{equation*}
		\cV_\fourB = \int_{\R^{12}}N^{-1/2}\tilde{V}(N^{1/2}(x - y))\mathds{1}(\vert x - z\vert \leq \lambda)\mathds{1}(\vert x - u\vert \leq \lambda)b_x^*b_y^*b_z^*b_u^*b_xb_yb_zb_u,
	\end{equation*}
	for any $\lambda$ satisfying
	\begin{equation*}
		0 < \lambda \leq cN^{-1/3},
	\end{equation*}
	for some nonnegative constant $c$. Then, the estimates
	\begin{equation}
		\label{eq:general_two_three_four_body_bound_non_renormalised}
		\tilde{\cV}_\twoB \leq C(\cN_+ + \cK + \cV_N), \quad \cV_\threeB \leq C(\cN_+ + \cK + \cV_N), \quad \cV_\fourB \leq C(\cN_+ + \cK + \cV_N)
	\end{equation}
	and
	\begin{equation}
		\label{eq:general_four_body_bound_renormalised}
		\cV_\fourB \leq C(\cN^\ren + \cK^\ren + \cV^\ren + N^{1/2})
	\end{equation}
	hold on $\cF_+^{\leq N}(t)$ for some constant $C$ depending only on $c, V$, and additionally on $\Vert\varphi_0\Vert_{H^6}$ for \eqref{eq:general_four_body_bound_renormalised}.
\end{proposition}

\begin{remark}
	Though we do not need it, the estimates on $\tilde{\cV}_\twoB$ and $\cV_\threeB$ can be proven with renormalised operators in the right-hand side by following the same proof strategy as for \eqref{eq:general_two_body_bound_renormalised} and \eqref{eq:general_four_body_bound_renormalised}.
\end{remark}

\subsection{Proof of Lemma \ref{lemma:generator_estimate}}

Recall that properties of the solution $\varphi_t$ of the Gross--Pitaevskii equation \eqref{eq:gross_pitaevskii_equation} are given in Appendix~\ref{sec:prop_gp_equation}. We will often use the regularity estimates \eqref{eq:gross_pitaevskii_solution_estimate_H1_L6}, \eqref{eq:gross_pitaevskii_solution_estimate_H4} and \eqref{eq:gross_pitaevskii_solution_estimate_Linf}, and we do not mention this every time to reduce clutter.

Thanks to the relations \eqref{eq:excitation_map_drift}~and~\eqref{eq:hamiltonian_conjugation_excitation_map}, the facts that $\varphi_t$ solves \eqref{eq:gross_pitaevskii_equation} and that $a(\varphi_t) = 0$ on $\cF_+$, we have
\begin{equation*}
	\mathcal{H}_N = \cK + \cV_N - \dfrac{b(V)}{2}\big(\sqrt{N - \cN_+}b(\vert\varphi_t\vert^4\varphi_t) + \hc\big) -\dfrac{b(V)}{6}\Vert\varphi_t\Vert_6^6(N - \cN_+) + \sum_{j = 0}^5\mathcal{L}_j 
\end{equation*}
in the quadratic form sense on $\cF_+^{\leq N}$. Since all the estimates in this section should be interpreted in this sense, we mention this only once. We begin by bounding all the $\cL_j$ terms that do not contain a contribution of order $N$. We claim that
\begin{equation}
	\label{eq:bound_L21}
	\pm\cL_2^{(1)} \leq \varepsilon\cK^\ren + \varepsilon\cV^\ren + C_\varepsilon\cN^\ren + C_\varepsilon N^{1/2},
\end{equation}
for all $\varepsilon > 0$. Since the proof of this inequality is somewhat tedious and unrelated to the rest of the proof, we postpone it to Section~\ref{subsec:quadratic_term_estimate}. In short, this estimate is a consequence of a bound on quadratic Hamiltonians, namely \cite[Theorem 2]{NamNapSol16}, which we conjugate by an appropriate Bogoliubov transformation to obtain a bound with the renormalised operator $\cK^\ren$. Next, we apply the Cauchy--Schwarz inequality to get
\begin{equation*}
	\pm \mathcal{L}_2^{(2)}, \pm \mathcal{L}_2^{(3)} \leq C\mathcal N_+.
\end{equation*}
Define $\cV_\twoB$ as in Proposition~\ref{prop:general_two_body_bound_renormalised}. Then, using the Cauchy--Schwarz inequality, we find
\begin{equation*}
	\pm \mathcal{L}_3^{(2)}, \pm \mathcal{L}_3^{(3)}, \pm \mathcal{L}_4^{(2)}, \pm \mathcal{L}_4^{(3)} \leq C\mathcal V_\twoB + C\mathcal N_+.
\end{equation*}
To bound $\mathcal{L}_5$, we write
\begin{equation*}
	\cL_5 = \dfrac{N^{1/2}}{2}\int_{\R^9}\big(V_N(x,y,z)\varphi_t(z)d_{xyz}^*b_xb_y\Theta_5 - (V_N\overline{k_t})(x,y,z)\varphi_t(z)b_xb_y\Theta_5\big) + \hc,
\end{equation*}
where we recall that $d_{xyz}$ was defined in \eqref{eq:renormalised_annihilation_op_three}. Bounding the first term by $\varepsilon_1\cV^\ren + C_{\varepsilon_1}\cV_\twoB$ and the second one just as in \eqref{eq:bound_L21}, we obtain
\begin{equation*}
	\pm\mathcal{L}_5 \leq \varepsilon\cK^\ren + \varepsilon\cV^\ren + C_\varepsilon\cV_\twoB + C_\varepsilon\cN^\ren + C_\varepsilon N^{1/2},
\end{equation*}
for all $\varepsilon > 0$. The remaining terms we have to bound are $\cK, \cL_0, \cL_1, \cL_3^{(1)}, \cL_4^{(1)}$ and $\cV_N$, for which we need to identify cancellations. We summarise these cancellations in the following two lemmas.
\begin{lemma}
	\label{lemma:generator_estimate_kinetic_cubic_interaction}
	Assume the same hypotheses as in Proposition~\ref{prop:gronwall_bound}. Then,
	\begin{equation}
		\label{eq:generator_estimate_kinetic_cubic_interaction}
		\cL_0 + \cK + \cL_3^{(1)} + \cV_N = \dfrac{N^3}{6}\int_{\R^9}(V_Nf_N)(x,y,z)\vert\varphi_t(x)\vert^2\vert\varphi_t(y)\vert^2\vert\varphi_t(z)\vert^2 + \cK^\ren + \cV^\ren + \cE,
	\end{equation}
	for some $\cE$ that satisfies
	\begin{equation*}
		\pm\cE \leq \varepsilon\cK^\ren + \varepsilon\cV^\ren + C_\varepsilon\cN^\ren + C_\varepsilon N^{-1/2}(\cK + \cV_N + N),
	\end{equation*}
	for all $\varepsilon > 0$.
\end{lemma}

\begin{remark}
	The analysis of $\cK$ and $\cV_N$ can easily be adapted to show the weaker estimate
	\begin{equation}
		\label{eq:kinetic_plus_interaction_bdd_renormalised_op}
		\pm(\cK + \cV_N) \leq C(\cK^\ren + \cV^\ren + N).
	\end{equation}
	Hence, the error in Lemma~\ref{lemma:generator_estimate_kinetic_cubic_interaction} can be simplified to
	\begin{equation*}
		\pm\cE \leq \varepsilon\cK^\ren + \varepsilon\cV^\ren + C_\varepsilon\cN^\ren + C_\varepsilon N^{1/2}.
	\end{equation*}
	Since this will play a role when proving Propositions~\ref{prop:general_two_body_bound_renormalised}~and~\ref{prop:general_three_and_four_body_bound_renormalised}, we mention that the proof of \eqref{eq:kinetic_plus_interaction_bdd_renormalised_op} does \textit{not} require the usage of the bounds \eqref{eq:general_two_body_bound_renormalised} and \eqref{eq:general_four_body_bound_renormalised}.
\end{remark}

\begin{lemma}
	\label{lemma:generator_estimate_drift}
	Assume the same hypotheses as in Proposition~\ref{prop:gronwall_bound}. Then,
	\begin{equation*}
		\cL_1 + \cL_4^{(1)} = \dfrac{N^{5/2}}{2}\int_{\R^9}\left((V_Nf_N)(x,y,z)\varphi_t(x)\vert\varphi_t(y)\vert^2\vert\varphi_t(z)\vert^2b_x + \hc \right) + \cE,
	\end{equation*}
	for some $\cE$ that satisfies
	\begin{equation*}
		\pm\cE \leq \varepsilon\cK^\ren + C_\varepsilon\cN^\ren,
	\end{equation*}
	for all $\varepsilon > 0$.
\end{lemma}

\subsubsection{Analysis of $\cK, \cL_3^{(1)}$ and $\cV_N$: proof of Lemma~\ref{lemma:generator_estimate_kinetic_cubic_interaction}}

By definition of $\cV^\ren$, and since $\omega_{\lambda,N}$ coincides with $\omega_N$ on the support of $V_N$, we have
\begin{multline}
	\label{eq:interaction_term_renormalised}
	\cV_N = \cV^\ren - \dfrac{N^{3/2}}{6}\int_{\R^9}\left((V_N\omega_N)(x,y,z)\varphi_t(x)\varphi_t(y)\varphi_t(z)d_{xyz}^* + \hc\right)\\
	+ \dfrac{N^3}{6}\int_{\R^9}(V_N\omega_N^2)(x,y,z)\vert\varphi_t(x)\vert^2\vert\varphi_t(y)\vert^2\vert\varphi_t(z)\vert^2.
\end{multline}
Similarly, using additionally that $\Theta_3^{(1)}$ satisfies \eqref{eq:theta_coef_estimate_almost_1}, we write
\begin{multline}
	\label{eq:cubic_term_renormalised}
	\cL_3^{(1)} = \dfrac{N^{3/2}}{6}\int_{\R^9}V_N(x,y,z)\varphi_t(x)\varphi_t(y)\varphi_t(z)d_{xyz}^* + \hc\\
		-\dfrac{N^3}{3}\int_{\R^9}(V_N\omega_N)(x,y,z)\vert\varphi_t(x)\vert^2\vert\varphi_t(y)\vert^2\vert\varphi_t(z)\vert^2 + \cE_3^{(1)},
\end{multline}
for some $\cE_3^{(1)}$ that satisfies
\begin{equation*}
	\pm\cE_3^{(1)} \leq \varepsilon\cV^\ren + C_\varepsilon\cN_+,
\end{equation*}
for all $\varepsilon > 0$. Furthermore, by definition of $\cK^\ren$ we have
\begin{multline}
	\label{eq:kinetic_term_rewritten}
	\cK = \cK^\ren - \dfrac{1}{2}\int_{\R^9}\left((-\Delta_xk_t(x,y,z))b_x^*b_y^*b_z^* + \hc\right)\\
	-\dfrac{1}{4}\int_{\R^{15}}\overline{k_t(x,y,z)}(-\Delta_x)k_t(x,y',z')b_yb_zb_{y'}^*b_{z'}^*.
\end{multline}
We claim that this implies
\begin{multline*}
	\cK = \cK^\ren - \dfrac{N^{3/2}}{6}\int_{\R^9}\left((V_Nf_N)(x,y,z)\varphi_t(x)\varphi_t(y)\varphi_t(z)d_{xyz}^* + \hc\right) \numberthis \label{eq:kinetic_term_renormalised}\\
	+ \dfrac{N^3}{6}\int_{\R^9}(V_Nf_N\omega_N)(x,y,z)\vert\varphi_t(x)\vert^2\vert\varphi_t(y)\vert^2\vert\varphi_t(z)\vert^2 + \cE_\Delta,
\end{multline*}
for some $\cE_\Delta$ satisfying
\begin{equation*}
	\pm\cE_\Delta \leq \varepsilon\cK^\ren + C_\varepsilon N^{-1/2}(\cN_+ + \cK + \cV_N + N),
\end{equation*}
for all $\varepsilon > 0$. Combining \eqref{eq:interaction_term_renormalised},\eqref{eq:cubic_term_renormalised} and \eqref{eq:kinetic_term_renormalised}, and using the identity $1 - \omega_N = f_N$ we obtain
\begin{multline}
	\label{eq:interaction_cubic_kinetic_term_renormalised}
	\cK + \cL_3^{(1)} + \cV_N = -\dfrac{N^3}{6}\int_{\R^9}(V_N\omega_N)(x,y,z)\vert\varphi_t(x)\vert^2\vert\varphi_t(y)\vert^2\vert\varphi_t(z)\vert^2 + \cK^\ren + \cV^\ren\\
	+ \cE_3^{(1)} + \cE_\Delta,
\end{multline}
which when summed with $\cL_0$ gives \ref{eq:generator_estimate_kinetic_cubic_interaction}.

Let us prove \eqref{eq:kinetic_term_renormalised}. We start by analysing the second term in \eqref{eq:kinetic_term_rewritten}. By symmetry, and since $\omega_{\lambda,N}$ solves \eqref{eq:truncated_scattering_equation_no_modified_laplacian}, we can expand
\begin{equation*}
	-\dfrac{1}{2}\int_{\R^9}(-\Delta_xk_t(x,y,z))b_x^*b_y^*b_z^* + \hc = \cI_\Delta^{(1)} + \cI_\Delta^{(2)} + \cI_\Delta^{(3)} + \cI_\Delta^{(4)},
\end{equation*}
with $\cI_\Delta^{(1)}, \cI_\Delta^{(2)}, \cI_\Delta^{(3)}$ and $\cI_\Delta^{(4)}$ given by
\begin{align*}
	\cI_\Delta^{(1)} &= -\dfrac{N^{3/2}}{6}\int_{\R^9}(V_Nf_N)(x,y,z)\varphi_t(x)\varphi_t(y)\varphi_t(z)b_x^*b_y^*b_z^* + \hc,\\
	\cI_\Delta^{(2)} &= -\dfrac{N^{3/2}}{2}\int_{\R^9}\omega_{\lambda,N}(x,y,z)(-\Delta\varphi_t)(x)\varphi_t(y)\varphi_t(z)b_x^*b_y^*b_z^* + \hc,\\
	\cI_\Delta^{(3)} &= N^{3/2}\int_{\R^9}(\nabla_x\omega_{\lambda,N}(x,y,z))\cdot(\nabla\varphi_t)(x)\varphi_t(y)\varphi_t(z)b_x^*b_y^*b_z^* + \hc
\end{align*}
and
\begin{equation*}
	\cI_\Delta^{(4)} = \dfrac{N^{-1/2}}{6}\int_{\R^9}\varepsilon_\lambda(x,y,z)\varphi_t(x)\varphi_t(y)\varphi_t(z)b_x^*b_y^*b_z^* + \hc
\end{equation*}
Next, we shown that $\cI_\Delta^{(2)}, \cI_\Delta^{(3)}$ and $\cI_\Delta^{(4)}$ are errors.
\\

\noindent
\textit{Analysis of $\cI_\Delta^{(2)}$.} Using the Cauchy--Schwarz inequality and Lemma~\ref{lemma:three_body_kernel_second_quantisation_estimates} just as we did in the proof of Lemma~\ref{lemma:number_operator_renormalised}, we can show
\begin{equation*}
	\pm\cI_\Delta^{(2)} \leq C(\cN_+ + 1).
\end{equation*}
\\

\noindent
\textit{Analysis of $\cI_\Delta^{(3)}$.} We use an integration by parts to write
\begin{multline*}
	\cI_\Delta^{(3)} = N^{3/2}\int_{\R^3}\d{}xb_x^*\int_{\R^6}\omega_{\lambda,N}(x,y,z)(-\Delta\varphi_t)(x)\varphi_t(y)\varphi_t(z)b_y^*b_z^* + \hc\\
	-N^{3/2}\int_{\R^3}\d{}x(\nabla_xb_x^*)\int_{\R^6}\omega_{\lambda,N}(x,y,z)(\nabla\varphi_t)(x)\varphi_t(y)\varphi_t(z)b_y^*b_z^* + \hc
\end{multline*}
The first term can again be bounded using Lemma~\ref{lemma:three_body_kernel_second_quantisation_estimates}. Applying the Cauchy--Schwarz inequality and the identity \eqref{eq:second_quantised_kinetic_and_interaction_mod} to bound the second term, we find
\begin{multline*}
	\pm\cI_\Delta^{(3)} \leq C(\cN_+ + 1) + \varepsilon\cK\\
	+ \varepsilon^{-1}CN^3\int_{\R^{15}}\omega_{\lambda,N}(x,y,z)\omega_{\lambda,N}(x,y',z')\vert\nabla\varphi_t(x)\vert^2\varphi_t(y)\varphi_t(z)\overline{\varphi_t(y')\varphi_t(z')}b_{y'}b_{z'}b_y^*b_z^* + \hc,
\end{multline*}
for all $\varepsilon > 0$. Then, we use the CCR \eqref{eq:creation_annihilation_op_modified_ccr} to put the last term in normal order, and we use the bounds \eqref{eq:truncated_scattering_equation_Lp_estimates} and \eqref{eq:truncated_scattering_equation_LpLinf_estimates} to deduce
\begin{equation*}
	\pm\cI_\Delta^{(3)} \leq C(\cN_+ + 1) + \varepsilon\cK + \varepsilon^{-1}C(1 + \lambda^2N^{-1/2}\cN_+ + \lambda^2N^{-1/2}\tilde{\cV}_\twoB),
\end{equation*}
for all $\varepsilon > 0$, where $\tilde{\cV}_\twoB$ is defined as in Proposition~\ref{prop:general_three_and_four_body_bound_renormalised}. Here, we used that by definition of the cut-off contained in $\omega_{\lambda,N}$ we have $\vert y - z\vert \leq \lambda$ and likewise for $y'$ and $z'$. Thanks to Proposition~\ref{prop:general_three_and_four_body_bound_renormalised} and the fact that $\lambda$ is of order $N^{-1/3}$, we obtain
\begin{equation*}
	\pm\cI_\Delta^{(3)} \leq C\cN_+ + CN^{-1/2}(\cK + \cV_N + N).
\end{equation*}
\\

\noindent
\textit{Analysis of $\cI_\Delta^{(4)}$.} To bound $\cI_\Delta^{(4)}$ we will use \eqref{eq:truncated_scattering_equation_error_green_function_solves} to substitute $\varepsilon_\lambda(x,y,z)$ with $-\Delta_xu_\lambda(x,y,z)$ and integrate by parts to make good use of the kinetic operator $\cK$. However, when performing this substitution, we want to keep the information that $y$ and $z$ are at a distance less than $\lambda$ from one another, which is contained in $\varepsilon_\lambda$ but not in $u_\lambda$. Namely, write
\begin{align*}
	\cI_\Delta^{(4)} &= \dfrac{N^{-1/2}}{6}\int_{\R^9}\varepsilon_\lambda(x,y,z)\1(\vert y - z\vert \leq \lambda)\varphi_t(x)\varphi_t(y)\varphi_t(z)b_x^*b_y^*b_z^* + \hc\\
	&= \dfrac{N^{-1/2}}{6}\int_{\R^9}((-\Delta_x - \Delta_y - \Delta_z)u_\lambda(x,y,z))\1(\vert y - z\vert \leq \lambda)\varphi_t(x)\varphi_t(y)\varphi_t(z)b_x^*b_y^*b_z^* + \hc
\end{align*}
Recall that we use the short-hand notation $u_\lambda(x,y,z) = u_\lambda(x - y,x - z)$. Hence, we have the identity
\begin{equation*}
	(-\Delta_x - \Delta_y - \Delta_z)u_\lambda(x - y, x - z) = -2\Delta_xu_\lambda(x - y,x - z)
\end{equation*}
Thus, thanks to an integration by parts and the definition of $c_x^*$, we get
\begin{equation*}
	\cI_\Delta^{(4)} = \cJ_\Delta^{(4,1)} + \cJ_\Delta^{(4,2)} + \cJ_\Delta^{(4,3)},
\end{equation*}
with $\cJ_\Delta^{(4,1)}, \cJ_\Delta^{(4,2)}$ and $\cJ_\Delta^{(4,3)}$ given by
\begin{align*}
	\cJ_\Delta^{(4,1)} &= \dfrac{N^{-1/2}}{3}\int_{\R^9}(\nabla_x u_\lambda(x,y,z))\1(\vert y - z\vert \leq \lambda)\varphi_t(x)\varphi_t(y)\varphi_t(z)\nabla c_x^*b_y^*b_z^* + \hc,\\
	\cJ_\Delta^{(4,2)} &= \dfrac{N^{-1/2}}{3}\int_{\R^{9}}(\nabla_xu_\lambda(x,y,z))\1(\vert y - z\vert \leq \lambda)(\nabla\varphi_t)(x)\varphi_t(y)\varphi_t(z)b_x^*b_y^*b_z^* + \hc,
\end{align*}
and
\begin{equation*}
	\cJ_\Delta^{(4,3)} = \dfrac{N^{-1/2}}{6}\int_{\R^{15}}(\nabla_x u_\lambda(x,y,z))\1(\vert y - z\vert \leq \lambda)\varphi_t(x)\varphi_t(y)\varphi_t(z)\overline{\nabla_xk_t(x,y',z')}b_{y'}b_{z'}b_y^*b_z^* + \hc
\end{equation*}
The term $\cJ_\Delta^{(4,1)}$ can be bounded by following similar steps as for $\cI_\Delta^{(3)}$. Namely, we use the Cauchy--Schwarz inequality to separate $\nabla_xc_x^*$ from $b_y^*b_z^*$, which yields $\cK^\ren$ and a term proportional to $b_{y'}b_{z'}b_y^*b_z^*$. Then, we apply the CCR \eqref{eq:creation_annihilation_op_modified_ccr} to normal order $b_{y'}b_{z'}b_y^*b_z^*$ and use the estimate \eqref{eq:truncated_scattering_equation_error_Lp_estimates} to bound the simpler terms. Doing so, we are left with
\begin{align*}
	\pm\cJ_\Delta^{(4,1)} &\leq \varepsilon\cK^\ren + \varepsilon^{-1}C(\lambda^{-4}N^{-1} + \lambda^{-1}N^{-1}\cN_+)\\
	&\phantom{\leq} + \varepsilon^{-1}C N^{-1}\int_{\R^{15}}
	\begin{multlined}[t]
		(\nabla_x u_\lambda(x,y,z))(\nabla_x u_\lambda(x,y',z'))\1(\vert y - z\vert \leq \lambda)\1(\vert y' - z'\vert \leq \lambda)\\
		\times \vert\varphi_t(x)\vert^2\varphi_t(y)\overline{\varphi_t(y')}\varphi_t(z)\overline{\varphi_t(z')}b_y^*b_z^*b_{y'}b_{z'} + \hc,
	\end{multlined}
\end{align*}
for all $\varepsilon > 0$. Next, we use the Cauchy--Schwarz inequality, the estimates \eqref{eq:truncated_scattering_equation_error_Lp_estimates}, and Proposition~\ref{prop:general_three_and_four_body_bound_renormalised} to deduce
\begin{multline*}
	\pm N^{-1}\int_{\R^{15}}
	\begin{multlined}[t]
		(\nabla_x u_\lambda(x,y,z))(\nabla_x u_\lambda(x,y',z'))\1(\vert y - z\vert \leq \lambda)\1(\vert y' - z'\vert \leq \lambda)\vert\varphi_t(x)\vert^2\\
		\times\varphi_t(y)\overline{\varphi_t(y')}\varphi_t(z)\overline{\varphi_t(z')}b_y^*b_z^*b_{y'}b_{z'} + \hc
	\end{multlined}\\
	\leq CN^{-1}\lambda^{-1}\int_{\R^6}\1(\vert y - z\vert \leq \lambda)b_y^*b_z^*b_yb_z \leq CN^{-1}\lambda^{-1}(\cK + \cV_N + \cN_+),
\end{multline*}
which implies
\begin{equation*}
	\pm\cJ_\Delta^{(4,1)} \leq \varepsilon\cK^\ren + \varepsilon^{-1}CN^{-1}\lambda^{-1}(\cN_+ + \cK + \cV_N + \lambda^{-3}),
\end{equation*}
for all $\varepsilon > 0$. Following the same strategy to bound $\cJ_\Delta^{(4,2)}$, we find
\begin{equation*}
	\pm\cJ_\Delta^{(4,2)} \leq C\cN_+ + CN^{-1}\lambda^{-1}(\cN_+ + \cK + \cV_N + \lambda^{-3}).
\end{equation*}
Lastly, to estimate $\cJ_\Delta^{(4,3)}$ we once more use the CCR \eqref{eq:creation_annihilation_op_modified_ccr}. The two terms coming from the commutator are easily bounded using the estimate $\vert \nabla u_\lambda\vert \leq C\lambda^{-5}$ and the $L^1$-norm estimates \eqref{eq:truncated_scattering_equation_Lp_estimates}, which gives
\begin{multline*}
	\cJ_\Delta^{(4,3)} = \dfrac{N^{-1/2}}{6}\int_{\R^{15}}(\nabla_x u_\lambda(x,y,z))\1(\vert y - z\vert \leq \lambda)\varphi_t(x)\varphi_t(y)\varphi_t(z)\overline{\nabla_xk_t(x,y',z')}b_y^*b_z^*b_{y'}b_{z'} + \hc\\
	+ \cE_\Delta^{(1)},
\end{multline*}
for some $\cE_\Delta^{(1)}$ satisfying
\begin{equation*}
	\pm\cE_\Delta^{(1)} \leq CN^{-1}\lambda^{-1}(\cN_+ + \lambda^{-3}).
\end{equation*}
Thanks to the Cauchy--Schwarz inequality, the remaining term can be split as
\begin{multline*}
	\pm\dfrac{N^{-1/2}}{6}\int_{\R^{15}}(\nabla_x u_\lambda(x,y,z))\1(\vert y - z\vert \leq \lambda)\varphi_t(x)\varphi_t(y)\varphi_t(z)\overline{\nabla_xk_t(x,y',z')}b_y^*b_z^*b_{y'}b_{z'} + \hc\\
	\leq \varepsilon CN(\Vert\omega_{\lambda,N}\Vert_{L^1} + \Vert\nabla\omega_{\lambda,N}\Vert_{L^1})\int_{\R^9}\vert \nabla_xu_\lambda(x,y,z)\vert^{1-\eta}\1(\vert y - z\vert \leq \lambda)b_y^*b_z^*b_yb_z\\
	+ \varepsilon^{-1} CN(\Vert\omega_{\lambda,N}\Vert_{L^1L^\infty} + \Vert\nabla\omega_{\lambda,N}\Vert_{L^1L^\infty})\Vert \nabla u_{\lambda}\Vert_{L^{1+\eta}}^{1+\eta}\int_{\R^6}\1(\vert y' - z'\vert \leq \lambda)b_{y'}^*b_{z'}^*b_{y'}b_{z'},
\end{multline*}
for any $\varepsilon,\eta > 0$. Taking $\eta = 3/10$ so that the integrations can be carried, using the estimates \eqref{eq:truncated_scattering_equation_error_green_function_pointwise_estimate}, \eqref{eq:truncated_scattering_equation_Lp_estimates} and \eqref{eq:truncated_scattering_equation_LpLinf_estimates}, and Proposition~\ref{prop:general_three_and_four_body_bound_renormalised}, and then optimising over $\eta$, we get
\begin{multline*}
	\pm\dfrac{N^{-1/2}}{6}\int_{\R^{15}}(\nabla_x u_\lambda(x,y,z))\1(\vert y - z\vert \leq \lambda)\varphi_t(x)\varphi_t(y)\varphi_t(z)\overline{\nabla_xk_t(x,y',z')}b_y^*b_z^*b_{y'}b_{z'} + \hc\\
	\leq CN^{-1/2}(\cN_+ + \cK + \cV_N).
\end{multline*}
In summary, we have shown that
\begin{equation*}
	\pm\cI_\Delta^{(4)} \leq \varepsilon\cK^\ren + C_\varepsilon\cN_+ + C_\varepsilon N^{-1/2}(\cK + \cV_N) + C_\varepsilon N^{1/3},
\end{equation*}
for all $\varepsilon > 0$.

Gathering the previous estimates and using that $\lambda = 2RN^{-1/3}$, we obtain
\begin{multline}
	\label{eq:kinetic_term_renormalised_first_term_rewritten}
	-\dfrac{1}{2}\int_{\R^9}\left((-\Delta_x)k_t(x,y,z)b_x^*b_y^*b_z^* + \hc\right) = \dfrac{N^3}{3}\int_{\R^9}(V_Nf_N\omega_N)(x,y,z)\vert\varphi_t(x)\vert^2\vert\varphi_t(y)\vert^2\vert\varphi_t(z)\vert^2\\
	-\dfrac{N^{3/2}}{6}\int_{\R^9}\left((V_Nf_N)(x,y,z)\varphi_t(x)\varphi_t(y)\varphi_t(z)d_{xyz}^* + \hc\right) + \cE_\Delta^{(2)},
\end{multline}
for some $\cE_\Delta^{(2)}$ satisfying
\begin{equation*}
	\pm \cE_\Delta^{(2)} \leq \varepsilon\cK^\ren + C_\varepsilon N^{-1/2}(\cN_+ + \cK + \cV_N + N),
\end{equation*}
for all $\varepsilon > 0$.
\\

Finally, we analyse the last term in \eqref{eq:kinetic_term_rewritten}. We use the commutation relations \eqref{eq:creation_annihilation_op_modified_ccr} to expand
\begin{multline}
	\label{eq:kinetic_term_renormalised_second_term}
	- \dfrac{1}{4}\int_{\R^{15}}\overline{k_t(x,y,z)}(-\Delta_x)k_t(x,y',z')b_yb_zb_{y'}^*b_{z'}^*\\
	= -\dfrac{1}{2}\Vert\nabla\otimes Q_t^{\otimes 2}k_t\Vert_{L^2}^2 -\int_{\R^{12}}\overline{(\nabla\otimes Q_t\otimes \1k_t)(x,y,z)}\nabla_xk_t(x,y,z')b_{z'}^*b_z\\
	-\dfrac{1}{4}\int_{\R^{15}}\overline{k_t(x,y,z)}(-\Delta_x)k_t(x,y',z')b_{y'}^*b_{z'}^*b_yb_z
\end{multline}
The first two terms are easily dealt with, using in particular the scattering equation \eqref{eq:three_body_scattering_equation}, and we find
\begin{multline*}
	-\dfrac{1}{2}\Vert\nabla\otimes Q_t^{\otimes 2}k_t\Vert_{L^2}^2 -\int_{\R^{12}}\overline{(\nabla\otimes Q_t\otimes \1k_t)(x,y,z)}\nabla_xk_t(x,y,z')b_{z'}^*b_z\\
	= -\dfrac{N^3}{6}\int_{\R^9}(V_N\omega_Nf_N)(x,y,z)\vert\varphi_t(x)\vert^2\vert\varphi_t(y)\vert^2\vert\varphi_t(z)\vert^2 + \cE_\Delta^{(3)},
\end{multline*}
for some $\cE_\Delta^{(3)}$ that satisfies
\begin{equation*}
	\pm\cE_\Delta^{(3)} \leq C\cN_+.
\end{equation*}
The only difficulty in the bounding the remaining term is when the two derivatives hit the function $\omega_{\lambda,N}$ contained in the kernel $k_t(x,y',z')$. Since $\omega_{\lambda,N}$ solves \eqref{eq:three_body_scattering_equation}, this can be rewritten as
\begin{multline*}
	-\dfrac{N^{3/2}}{4}\int_{\R^{15}}\overline{k_t(x,y,z)}(-\Delta_x\omega_{\lambda,N}(x,y,z))\varphi_t(x)\varphi_t(y')\varphi_t(z')b_{y'}^*b_{z'}^*b_yb_z\\
	\begin{multlined}[t]
		= -CN^{3/2}\int_{\R^{15}}\overline{k_t(x,y,z)}(V_Nf_N)(x,y',z')\varphi_t(x)\varphi_t(y')\varphi_t(z')b_{y'}^*b_{z'}^*b_yb_z\\
		-CN^{-1/2}\int_{\R^{15}}\overline{k_t(x,y,z)}\varepsilon_\lambda(x,y',z')\varphi_t(x)\varphi_t(y')\varphi_t(z')b_{y'}^*b_{z'}^*b_yb_z.
	\end{multlined}
\end{multline*}
Both terms can be bounded using the Cauchy--Schwarz inequality, the estimates \eqref{eq:truncated_scattering_equation_Lp_estimates}, \eqref{eq:truncated_scattering_equation_LpLinf_estimates} and \eqref{eq:truncated_scattering_equation_error_Lp_estimates}, and applying Propositions~\ref{prop:general_two_body_bound_renormalised}~and~\ref{prop:general_three_and_four_body_bound_renormalised}. The reason we can use the former proposition when bounding some of the errors is that the range of $V_Nf_N$ is proportional to $N^{-1/2}$ (compared to $\lambda$ for $\omega_{\lambda,N}$ or $\varepsilon_\lambda$). All and all, we find
\begin{multline}
	\label{eq:kinetic_term_renormalised_second_term_rewritten}
	- \dfrac{1}{4}\int_{\R^{15}}\overline{k_t(x,y,z)}(-\Delta_x)k_t(x,y',z')b_yb_zb_{y'}^*b_{z'}^*\\
	= -\dfrac{N^3}{6}\int_{\R^9}(V_N\omega_Nf_N)(x,y,z)\vert\varphi_t(x)\vert^2\vert\varphi_t(y)\vert^2\vert\varphi_t(z)\vert^2 + \cE_\Delta^{(4)},
\end{multline}
for some $\cE_\Delta^{(4)}$ that satisfies
\begin{equation*}
	\pm\cE_\Delta^{(4)} \leq C\cN_+ + CN^{-7/12}(\cK + \cV_N).
\end{equation*}
This concludes the proof of Lemma~\ref{lemma:generator_estimate_kinetic_cubic_interaction}.

\subsubsection{Analysis of $\cL_1$ and $\cL_4^{(1)}$: proof of Lemma~\ref{lemma:generator_estimate_drift}}

Thanks to the estimate \eqref{eq:theta_coef_estimate_almost_1} and the Cauchy--Schwarz inequality, we have
\begin{equation*}
	\mathcal{L}_1 = \dfrac{N^{5/2}}{2}\int_{\R^9}\left(V_N(x,y,z)\vert\varphi_t(y)\vert^2\vert\varphi_t(z)\vert^2\varphi_t(x)b_x^* + \hc\right) + \cE_1,
\end{equation*}
with $\cE_1$ satisfying
\begin{equation*}
	\pm\cE_1 \leq C\cN_+.
\end{equation*}
Similarly, we have
\begin{align*}
	\mathcal{L}_4^{(1)} &=
	\begin{multlined}[t]
		\dfrac{N}{2}\int_{\R^9}V_N(x,y,z)\varphi_t(y)\varphi_t(z)d_{xyz}^*b_x\Theta_4^{(1)} + \hc\\
		-\dfrac{N^{5/2}}{2}\int_{\R^9}(V_N\omega_N)(x,y,z)\overline{\varphi_t(x)}\vert\varphi_t(y)\vert^2\vert\varphi_t(z)\vert^2b_x\Theta_4^{(1)} + \hc
	\end{multlined}\\
	&= -\dfrac{N^{5/2}}{2}\int_{\R^9}\left((V_N\omega_N)(x,y,z)\overline{\varphi_t(x)}\vert\varphi_t(y)\vert^2\vert\varphi_t(z)\vert^2b_x + \hc\right) + \cE_4^{(1)},
\end{align*}
with $\cE_4^{(1)}$ satisfying
\begin{equation*}
	\pm\cE_4^{(1)} \leq \varepsilon\cV^\ren + C_\varepsilon\cN_+,
\end{equation*}
for any $\varepsilon > 0$. Therefore,
\begin{equation*}
	\cL_1 + \cL_4^{(1)} = \dfrac{N^{5/2}}{2}\int_{\R^9}\left((V_Nf_N)(x,y,z)\varphi_t(x)\vert\varphi_t(y)\vert^2\vert\varphi_t(z)\vert^2b_x^* + \hc\right) + \cE_1 + \cE_4^{(1)}.
\end{equation*}
This concludes the proof of Lemma~\ref{lemma:generator_estimate_drift}.

\subsubsection{Conclusion of the proof of Lemma~\ref{lemma:generator_estimate}}

Applying Lemmas~\ref{lemma:generator_estimate_kinetic_cubic_interaction}~and~\ref{lemma:generator_estimate_drift}, gathering the estimates on the error terms, bounding $\cV_\twoB$ using Proposition~\ref{prop:general_two_body_bound_renormalised}, and using that $\lambda = 2RN^{-1/3}$, we obtain
\begin{multline*}
	\cH_N = \dfrac{N^{5/2}}{2}\int_{\R^9}(V_Nf_N)(x,y,z)\varphi_t(x)\vert\varphi_t(y)\vert^2\vert\varphi_t(z)\vert^2b_x^* - \dfrac{N^{1/2}}{2}b(V)b^*(\vert\varphi_t\vert^4\varphi_t) + \hc\\
	+ \dfrac{N^3}{6}\int_{\R^9}(V_Nf_N)(x,y,z)\vert\varphi_t(x)\vert^2\vert\varphi_t(y)\vert^2\vert\varphi_t(z)\vert^2 - \dfrac{N}{6}b(V)\Vert\varphi_t\Vert_6^6 + \cK^\ren + \cV^\ren + \cE,
\end{multline*}
for some $\cE$ satisfying \eqref{eq:generator_estimate}. Finally, the desired estimate \eqref{eq:generator_estimate_error} follows from the estimates
\begin{equation}
	\label{eq:constant_contribution_close_GP}
	\pm\Big(\dfrac{N^3}{6}\int_{\R^9}(V_Nf_N)(x,y,z)\vert\varphi_t(x)\vert^2\vert\varphi_t(y)\vert^2\vert\varphi_t(z)\vert^2 - \dfrac{N}{6}b(V)\Vert\varphi_t\Vert_6^6\Big) \leq CN^{1/2}
\end{equation}
and
\begin{multline}
	\label{eq:linear_terms_close}
	\pm\Big(\dfrac{N^{5/2}}{2}\int_{\R^9}(V_Nf_N)(x,y,z)\varphi_t(x)\vert\varphi_t(y)\vert^2\vert\varphi_t(z)\vert^2b_x^* - \dfrac{N^{1/2}}{2}b(V)b^*(\vert\varphi_t\vert^4\varphi_t) + \hc\Big)\\
	\leq C(\cN_+ + 1).
\end{multline}
Both \eqref{eq:constant_contribution_close_GP} and \eqref{eq:linear_terms_close} follow from the estimate
\begin{equation}
	\label{eq:constant_contribution_close_GP_Linf_estimate}
	\Big\Vert N^2\int_{\R^6}\d{}y\d{}z(V_Nf_N)(\cdot,y,z)\vert\varphi_t(y)\vert^2\vert\varphi_t(z)\vert^2 - b(V)\vert\varphi_t(\cdot)\vert^4\Big\Vert_{L^\infty(\R^3)} \leq CN^{-1/2}.
\end{equation}
To prove it, we use a change of variables and the bound
\begin{equation*}
	\left\vert\vert\varphi_t(x + N^{-1/2}z)\vert^2 - \vert\varphi_t(x)\vert^2\right\vert = \Big\vert\int_0^1\d{}s\dfrac{\d{}}{\d{} s}\vert\varphi_t(x + sN^{-1/2}z)\vert^2\Big\vert \leq CN^{-1/2}\vert z\vert,
\end{equation*}
and then repeat this with the variable $y$ instead of $z$. This concludes the proof of Lemma~\ref{lemma:generator_estimate}.

\subsection{Proof of Lemma~\ref{lemma:generator_commutator_estimates}}

We frequently use the regularity estimates \eqref{eq:gross_pitaevskii_solution_estimate_H1_L6}, \eqref{eq:gross_pitaevskii_solution_estimate_H4} and \eqref{eq:gross_pitaevskii_solution_estimate_Linf} without mentioning them to avoid repetitions. Moreover, most estimates and identities here should be interpreted in the quadratic form sense on $\cF_+^{\leq N}$, and we therefore do not say this each time. 

Using that $\Phi_{N,t}$ solves \eqref{eq:fluctuation_vector_equation}, we write
\begin{equation}
	\label{eq:generator_fluctuation_time_derivative}
	\partial_t\langle\cH_N + \cQ^\ren\rangle_{\Phi_{N,t}} = \langle\partial_t(\cH_N + \cQ^\ren) + \ui[\cH_N,\cQ^\ren]\rangle_{\Phi_{N,t}}.
\end{equation}
Recall that
\begin{multline}
	\label{eq:generator_fluctuation_rewritten_repeat}
	\mathcal{H}_N = \cK + \cV_N + a^*(\varphi_t)b(\ui\partial_t \varphi_t) - \dfrac{b(V)}{2}\big(\sqrt{N - \cN_+}b(\vert\varphi_t\vert^4\varphi_t) + \hc\big)\\
	-\dfrac{b(V)}{6}\Vert\varphi_t\Vert_6^6(N - \cN_+) + \sum_{j = 0}^5\mathcal{L}_j,
\end{multline}
with $\cL_1,\dots,\cL_5$ given below \eqref{eq:hamiltonian_conjugation_excitation_map}. Just as in the proof of Lemma~\ref{lemma:generator_estimate}, the terms for which we need to identify a cancellation at the main order are $\cK, \cL_0, \cL_1, \cL_3^{(1)}, \cL_4^{(1)}$ and $\cV_N$; the other terms are errors that can be bounded directly.

We start by explaining how to bound the error terms from $\partial_t(\cH_N + \cQ^\ren)$. The time derivatives of $\cL_2^{(1)}$, $\cL_2^{(2)}$, $\cL_2^{(3)}$, $\cL_3^{(2)}$, $\cL_3^{(3)}$, $\cL_4^{(2)}$, $\cL_4^{(3)}$ and $\cL_5$ are bounded more or less the same way as in the previous section, and the time derivative of $\cQ^\ren$ is bounded just as in the proof of Lemma~\ref{lemma:number_operator_renormalised}. To see why this is true, we distinguish between three cases. Firstly, the time derivative can hit one the wavefunctions $\varphi_t$ in the kernel of these operators, in which case the term is bounded as in the proof of Lemma~\ref{lemma:generator_estimate} using the higher regularity estimate \eqref{eq:gross_pitaevskii_solution_estimate_Linf}. Secondly, the time derivative can hit one of the modified creation and annihilation operators, in which case we use the identity
\begin{equation}
	\label{eq:modified_creation_annihilation_op_derivative}
	\partial_tb_x = -\partial_t\varphi_t(x)a(\varphi_t) - \varphi_t(x)a(\partial_t\varphi_t) = -\varphi_t(x)b(\partial_t\varphi_t)
\end{equation}
on $\cF_+^{\leq N}$ (and likewise for $\partial_tb_x^*$). Then, it is easy to see that such terms are bounded by $\cN_+ + \cV_\twoB + \cV^\ren$, and therefore by $\cN^\ren + \cK^\ren + \cV^\ren$ from Proposition~\ref{prop:general_two_body_bound_renormalised}. Thirdly, the time derivative can hit one of the prefactors $\Theta_j^{(i)}(\cN_+)$. These terms however vanish when tested on both sides by elements of $\cF_+^{\leq N}$ due to the identity $\partial_t\cN_+ = -a^*(\partial_t\varphi_t)a(\varphi_t) - a^*(\varphi_t)a(\partial_t\varphi_t)$; to see why this identity implies the claim one can for instance use a series expansion of $\Theta_j^{(i)}$.

Unlike in the proof of Lemma~\ref{lemma:generator_estimate} where contributions of order $N$ were extracted from $\cK$ and $\cV_N$, the terms $\partial_t\cK$ and $\partial_t\cV_N$ are both errors, as we now show. As an immediate consequence of \eqref{eq:modified_creation_annihilation_op_derivative}, we have
\begin{equation*}
	\pm\partial_t\cK = \mp(b^*(-\Delta\varphi_t)b(\partial_t\varphi_t) + \hc) \leq C\cN_+.
\end{equation*}
Furthermore, using again \eqref{eq:modified_creation_annihilation_op_derivative} and the identity
\begin{equation}
	\label{eq:renormalised_annihilation_op_three_rewritten}
	b_xb_yb_z = d_{xyz} - k_t(x,y,z),
\end{equation}
we find
\begin{multline*}
	\partial_t\cV_N = -\dfrac{1}{2}\int_{\R^9}\left(V_N(x,y,z)\varphi_t(z)d_{xyz}^*b_xb_yb(\partial_t\varphi_t) + \hc\right)\\
	+ \dfrac{1}{2}\int_{\R^9}(V_Nk_t)(x,y,z)\varphi_t(z)b_xb_yb(\partial_t\varphi_t) + \hc.
\end{multline*}
Thanks to the Cauchy--Schwarz inequality, the bound
\begin{equation*}
	\pm b^*(\partial_t\varphi_t)b(\partial_t\varphi_t) \leq C\cN_+
\end{equation*}
and Proposition~\ref{prop:general_two_body_bound_renormalised}, we deduce
\begin{equation*}
	\pm\partial_t\cV_N \leq C(\cN_+ + \cV_\twoB + \cV^\ren) \leq C(\cN^\ren + \cK^\ren + \cV^\ren + N^{1/2}).
\end{equation*}

Having bounded the error terms from $\partial_t(\cH_N + \cQ^\ren)$, we are left with $\partial_t\cL_0, \partial_t\cL_1, \partial_t\cL_3^{(1)}, \partial_t\cL_4^{(1)}$ and $[\cH_N,\cQ^\ren]$. Because the analysis of the error terms from $[\cH_N,\cQ^\ren]$ is tedious but rather straightforward, we postpone it to the end of the present proof and focus first on the leading-order terms. Recall that in the proof of Lemma~\ref{lemma:generator_estimate} a contribution of order $N$ was extracted from $\cL_3^{(1)}$, and that this term was combined with $\cK$ and $\cV_N$ to make $\cK^\ren$ and $\cV^\ren$ appear (see \eqref{eq:interaction_cubic_kinetic_term_renormalised}). Though the same procedure can be done with $\partial_t\cK, \partial_t\cL_3^{(1)}$ and $\partial_t\cV_N$, this would not yield the correct bound. Namely, we would obtain the time derivative of the leading-order contribution from \eqref{eq:interaction_cubic_kinetic_term_renormalised}, as well as the time derivatives of $\cK^\ren$ and $\cV^\ren$. The former is not a problem, since it would cancel out with the time derivative of $\cL_0 - b(V)\Vert\varphi_t\Vert_6^6N/6$ (see \eqref{eq:generator_fluctuation_rewritten_repeat}). The latter is however a problem since $\partial_t\cK^\ren$ and $\partial_t\cV^\ren$ cannot be controlled in terms of $\cK^\ren$ and $\cV^\ren$. Instead, the correct way to proceed is to identity a cancellation between $\partial_t\cL_3^{(1)}$ and the commutator $[\cK + \cV_N,\cQ^\ren]$, which is why it is important to include $\cQ^\ren$ in \eqref{eq:generator_fluctuation_time_derivative}, even though it is bounded by $\cN^\ren$. In a similar way, it is crucial to combine $\partial_t\cL_0$ with $[\cL_3^{(1)},\cQ^\ren]$ and $\partial_t\cL_1 + \partial_t\cL_4^{(1)}$ with $[\cL_4^{(1)},\cQ^\ren]$. We formulate these cancellations more precisely in the following three lemmas.
\begin{lemma}
	\label{lemma:generator_commutator_estimates_cubic_derivative}
	Assume the same hypotheses as in Proposition~\ref{prop:gronwall_bound}. Then,
	\begin{equation*}
		\pm\left(\partial_t\cL_3^{(1)} + \ui[\cK + \cV_N,\cQ^\ren]\right) \leq C(\cN^\ren + \cK^\ren + \cV^\ren + N^{1/2}).
	\end{equation*}
\end{lemma}
\begin{lemma}
	\label{lemma:generator_commutator_estimates_cubic_commutator}
	Assume the same hypotheses as in Proposition~\ref{prop:gronwall_bound}. Then,
	\begin{equation*}
		\partial_t\cL_0 + \ui[\cL_3^{(1)},\cQ^\ren] = \dfrac{N^3}{6}\int_{\R^9}(V_Nf_N)(x,y,z)\partial_t(\vert\varphi_t(x)\vert^2\vert\varphi_t(y)\vert^2\vert\varphi_t(z)\vert^2) + \tilde{\cE},
	\end{equation*}
	for some $\tilde{\cE}$ that satisfies
	\begin{equation*}
		\pm\tilde{\cE} \leq C\cN^\ren + CN^{-1/2}(\cK^\ren + \cV^\ren + CN).
	\end{equation*}
\end{lemma}
\begin{lemma}
	\label{lemma:generator_commutator_estimates_drift}
	Assume the same hypotheses as in Proposition~\ref{prop:gronwall_bound}. Then,
	\begin{multline*}
		\partial_t\cL_1 + \partial_t\cL_4^{(1)} + \ui[\cL_4^{(1)},\cQ^\ren] = \dfrac{N^{5/2}}{2}\int_{\R^9}\left((V_Nf_N)(x,y,z)\partial_t(\varphi_t(x)\vert\varphi_t(y)\vert^2\vert\varphi_t(z)\vert^2)b_x^* + \hc\right)\\
		- \dfrac{N^{5/2}}{2}\int_{\R^9}\left((V_Nf_N)(x,y,z)\vert\varphi_t(x)\vert^2\vert\varphi_t(y)\vert^2\vert\varphi_t(z)\vert^2b^*(\partial_t\varphi_t) + \hc\right) + \tilde{\cE},
	\end{multline*}
	for some $\tilde{\cE}$ that satisfies
	\begin{equation*}
		\pm\tilde{\cE} \leq C(\cN^\ren + \cK^\ren + \cV^\ren + N^{1/2}).
	\end{equation*}
\end{lemma}

\subsubsection{Analysis of $\partial_t\cL_3^{(1)}$ and $[\cK + \cV_N,\cQ^\ren]$: proof of Lemma~\ref{lemma:generator_commutator_estimates_cubic_derivative}}

Firstly, we use the identity \eqref{eq:renormalised_annihilation_op_three_rewritten} and that $\Theta_3^{(1)}$ satisfies \eqref{eq:theta_coef_estimate_almost_1} to write
\begin{multline*}
	\partial_t\cL_3^{(1)} = -\dfrac{N^{3/2}}{2}\int_{\R^9}\big((V_N\omega_N)(x,y,z)(\partial_t\varphi_t)(x)\overline{\varphi_t(x)}\vert\varphi_t(y)\vert^2\vert\varphi_t(z)\vert^2 + \hc\big)\\
	+ \dfrac{N^{3/2}}{2}\int_{\R^9}\left(V_N(x,y,z)(\partial_t\varphi_t)(x)\varphi_t(y)\varphi_t(z)d_{xyz}^* + \hc\right) + \tilde{\cE}_3^{(1)},
\end{multline*}
for some $\tilde{\cE}_3^{(1)}$ satisfying
\begin{equation}
	\label{eq:cubic_term_derivative_error}
	\pm\tilde{\cE}_3^{(1)} \leq C(\cN^\ren + 1) + C\cV^\ren.
\end{equation}

Secondly, following essentially the same proof as for \eqref{eq:kinetic_term_renormalised}, using additionally that $\Theta_N$ satisfies \eqref{eq:theta_coef_cubic_transform_almost_1}, we can show that
\begin{align*}
	\ui\left[\cK,\cQ^\ren\right] &= -\dfrac{1}{6}\int_{\R^{12}}\big(\partial_tk_t(x,y,z)b_{x'}^*(-\Delta_{x'})\left[b_{x'},b_x^*b_y^*b_z^*\right]\Theta_N + \hc\big)\\
	&= 
	\begin{multlined}[t]
		-\dfrac{N^{3/2}}{2}\int_{\R^9}\big((V_Nf_N)(x,y,z)(\partial_t\varphi_t(x))\varphi_t(y)\varphi_t(z)d_{xyz}^* + \hc\big)\\
		+ \dfrac{N^{3/2}}{2}\int_{\R^9}(V_Nf_N\omega_N)(x,y,z)(\partial_t\varphi_t(x))\overline{\varphi_t(x)}\vert\varphi_t(y)\vert^2\vert\varphi_t(z)\vert^2 + \hc + \tilde{\cE}_\Delta
	\end{multlined}
\end{align*}
with $\tilde{\cE}_\Delta$ satisfying
\begin{equation}
	\label{eq:kinetic_term_commutator_Q_ren_error}
	\pm\tilde{\cE}_\Delta \leq C(\cN^\ren + \cK^\ren + \cV^\ren + N^{1/2}).
\end{equation}

Thirdly, we use the commutation relations \eqref{eq:creation_annihilation_op_modified_ccr} to expand
\begin{equation*}
	\ui\left[\cV_N,\cQ^\ren\right] = \tilde{\cI}_6^{(1)} + \tilde{\cI}_6^{(2)} + \tilde{\cI}_6^{(3)},
\end{equation*}
with $\tilde{\cI}_6^{(1)}, \tilde{\cI}_6^{(2)}$ and $\tilde{\cI}_6^{(3)}$ given by
\begin{align*}
	\tilde{\cI}_6^{(1)} &= -\dfrac{1}{6}\int_{\R^9}\left(V_N(x,y,z)(Q_t^{\otimes 3}\partial_t k_t)(x,y,z)b_x^*b_y^*b_z^*\Theta_N + \hc\right),\\
	\tilde{\cI}_6^{(2)} &= - \dfrac{1}{2}\int_{\R^{12}}\left(V_N(x,y,z)(Q_t\otimes Q_t\otimes\mathds{1}\partial_tk_t)(x,y,z')b_x^*b_y^*b_z^*b_{z'}^*b_z\Theta_N + \hc\right),\\
	\tilde{\cI}_6^{(3)} &= -\dfrac{1}{12}\int_{\R^{15}}\left(V_N(x,y,z)(Q_t\otimes\mathds{1}\otimes\mathds{1}\partial_tk_t)(x,y',z')b_x^*b_y^*b_z^*b_{y'}^*b_{z'}^*b_yb_z\Theta_N + \hc\right).
\end{align*}
Thanks to \eqref{eq:renormalised_annihilation_op_three_rewritten}, we can rewrite $\tilde{\cI}_6^{(1)}$ as
\begin{multline*}
	\tilde{\cI}_6^{(1)} = \dfrac{1}{2}\int_{\R^9}(V_N\omega_N^2)(x,y,z)(\partial_t\varphi_t)(x)\overline{\varphi_t(x)}\vert\varphi_t(y)\vert^2\vert\varphi_t(z)\vert^2 + \hc\\
	-\dfrac{1}{2}\int_{\R^9}\left((V_N\omega_N)(x,y,z)(\partial_t\varphi_t)(x)\varphi_t(y)\varphi_t(z)d_{xyz}^* + \hc\right) + \tilde{\cE}_6^{(1)},
\end{multline*}
for some $\tilde{\cE}_6^{(1)}$ satisfying
\begin{equation}
	\label{eq:interaction_term_commutator_drift_error}
	\pm\tilde{\cE}_6^{(1)} \leq C(\cN^\ren + 1) + C\cV^\ren.
\end{equation}
To get rid of the projection $Q_t^{\otimes 3}$, we wrote $Q_t = 1 - \vert\varphi_t\rangle\langle\varphi_t\vert$ and estimated the $\vert\varphi_t\rangle\langle\varphi_t\vert$ part using that we essentially gain an integration over an extra variable from that projection.

Next, we treat $\tilde{\cI}_6^{(2)}$ and $\tilde{\cI}_6^{(3)}$. To this end, we use the identity \eqref{eq:renormalised_annihilation_op_three_rewritten}, and the estimates \eqref{eq:truncated_scattering_equation_Lp_estimates} and \eqref{eq:truncated_scattering_equation_LpLinf_estimates} to write
\begin{equation*}
	\tilde{\cI}_6^{(2)} = - \dfrac{1}{2}\int_{\R^{9}}\Big(V_N(x,y,z)d_{xyz}^*\Big(\int_{\R^3}\d{}z'(Q_t\otimes Q_t\otimes\mathds{1}\partial_tk_t)(x,y,z')b_{z'}^*\Big)b_z\Theta_N + \hc\Big) + \tilde{\cE}_6^{(2)},
\end{equation*}
for some error $\tilde{\cE}_6^{(2)}$ such that
\begin{equation*}
	\pm\tilde{\cE}_6^{(2)} \leq C\lambda\cN_+
\end{equation*}
Then, using the Cauchy--Schwarz inequality, the CCR \eqref{eq:creation_annihilation_op_modified_ccr} and the estimate \eqref{eq:truncated_scattering_equation_LpLinf_estimates}, we get
\begin{align*}
	\pm\tilde{\cI}_6^{(2)} &\leq
	\begin{multlined}[t]
		C\cV^\ren + C\int_{\R^{12}}V_N(x,y,z)\vert\partial_tk_t(x,y,z')\vert^2b_z^*b_z\\
		+ C\int_{\R^{15}}V_N(x,y,z)\partial_tk_t(x,y,z')\overline{\partial_tk_t(x,y,z'')}b_z^*b_{z'}^*b_{z''}b_z + C\lambda\cN_+
	\end{multlined}\\
	&\leq C\cV^\ren + CN^{-1/2}\cN_+ + CN^{3/2}\int_{\R^{15}}V_N(x,y,z)\omega_{\lambda,N}(x,y,z')b_z^*b_{z'}^*b_zb_{z'} + C\lambda\cN_+.
\end{align*}
Bounding the last term using Proposition~\ref{prop:general_three_and_four_body_bound_renormalised}, we obtain
\begin{equation*}
	\pm\tilde{\cI}_6^{(2)} \leq C(\cN^\ren + \cK^\ren + \cV^\ren + N^{1/2}).
\end{equation*}
To bound $\tilde{\cI}_6^{(3)}$, we use \eqref{eq:renormalised_annihilation_op_three_rewritten} and write
\begin{equation*}
	\tilde{\cI}_6^{(3)} = \tilde{\cJ}_6^{(1)} + \tilde{\cJ}_6^{(2)},
\end{equation*}
with $\tilde{\cJ}_6^{(1)}$ and $\tilde{\cJ}_6^{(2)}$ given by
\begin{equation*}
	\tilde{\cJ}_6^{(1)} = -\dfrac{1}{12}\Big(\int_{\R^{15}}V_N(x,y,z)(Q_t\otimes\mathds{1}\otimes\mathds{1}\partial_tk_t)(x,y',z')d_{xyz}^*b_{y'}^*b_{z'}^*b_yb_z\Theta_N + \hc\Big)
\end{equation*}
and
\begin{equation*}
	\tilde{\cJ}_6^{(2)} = \dfrac{1}{12}\int_{\R^{15}}(V_N\overline{k_t})(x,y,z)(Q_t\otimes\mathds{1}\otimes\mathds{1}\partial_tk_t)(x,y',z')b_{y'}^*b_{z'}^*b_yb_z\Theta_N + \hc
\end{equation*}
Thanks to the estimates \eqref{eq:truncated_scattering_equation_Lp_estimates}~and~\ref{eq:truncated_scattering_equation_LpLinf_estimates}, and the Cauchy--Schwarz inequality, the second term can be bounded by
\begin{equation*}
	\pm\tilde{\cJ}_6^{(2)} \leq C\lambda^2\cV_\twoB + CN^{-1/2}\tilde{\cV}_\twoB
\end{equation*}
with $\cV_\twoB$ and $\tilde{\cV}_\twoB$ defined as in Propositions~\ref{prop:general_two_body_bound_renormalised}~and~\ref{prop:general_three_and_four_body_bound_renormalised}. Hence applying these propositions, we deduce
\begin{equation*}
	\pm\tilde{\cJ}_6^{(2)} \leq CN^{-1/2}(\cN^\ren + \cK^\ren + \cV^\ren + N).
\end{equation*}
To bound $\tilde{\cJ}_6^{(1)}$, we write
\begin{equation*}
	\tilde{\cJ}_6^{(1)} = -\dfrac{1}{12}\Big(\int_{\R^9}V_N(x,y,z)d_{xyz}^*\Big(\int_{\R^6}\d{}y'\d{}z'(Q_t\otimes\mathds{1}\otimes\mathds{1}\partial_tk_t)(x,y',z')b_{y'}^*b_{z'}^*\Big)b_yb_z\Theta_N + \hc\Big)
\end{equation*}
and use the Cauchy--Schwarz inequality as well as the CCR \eqref{eq:creation_annihilation_op_modified_ccr} to find
\begin{multline*}
	\pm\tilde{\cJ}_6^{(1)} \leq C\cV^\ren + C\int_{\R^{15}}V_N(x,y,z)\vert\partial_tk_t(x,y',z')\vert^2b_y^*b_z^*b_yb_z\\
	+ C\int_{\R^{18}}V_N(x,y,z)(\partial_tk_t)(x,y',z')\overline{(\partial_tk_t)(x,y',z'')}b_y^*b_z^*b_{z''}^*b_{z'}b_yb_z\\
	+ C\int_{\R^{21}}V_N(x,y,z)(\partial_tk_t)(x,y',z')\overline{(\partial_tk_t)(x,y'',z'')}b_y^*b_z^*b_{y''}^*b_{z''}^*b_{y'}b_{z'}b_yb_z,
\end{multline*}
Using the estimates \eqref{eq:truncated_scattering_equation_Lp_estimates}~and~\ref{eq:truncated_scattering_equation_LpLinf_estimates}, this can be bounded further by
\begin{equation*}
	\pm\tilde{\cJ}_6^{(1)} \leq C(\cV^\ren + N^{-1}\cV_\twoB + N^{-1/2}\cV_\threeB + \cV_\fourB),
\end{equation*}
with $\cV_\threeB$ and $\cV_\fourB$ defined as in Proposition \ref{prop:general_three_and_four_body_bound_renormalised}. Therefore, an application of Propositions \ref{prop:general_two_body_bound_renormalised} and \ref{prop:general_three_and_four_body_bound_renormalised} yields
\begin{equation*}
	\pm\tilde{\cJ}_6^{(1)} \leq C(\cN^\ren + \cK^\ren + \cV^\ren + N^{1/2}).
\end{equation*}

Summing up, we have shown that
\begin{multline*}
	\ui\left[\cV_N,\cQ^\ren\right] = \dfrac{1}{2}\int_{\R^9}(V_N\omega_N^2)(x,y,z)(\partial_t\varphi_t)(x)\overline{\varphi_t(x)}\vert\varphi_t(y)\vert^2\vert\varphi_t(z)\vert^2 + \hc\\
	- \dfrac{1}{2}\int_{\R^9}\left((V_N\omega_N)(x,y,z)(\partial_t\varphi_t)(x)\varphi_t(y)\varphi_t(z)d_{xyz}^* + \hc\right) + \tilde{\cE}_6,
\end{multline*}
with $\tilde{\cE}_6$ satisfying
\begin{equation}
	\label{eq:L6_term_commutator_Q_ren_error}
	\pm\tilde{\cE}_6 \leq C(\cN^\ren + \cK^\ren + \cV^\ren + N^{1/2}).
\end{equation}
Gathering the previous estimates, we obtain
\begin{equation*}
	\partial_t\cL_3^{(1)} + \ui\left[\cK + \cV_N,\cQ^\ren\right] = \tilde{\cE}_\Delta + \tilde{\cE}_3^{(1)} + \tilde{\cE}_6,
\end{equation*}
with $\tilde{\cE}_3^{(1)}, \tilde{\cE}_\Delta$ and $\tilde{\cE}_6$ satisfying respectively \eqref{eq:cubic_term_derivative_error}, \eqref{eq:kinetic_term_commutator_Q_ren_error} and \eqref{eq:L6_term_commutator_Q_ren_error}. This finishes the proof of Lemma~\ref{lemma:generator_commutator_estimates_cubic_derivative}.

\subsubsection{Analysis of $[\cL_3^{(1)},\cQ^\ren]$: proof of Lemma~\ref{lemma:generator_commutator_estimates_cubic_commutator}}

Define
\begin{equation*}
	\tilde{\cL}_3^{(1)} = \dfrac{N^{3/2}}{6}\int_{\R^9}V_N(x,y,z)\overline{\varphi_t(x)\varphi_t(y)\varphi_t(z)}b_xb_yb_z
\end{equation*}
and
\begin{equation}
	\label{eq:renormalised_drift_operator_def_half}
	\tilde{Q}^\ren = \dfrac{1}{6}\int_{\R^9}\overline{(\ui\partial_tk_t)}(x,y,z)b_xb_yb_z.
\end{equation}
Then, we expand the commutator $\ui[\cL_3^{(1)},\cQ^\ren]$ as
\begin{equation*}
	\ui[\cL_3^{(1)},\cQ^\ren] = \tilde{\cI}_3^{(1)} + \tilde{\cI}_3^{(2)},
\end{equation*}
with $\tilde{\cI}_3^{(1)}$ and $\tilde{\cI}_3^{(2)}$ given by
\begin{equation*}
	\tilde{\cI}_3^{(1)} = \ui\Theta_3^{(1)}[\tilde{\cL}_3^{(1)},(\tilde{\cQ}^\ren)^*]\Theta_N + \hc
\end{equation*}
and
\begin{multline}
	\tilde{\cI}_3^{(2)} = \ui([\Theta_3^{(1)},(\tilde{\cQ}^\ren)^*]\Theta_N\tilde{\cL}_3^{(1)} + \Theta_3^{(1)}\tilde{\cQ}^\ren[\tilde{\cL}_3^{(1)},\Theta_N]) + \hc\\
	+ \ui(\Theta_N[\Theta_3^{(1)},\tilde{\cQ}^\ren]\tilde{\cL}_3^{(1)} + \Theta_3^{(1)}[\tilde{\cL}_3^{(1)},\Theta_N]\tilde{\cQ}^\ren) + \hc \label{eq:I_3_2_def}
\end{multline}
The reason we separated the terms this way is because $\tilde{\cI}_3^{(1)}$ contains a contribution of order $N$ that we need to extract, whereas $\tilde{\cI}_3^{(2)}$ contains only errors, which we bound directly. To extract the main contribution from $\tilde{\cI}_3^{(1)}$, we use the CCR \eqref{eq:creation_annihilation_op_modified_ccr} to expand further
\begin{equation*}
	\tilde{\cI}_3^{(1)} = \tilde{\cJ}_3^{(1)} + \tilde{\cJ}_3^{(2)} + \tilde{\cJ}_3^{(3)},
\end{equation*}
with $\tilde{\cJ}_3^{(1)}, \tilde{\cJ}_3^{(2)}$ and $\tilde{\cJ}_3^{(3)}$ given by
\begin{align*}
	\tilde{\cJ}_3^{(1)} &= -\Theta_3^{(1)}\Theta_N\dfrac{N^3}{6}\int_{\R^9}\left(V_N(x,y,z)(Q_t^{\otimes 3}k_t)(x,y,z)\overline{\varphi_t(x)\varphi_t(y)\varphi_t(z)} + \hc\right),\\
	\tilde{\cJ}_3^{(2)} &= -\Theta_3^{(1)}\Theta_N\dfrac{N^{3/2}}{2}\int_{\R^{12}}(V_N(x,y,z)(Q_t\otimes Q_t\otimes \mathds{1}\partial_tk_t)(x,y,z')\overline{\varphi_t(x)}\overline{\varphi_t(y)}\overline{\varphi_t(z)}b_{z'}^*b_z + \hc)
\end{align*}
and
\begin{multline*}
	\tilde{\cJ}_3^{(3)} = -\Theta_3^{(1)}\Theta_N\dfrac{N^{3/2}}{12}\int_{\R^{15}}(V_N(x,y,z)(Q_t\otimes\mathds{1}\otimes\mathds{1}\partial_tk_t)(x,y',z')\overline{\varphi_t(x)}\overline{\varphi_t(y)}\overline{\varphi_t(z)}b_{y'}^*b_{z'}^*b_yb_z\\
	+ \hc).
\end{multline*}
Since $\Theta_N$ and $\Theta_3^{(1)}$ satisfy \eqref{eq:theta_coef_cubic_transform_almost_1} and \eqref{eq:theta_coef_estimate_almost_1}, we may rewrite $\tilde{\cJ}_3^{(1)}$ as
\begin{equation*}
	\tilde{\cJ}_3^{(1)} = -\dfrac{N}{6}\int_{\R^9}(V_N\omega_N)(x,y,z)\partial_t(\vert\varphi_t(x)\vert^2\vert\varphi_t(y)\vert^2\vert\varphi_t(z)\vert^2) + \tilde{\cE}_3^{(1)},
\end{equation*}
with $\tilde{\cE}_3^{(1)}$ satisfying
\begin{equation*}
	\pm\tilde{\cE}_3^{(1)} \leq C(\cN_+ + 1).
\end{equation*}
Moreover, using the Cauchy--Schwarz inequality and the estimates \eqref{eq:truncated_scattering_equation_Lp_estimates} and \eqref{eq:truncated_scattering_equation_LpLinf_estimates}, we bound $\tilde{\cJ}_3^{(2)}$ by
\begin{equation*}
	\pm\tilde{\cJ}_3^{(2)} \leq C\lambda \cN_+.
\end{equation*}
The term $\tilde{\cJ}_3^{(3)}$ is bounded just as $\tilde{\cJ}_6^{(2)}$, and we find
\begin{equation*}
	\pm\tilde{\cJ}_3^{(3)} \leq CN^{-1/2}(\cN^\ren + \cK^\ren + \cV^\ren + N).
\end{equation*}
Hence, so far we have shown that
\begin{equation*}
	\ui\big[\cL_3^{(1)},\cQ^\ren\big] = - \dfrac{N}{6}\int_{\R^9}(V_N\omega_N)(x,y,z)\partial_t(\vert\varphi_t(x)\vert^2\vert\varphi_t(y)\vert^2\vert\varphi_t(z)\vert^2) + \tilde{\cE}_3 + \tilde{\cI}_3^{(2)},
\end{equation*}
for some $\tilde{\cE}_3$ satisfying
\begin{equation}
	\label{eq:cubic_term_commutator_Q_ren_error}
	\pm\tilde{\cE}_3 \leq C\cN^\ren + CN^{-1/2}(\cK^\ren + \cV^\ren + N),
\end{equation}
and with $\tilde{\cI}_3^{(2)}$ given by \eqref{eq:I_3_2_def}.

Let us now explain briefly how to bound $\tilde{\cI}_3^{(2)}$. We focus on $\ui[\Theta_3^{(1)},(\tilde{\cQ}^\ren)^*]\Theta_N\tilde{\cL}_3^{(1)}$ for clarity. Using that $\tilde{\cQ}^\ren$ contains three creation operators and no annihilation operators, we rewrite the previous commutator as
\begin{equation*}
	[\Theta_3^{(1)},(\tilde{\cQ}^\ren)^*] = (\tilde{\cQ}^\ren)^*(\Theta_3^{(1)}(\cN_+ + 3) - \Theta_3^{(1)}(\cN_+)).
\end{equation*}
The important point here is that, as a result of \eqref{eq:theta_coef_estimate_difference}, we gain a factor $N^{-1}$ from the prefactors $\Theta_3^{(1)}$ and $\Theta_N$. This allows us to write
\begin{multline*}
	\ui[\Theta_3^{(1)},(\tilde{\cQ}^\ren)^*]\Theta_N\tilde{\cL}_3^{(1)} + \hc\\
	= -\dfrac{N^{3/2}}{36}\int_{\R^{18}}(V_N(x,y,z)(\partial_tk_t)(x',y',z')\overline{\varphi_t(x)\varphi_y(t)\varphi_t(z)}\xi_N(\cN_+)b_{x'}^*b_{y'}^*b_{z'}^*b_xb_yb_z + \hc),
\end{multline*}
for some function $\xi_N:\{0,\dots,N\}\rightarrow[0,1]$ that satisfies
\begin{equation*}
	\vert\xi_N\vert \leq CN^{-1}.
\end{equation*}
Then, we define
\begin{equation*}
	A_{x'y'} = \int_{\R^3}\d{}z'(\partial_tk_t)(x',y',z')b_{z'}^*
\end{equation*}
and use the Cauchy--Schwarz inequality to deduce
\begin{multline*}
	\pm(\ui[\Theta_3^{(1)},(\tilde{\cQ}^\ren)^*]\Theta_N\tilde{\cL}_3^{(1)} + \hc) \leq \varepsilon CN^{1/2}\int_{\R^{15}}V_N(x,y,z)b_x^*b_y^*b_z^*A_{x'y'}A_{x'y'}^*b_xb_yb_z \\
	+ \varepsilon^{-1} CN^{1/2}\int_{\R^{15}}V_N(x,y,z)\vert\varphi_t(x)\vert^2\vert\varphi_t(y)\vert^2\vert\varphi_t(z)\vert^2\1(\vert x' - y'\vert \leq C\lambda)b_{x'}^*b_{y'}^*b_{x'}b_{y'},
\end{multline*}
for all $\varepsilon > 0$. Bounding $A_{x'y'}A_{x'y'}^*$ using Lemma~\ref{lemma:three_body_kernel_second_quantisation_estimates}, applying Proposition~\ref{prop:general_three_and_four_body_bound_renormalised}, and optimising over $\varepsilon$ we obtain
\begin{equation*}
	\pm(\ui[\Theta_3^{(1)},(\tilde{\cQ}^\ren)^*]\Theta_N\tilde{\cL}_3^{(1)} + \hc) \leq CN^{-1/2}(\cN^\ren + \cK^\ren + \cV^\ren + N).
\end{equation*}
The other three terms are bounded in the same manner and we leave the detail to the reader. Finally, we have shown that
\begin{equation}
	\label{eq:cubic_term_commutator_Q_ren_error_second}
	\pm\tilde{\cI}_3^{(2)} \leq CN^{-1/2}(\cN^\ren + \cK^\ren + \cV^\ren + N).
\end{equation}

Therefore,
\begin{equation*}
	\partial_t\cL_0 + \ui\big[\cL_3^{(1)},\cQ^\ren\big] = \dfrac{N^3}{6}\int_{\R^9}(V_Nf_N)(x,y,z)\partial_t(\vert\varphi_t(x)\vert^2\vert\varphi_t(y)\vert^2\vert\varphi_t(z)\vert^2) + \tilde{\cE}_3 + \tilde{\cI}_3^{(2)}
\end{equation*}
with $\tilde{\cE}_3$ and $\tilde{\cI}_3^{(2)}$ satisfying respectively \eqref{eq:cubic_term_commutator_Q_ren_error} and \eqref{eq:cubic_term_commutator_Q_ren_error_second}. This concludes the proof of Lemma \ref{lemma:generator_commutator_estimates_cubic_commutator}.

\subsubsection{Analysis of $\partial_t\cL_1, \partial_t\cL_4^{(1)}$ and $[\cL_4^{(1)},\cQ^\ren]$: proof of Lemma~\ref{lemma:generator_commutator_estimates_drift}} Thanks to \eqref{eq:theta_coef_estimate_almost_1} and \eqref{eq:modified_creation_annihilation_op_derivative} we have
\begin{multline*}
	\partial_t\cL_1 = \dfrac{N^{5/2}}{2}\int_{\R^9}V_N(x,y,z)\partial_t(\varphi_t(x)\vert\varphi_t(y)\vert^2\vert\varphi_t(z)\vert^2)b_x^* + \hc\\
	- \dfrac{N^{5/2}}{2}\int_{\R^9}\left(V_N(x,y,z)\vert\varphi_t(x)\vert^2\vert\varphi_t(y)\vert^2\vert\varphi_t(z)\vert^2b^*(\partial_t\varphi_t) + \hc\right) + \tilde{\cE}_1,
\end{multline*}
with $\tilde{\cE}_1$ satisfying
\begin{equation*}
	\pm\tilde{\cE}_1 \leq C\cN_+.
\end{equation*}

In a similar way, using additionally the identity \eqref{eq:renormalised_annihilation_op_three_rewritten} and that $\Theta_4^{(1)}$ satisfies \eqref{eq:theta_coef_estimate_almost_1}, we find
\begin{multline*}
	\partial_t\cL_4^{(1)} = - \dfrac{N}{2}\int_{\R^9}\left((V_Nk_t)(x,y,z)\partial_t(\overline{\varphi_t(y)\varphi_t(z)})b_x^* + \hc\right)\\
	+ \dfrac{N^{5/2}}{2}\int_{\R^9}\left((V_N\omega_N)(x,y,z)\vert\varphi_t(x)\vert^2\vert\varphi_t(y)\vert^2\vert\varphi_t(z)\vert^2b^*(\partial_t\varphi_t) + \hc\right) + \tilde{\cE}_4^{(1)},
\end{multline*}
for some $\tilde{\cE}_4^{(1)}$ that satisfies
\begin{equation*}
	\pm\tilde{\cE}_4^{(1)} \leq C\cN_+ + C\cV^\ren.
\end{equation*}

We now deal with $[\cL_4^{(1)},\cQ^\ren]$. Define
\begin{equation*}
	\tilde{\cL}_4^{(1)} = \dfrac{N}{2}\int_{\R^9}V_N(x,y,z)\overline{\varphi_t(y)\varphi_t(z)}b_x^*b_xb_yb_z
\end{equation*}
and $\tilde{\cQ}^\ren$ as in \eqref{eq:renormalised_drift_operator_def_half}. Then,
\begin{equation*}
	[\cL_4^{(1)},\cQ^\ren] = [\Theta_4^{(1)}\tilde{\cL}_4^{(1)},(\tilde{\cQ}^\ren)^*\Theta_N] + [\Theta_4^{(1)}\tilde{\cL}_4^{(1)},\Theta_N\tilde{Q}^\ren] - \hc
\end{equation*}
The only commutator that we cannot bound directly and that requires identifying a cancellation with $\partial_t\cL_1$ and $\partial_t\cL_4^{(1)}$ is $[\tilde{\cL}_4^{(1)},(\tilde{\cQ}^\ren)^*]$; all the others are errors. For that reason, we split $\ui[\tilde{\cL}_4^{(1)},\tilde{\cQ}^\ren]$ in the following way:
\begin{equation*}
	\ui[\tilde{\cL}_4^{(1)},\tilde{\cQ}^\ren] = \tilde{\cI}_4^{(1)} + \tilde{\cI}_4^{(2)},
\end{equation*}
with $\tilde{\cI}_4^{(1)}$ and $\tilde{\cI}_4^{(2)}$ given by
\begin{equation*}
	\tilde{\cI}_4^{(1)} = \ui\Theta_4^{(1)}[\tilde{\cL}_4^{(1)},(\tilde{\cQ}^\ren)^*]\Theta_N + \hc
\end{equation*}
and
\begin{equation}
	\label{eq:I_4_2}
	\tilde{\cI}_4^{(2)} = \ui([\Theta_4^{(1)}\tilde{\cL}_4^{(1)},\Theta_N\tilde{Q}^\ren] + \Theta_4^{(1)}(\tilde{\cQ}^\ren)^*[\tilde{\cL}_4^{(1)},\Theta_N] + [\Theta_4^{(1)},(\tilde{\cQ}^\ren)^*]\Theta_N\tilde{\cL}_4^{(1)}) + \hc
\end{equation}
First, we analyse $\tilde{\cI}_4^{(1)}$. Using the CCR \eqref{eq:creation_annihilation_op_modified_ccr}, we expand
\begin{equation*}
	\tilde{\cI}_4^{(1)} = \tilde{\cJ}_4^{(1)} + \tilde{\cJ}_4^{(2)} + \tilde{\cJ}_4^{(3)} + \tilde{\cJ}_4^{(4)} + \tilde{\cJ}_4^{(5)},
\end{equation*}
with $\tilde{\cJ}_4^{(1)}, \tilde{\cJ}_4^{(2)}, \tilde{\cJ}_4^{(3)}, \tilde{\cJ}_4^{(4)}$ and $\tilde{\cJ}_4^{(5)}$ given by
\begin{align*}
	\tilde{\cJ}_4^{(1)} &= -\dfrac{N}{2}\int_{\R^9}(V_N(x,y,z)(Q_t^{\otimes 3}\partial_tk_t)(x,y,z)\overline{\varphi_t(y)\varphi_t(z)}\Theta_4^{(1)}b_x^*\Theta_N + \hc),\\
	\tilde{\cJ}_4^{(2)} &= -N\int_{\R^{12}}V_N(x,y,z)(Q_t\otimes Q_t\otimes\mathds{1}\partial_tk_t)(x,y,z')\overline{\varphi_t(y)\varphi_t(z)}\Theta_4^{(1)}b_x^*b_{z'}^*b_x\Theta_N + \hc),\\
	\tilde{\cJ}_4^{(3)} &= -\dfrac{N}{2}\int_{\R^{12}}(V_N(x,y,z)(\mathds{1}\otimes Q_t\otimes Q_t\partial_tk_t)(x',y,z)\overline{\varphi_t(y)\varphi_t(z)}\Theta_4^{(1)}b_x^*b_{x'}^*b_x\Theta_N + \hc),\\
	\tilde{\cJ}_4^{(4)} &= -\dfrac{N}{4}\int_{\R^{15}}(V_N(x,y,z)(Q_t\otimes\mathds{1}\otimes\mathds{1}\partial_tk_t)(x,y',z')\overline{\varphi_t(y)\varphi_t(z)}\Theta_4^{(1)}b_x^*b_{y'}^*b_{z'}^*b_yb_z\Theta_N + \hc)
\end{align*}
and
\begin{equation*}
	\tilde{\cJ}_4^{(5)} = -\dfrac{N}{2}\int_{\R^{15}}(V_N(x,y,z)(\mathds{1}\otimes\mathds{1}\otimes Q_t\partial_tk_t)(x',y',z)\overline{\varphi_t(y)\varphi_t(z)}\Theta_4^{(1)}b_x^*b_{x'}^*b_{y'}^*b_xb_y\Theta_N + \hc).
\end{equation*}
Writing $Q_t = 1 - \vert\varphi_t\rangle\langle\varphi_t\vert$ and using the fact that $\Theta_N$ and $\Theta_4^{(1)}$ satisfy \eqref{eq:theta_coef_cubic_transform_almost_1}, it is easy to see that
\begin{equation*}
	\tilde{\cJ}_4^{(1)} = -\dfrac{N}{2}\int_{\R^9}((V_N\partial_tk_t)(x,y,z)\overline{\varphi_t(y)\varphi_t(z)}b_x^* + \hc) + \tilde{\cE}_4^{(1,1)},
\end{equation*}
with $\tilde{\cE}_4^{(1,1)}$ satisfying
\begin{equation*}
	\pm\tilde{\cE}_4^{(1,1)} \leq CN^{-1/2}.
\end{equation*}
To bound $\tilde{\cJ}_4^{(2)}$, we notice that it is only nonzero when $\vert x - z'\vert \leq C\lambda$. This allows to use the first estimate in \eqref{eq:general_two_three_four_body_bound_non_renormalised} and the bounds \eqref{eq:truncated_scattering_equation_LpLinf_estimates} to deduce
\begin{equation*}
	\pm\tilde{\cJ}_4^{(2)} \leq C\cN_+ + CN^{-1/2}(\cK^\ren + \cV^\ren + N),
\end{equation*}
and likewise for $\tilde{\cJ}_4^{(3)}$. To bound $\tilde{\cJ}_4^{(4)}$, we define
\begin{equation*}
	B_x^* = \int_{\R^6}\d{}y'\d{}z'\ui\partial_tk_t(x,y',z')b_{y'}^*b_{z'}^*
\end{equation*}
and use the Cauchy--Schwarz inequality as well as $\Vert Q_t\Vert_\op = 1$ to obtain
\begin{equation*}
	\pm\tilde{\cJ}_4^{(4)} \leq \varepsilon^{-1} CN\int_{\R^9}V_N(x,y,z)b_x^*b_x + \varepsilon CN\int_{\R^9}V_N(x,y,z)b_y^*b_z^*B_xB_x^*b_yb_z,
\end{equation*}
for all $\varepsilon > 0$. The first term simplifies to $CN^{-1}\cN_+$. Regarding the second one, we use the CCR \eqref{eq:creation_annihilation_op_modified_ccr} to put the creation and annihilation in normal order, and then we apply the Cauchy--Schwarz inequality and use the estimate \eqref{eq:truncated_scattering_equation_Lp_estimates} on the $L^2$ norm of $\omega_{\lambda,N}$ to find
\begin{equation*}
	N\int_{\R^9}V_N(x,y,z)b_y^*b_z^*B_xB_x^*b_yb_z \leq C(\cV_\twoB + N^{1/2}\cV_\threeB + N\cV_\fourB).
\end{equation*}
Taking $\varepsilon = N^{-1}$ and applying Propositions~\ref{prop:general_two_body_bound_renormalised}~and~\ref{prop:general_three_and_four_body_bound_renormalised} yields
\begin{equation*}
	\pm\tilde{\cJ}_4^{(4)} \leq C(\cN^\ren + \cK^\ren + \cV^\ren + N^{1/2}).
\end{equation*}
The term $\tilde{\cJ}_4^{(5)}$ is bounded exactly the same way. Summing up, we have shown that
\begin{equation*}
	\ui\big[\cL_4^{(1)},\cQ^\ren\big] = -\dfrac{N^{5/2}}{2}\int_{\R^9}\left((V_N\omega_N)(x,y,z)\partial_t(\varphi_t(x)\vert\varphi_t(y)\vert^2\vert\varphi_t(z)\vert^2)b_x^* + \hc\right) + \tilde{\cE}_4^{(1)} + \tilde{\cI}_4^{(2)},
\end{equation*}
for some $\tilde{\cE}_4^{(1)}$ satisfying
\begin{equation*}
	\pm\tilde{\cE}_4^{(1)} \leq C(\cN^\ren + \cK^\ren + \cV^\ren + N^{1/2}),
\end{equation*}
and with $\tilde{\cI}_4^{(2)}$ defined in \eqref{eq:I_4_2}. To bound $\tilde{\cI}_4^{(2)}$, we follow the same steps as for $\tilde{\cI}_3^{(2)}$ (defined in \eqref{eq:I_3_2_def}). The last two commutators in \eqref{eq:I_4_2} are bounded in exactly the same way, but the first one is sightly more subtle. Namely, 
we need to compute $[\tilde{\cL}_4^{(1)},\tilde{\cQ}^\ren]$. This is however fairly simple, and bounding the error as in the proof of Lemma~\ref{lemma:number_operator_renormalised}, we find
\begin{equation*}
	\pm\tilde{\cI}_4^{(2)} \leq C(\cN^\ren + \cK^\ren + \cV^\ren + N^{1/2}).
\end{equation*}

Finally,
\begin{multline*}
	\partial_t\cL_1 + \partial_t\cL_4^{(1)} + \ui\big[\cL_4^{(1)},\cQ^\ren\big] = \dfrac{N^{5/2}}{2}\int_{\R^9}\left((V_Nf_N)(x,y,z)\partial_t(\varphi_t(x)\vert\varphi_t(y)\vert^2\vert\varphi_t(z)\vert^2)b_x^* + \hc\right)\\
	- \dfrac{N^{5/2}}{2}\int_{\R^9}\left((V_Nf_N)(x,y,z)\vert\varphi_t(x)\vert^2\vert\varphi_t(y)\vert^2\vert\varphi_t(z)\vert^2b^*(\partial_t\varphi_t) + \hc\right) + \tilde{\cE}_1 + \tilde{\cE}_4^{(1)} + \tilde{\cI}_4^{(2)}.
\end{multline*}

\subsubsection{Bounding the remaining commutators}

\label{subsec:generator_commutator_estimates_errors}

The only terms we are left with are the commutators of $\cQ^\ren$ with $\cL_2^{(1)}, \cL_2^{(2)}, \cL_2^{(3)}, \cL_3^{(2)}, \cL_3^{(3)}, \cL_4^{(2)}, \cL_4^{(3)}$ and $\cL_5$, which are all bounded in a similar fashion to $\tilde{\cI}_3^{(2)}$ from the proof of Lemma~\ref{lemma:generator_commutator_estimates_cubic_commutator} (see \eqref{eq:cubic_term_commutator_Q_ren_error_second} and above). Let us explain this briefly. All these terms can be written in the form
\begin{equation*}
	\cL_i^{(j)} = \Theta_i^{(j)}\tilde{\cL}_i^{(j)} + \hc
\end{equation*}
where $\Theta_i^{(j)}$ is the coefficient in the expression below \eqref{eq:hamiltonian_conjugation_excitation_map}. Similarly, $\cQ^\ren$ can be written as
\begin{equation*}
	\cQ^\ren = \Theta_N\tilde{\cQ}^\ren + \hc
\end{equation*}
(see \eqref{eq:renormalised_drift_operator_def_half}). Then, when expanding the commutator of $\cQ^\ren$ and $\cL_i^{(j)}$, we can essentially distinguish between two situations. Either we get a commutator between creation and annihilation operators, e.g. between $\tilde{\cQ}^\ren$ and $\tilde{\cL}_i^{(j)}$, or a commutator involving the prefactors $\Theta_i^{(j)}$ and $\Theta_N$, e.g. between $\tilde{\cQ}^\ren$ and $\Theta_i^{(j)}$. In both cases it is not very difficult to bound these terms using the estimates from Lemma~\ref{lemma:truncated_three_body_scattering_solution} and Propositions~\ref{prop:general_two_body_bound_renormalised}~and~\ref{prop:general_three_and_four_body_bound_renormalised}. In the second case we additionally need to use that we gain a factor $N^{-1}$ from the coefficients $\Theta_N$ and $\Theta_i^{(j)}$ due to \eqref{eq:theta_coef_estimate_difference} and \eqref{eq:theta_coef_cubic_transform_difference}. All and all, we find that these coefficients are all bounded by
\begin{equation*}
	C(\cN^\ren + \cK^\ren + \cV^\ren + N^{1/2}).
\end{equation*}

\subsubsection{Conclusion of the proof of Lemma~\ref{lemma:generator_commutator_estimates}}  Collecting the previous estimates, we obtain
\begin{multline*}
	\partial_t(\cH_N + \cQ^\ren) + \ui[\cH_N,\cQ^\ren] = \dfrac{N^3}{6}\int_{\R^9}(V_Nf_N)(x,y,z)\partial_t(\vert\varphi_t(x)\vert^2\vert\varphi_t(y)\vert^2\vert\varphi_t(z)\vert^2)\\
	\begin{aligned}[t]
		&+ \dfrac{N^{5/2}}{2}\int_{\R^9}(V_Nf_N)(x,y,z)\partial_t(\varphi_t(x)\vert\varphi_t(y)\vert^2\vert\varphi_t(z)\vert^2)b_x^* + \hc\\
		&- \dfrac{N^{5/2}}{2}\int_{\R^9}\left((V_Nf_N)(x,y,z)\vert\varphi_t(x)\vert^2\vert\varphi_t(y)\vert^2\vert\varphi_t(z)\vert^2b^*(\partial_t\varphi_t) + \hc\right)\\
		&+ \dfrac{N^{1/2}}{2}b(V)\Vert\varphi_t\Vert_6^6\left(b^*(\partial_t\varphi_t) + \hc\right) - \dfrac{N}{6}b(V)\partial_t\Vert\varphi_t\Vert_6^6\\
		&- \dfrac{N^{1/2}}{2}b(V)(b(\partial_t(\vert\varphi_t\vert^4\varphi_t)) + \hc) + \tilde{\cE},
	\end{aligned}
\end{multline*}
for some $\tilde{\cE}$ satisfying
\begin{equation*}
	\pm\tilde{\cE} \leq C(\cN^\ren + \cK^\ren + \cV^\ren + N^{1/2}).
\end{equation*}
Finally, the bound \eqref{eq:generator_commutator_Qren_estimate} follows from the estimates
\begin{equation}
	\label{eq:constant_contribution_close_GP_time_derivative}
	\pm\Big(\dfrac{N^3}{6}\int_{\R^9}(V_Nf_N)(x,y,z)\partial_t(\vert\varphi_t(x)\vert^2\vert\varphi_t(y)\vert^2\vert\varphi_t(z)\vert^2) - \dfrac{N}{6}b(V)\partial_t\Vert\varphi_t\Vert_6^6\Big) \leq CN^{1/2}
\end{equation}
and
\begin{multline}
	\label{eq:linear_terms_close_time_derivative}
	\pm\Big(\dfrac{N^{5/2}}{2}\int_{\R^9}(V_Nf_N)(x,y,z)\partial_t(\varphi_t(x)\vert\varphi_t(y)\vert^2\vert\varphi_t(z)\vert^2)b_x^* - \dfrac{N^{1/2}}{2}b(V)b^*(\partial_t(\vert\varphi_t\vert^4\varphi_t)) + \hc\Big)\\
	\leq C(\cN_+ + 1).
\end{multline}
The estimate \eqref{eq:constant_contribution_close_GP_time_derivative} follows directly from \eqref{eq:constant_contribution_close_GP_Linf_estimate}, and the estimate \eqref{eq:linear_terms_close_time_derivative} follows from \eqref{eq:constant_contribution_close_GP_Linf_estimate} and
\begin{equation*}
	\Big\Vert N^2\int_{\R^6}\d{}y\d{}z(V_Nf_N)(\cdot,y,z)(\partial_t\vert\varphi_t(y)\vert^2)\vert\varphi_t(z)\vert^2 - b(V)(\partial_t\vert\varphi_t(\cdot)\vert^2)\vert\varphi_t(\cdot)\vert^2\Big\Vert_{L^\infty(\R^3)} \leq CN^{-1/2},
\end{equation*}
whose proof is analogous to that of \eqref{eq:constant_contribution_close_GP_Linf_estimate}.

The proof of \eqref{eq:generator_commutator_Nren_estimate} is very similar and we leave the detail to the reader. This finishes the proof of Lemma~\ref{lemma:generator_commutator_estimates}.

\subsection{Proof of Proposition~\ref{prop:gronwall_bound}}

Notice first that Lemma \ref{lemma:generator_estimate} implies
\begin{equation*}
	\dfrac{1}{2}\cK^\ren + \dfrac{1}{2}\cV^\ren - C(\cN^\ren + N^{1/2}) \leq \cH_N \leq 2\cK^\ren + 2\cV^\ren + C(\cN^\ren + N^{1/2}).
\end{equation*}
Then, use that $\cK^\ren$ and $\cV^\ren$ are nonnegative operators, and apply Lemma~\ref{lemma:number_operator_renormalised} again to deduce
\begin{equation*}
	\cN_+ \leq \cH_N + \cQ^\ren + C(\cN^\ren + N^{1/2}),
\end{equation*}
which is precisely \eqref{eq:number_excitations_bound}. Next, we prove the Grönwall bound \eqref{eq:gronwall_estimate_main}. Use that $\Phi_{N,t}$ solves \eqref{eq:fluctuation_vector_equation} to evaluate
\begin{multline*}
	\partial_t\langle\cH_N + \cQ^\ren + C(\cN^\ren + N^{1/2})\rangle_{\Phi_{N,t}}\\
	= \partial_t\langle\cH_N + \cQ^\ren\rangle_{\Phi_{N,t}} + C\langle\ui[\cN^\ren,\cH_N]\rangle_{\Phi_{N,t}} + C\langle\partial_t\cN^\ren\rangle_{\Phi_{N,t}} + CN^{1/2}.
\end{multline*}
Thus, an application of Lemma~\ref{lemma:generator_commutator_estimates} yields
\begin{equation*}
	\partial_t\langle\cH_N + \cQ^\ren + C(\cN^\ren + N^{1/2})\rangle_{\Phi_{N,t}}	\leq c\langle\cH_N + \cQ^\ren + C(\cN^\ren + N^{1/2})\rangle_{\Phi_{N,t}},
\end{equation*}
which concludes the proof of Proposition~\ref{prop:gronwall_bound}.

\section{Estimates on effective interaction potentials}

\label{sec:effective_interaction_potentials_non_renormalised}

This section is devoted to the proof of the bounds from Propositions~\ref{prop:general_two_body_bound_renormalised}~and~\ref{prop:general_three_and_four_body_bound_renormalised} without renormalised quantities, namely \eqref{eq:general_two_body_bound_non_renormalised} and \eqref{eq:general_two_three_four_body_bound_non_renormalised}. The other estimates are proven by conjugating these bounds with an appropriate cubic Bogoliubov transformation; this is done in Section~\ref{sec:effective_interaction_potentials_renormalised}.

Our starting point is the observation that, for any potential $V$ satisfying Assumption \ref{ass:potential}, the function $\tilde{V}:x\in\R^3\mapsto\Vert V(x,\cdot)\Vert_{L^1(\R^3)}$ is compactly supported and bounded; this is proved in Section \ref{sec:mixed_norm_regularity_estimate}. Hence, the estimates \eqref{eq:general_two_body_bound_non_renormalised} and \eqref{eq:general_two_three_four_body_bound_non_renormalised} boil down to estimates on indicator functions (recall that $V$ has a positive core by assumption). In view of that, the key element of the proof of \eqref{eq:general_two_body_bound_non_renormalised} is the following bound, which estimates a two-body interaction potential in terms of the number operator, the kinetic operator and a three-body interaction potential.

\begin{proposition}[Two-body estimate]
	\label{prop:general_two_body_potential_bound}
	Let $\ell_1,\ell_2,\ell_3 > 0$ such that $\ell_1 \geq \ell_2$ and $\ell_1 \geq \ell_3$. Define
	\begin{equation*}
		\cV_\twoB^{(\ell_2)} = \int_{\R^6}\mathds{1}(|x-y|\leq \ell_2)a_x^*a_y^*a_xa_y \quad \textmd{and} \quad \cV_\threeB^{(\ell_3)} = \int_{\R^9}\mathds{1}(\vert(x-y,x-z)\vert\leq\ell_3)a_x^*a_y^*a_z^*a_xa_ya_z.
	\end{equation*}	
	Then, there exists a universal constant $C$ such that the inequality
	\begin{equation}
		\label{eq:general_two_body_potential_bound}
		\cV_\twoB^{(\ell_2)} \leq C\Big(\frac{\ell_2}{\ell_1}\Big)^3\mathcal{N} + C\ell_2^2\Big(1 + \dfrac{\ell_2}{\ell_3}\Big)\Big(\frac{\ell_1}{\ell_3}\Big)^3\d{}\Gamma(-\Delta) + C\Big(1 + \Big(\frac{\ell_2}{\ell_3}\Big)^3\Big)\Big(\frac{\ell_1}{\ell_3}\Big)^3\cV_\threeB^{(\ell_3)}
	\end{equation}
	holds on $\mathcal{F}$.
\end{proposition}

The proof of Proposition \ref{prop:general_two_body_potential_bound} is postponed to Section \ref{sec:general_two_body_potential_bound}. Furthermore, as a consequence of Proposition~\ref{prop:general_two_body_potential_bound}, we have the following two results, which are proved in Section \ref{subsec:proof_non_renormalasied_estimates_3B_4B}.

\begin{corollary}[Three-body estimate]
	\label{cor:general_three_body_potential_bound}
	Let $\ell_1,\ell_2,\ell_3 > 0$ such that $\ell_1 \geq \ell_2$ and $\ell_1 \geq \ell_3$. Define
	\begin{equation*}
		\tilde{\cV}_\threeB^{(\ell_1,\ell_2)} = \int_{\R^{9}}\mathds{1}(\vert x - y\vert \leq \ell_2)\mathds{1}(\vert x - z\vert \leq \ell_1)a_x^*a_y^*a_z^*a_xa_ya_z
	\end{equation*}
	and
	\begin{equation*}
		\cV_\threeB^{(\ell_3)} = \int_{\R^9}\mathds{1}(\vert(x-y,x-z)\vert\leq\ell_3)a_x^*a_y^*a_z^*a_xa_ya_z.
	\end{equation*}
	Then, there exists a universal constant $C > 0$ such that the inequality
	\begin{equation}
		\label{eq:general_three_body_potential_bound}
		\tilde{\cV}_\threeB^{(\ell_1,\ell_2)} \leq C\Big(\dfrac{\ell_2}{\ell_3}\Big)^3\mathcal{N} + C\ell_2^2\Big(1 + \dfrac{\ell_2}{\ell_3}\Big)\Big(\dfrac{\ell_1}{\ell_3}\Big)^6\d{}\Gamma(-\Delta) + C\Big(1 + \Big(\dfrac{\ell_2}{\ell_3}\Big)^3\Big)\Big(\dfrac{\ell_1}{\ell_3}\Big)^6\cV_\threeB^{(\ell_3)}
	\end{equation}
	holds on $\mathcal{F}$.
\end{corollary}

\begin{corollary}[Four-body estimate]
	\label{cor:general_four_body_potential_bound}
	Let $\ell_1,\ell_2,\ell_3 > 0$ such that $\ell_1 \geq \ell_2$ and $\ell_1 \geq \ell_3$. Define
	\begin{equation*}
		\cV_\fourB^{(\ell_1,\ell_2)} = \int_{\R^{12}}\mathds{1}(\vert x - y\vert \leq \ell_2)\mathds{1}(\vert x - z\vert \leq \ell_1)\mathds{1}(\vert x - u\vert \leq \ell_1)a_x^*a_y^*a_z^*a_u^*a_xa_ya_za_u
	\end{equation*}
	and
	\begin{equation*}
		\cV_\threeB^{(\ell_3)} = \int_{\R^9}\mathds{1}(\vert(x-y,x-z)\vert\leq\ell_3)a_x^*a_y^*a_z^*a_xa_ya_z.
	\end{equation*}
	Then, there exists a universal constant $C > 0$ such that the inequality
	\begin{equation}
		\label{eq:general_four_body_potential_bound}
		\cV_\fourB^{(\ell_1,\ell_2)} \leq C\Big(\dfrac{\ell_2}{\ell_3}\Big)^3\Big(\dfrac{\ell_1}{\ell_3}\Big)^3\mathcal{N} + C\dfrac{\ell_2^3}{\ell_3}\Big(\dfrac{\ell_1}{\ell_3}\Big)^9\d{}\Gamma(-\Delta) + C\Big(\dfrac{\ell_2}{\ell_3}\Big)^3\Big(\dfrac{\ell_1}{\ell_3}\Big)^3\Big(\mathcal{N} + \Big(\frac{\ell_1}{\ell_3}\Big)^6\Big)\cV_\threeB^{(\ell_3)}
	\end{equation}
	holds on $\mathcal{F}$.
\end{corollary}

The bound \eqref{eq:general_two_body_bound_non_renormalised} follows from an application of Proposition \ref{prop:general_two_body_potential_bound} on the space $\cF_+(t)$ with $\ell_1 = (\varepsilon N^{-1})^{1/3}$, and $\ell_2$ and $\ell_3$ of order $N^{-1/2}$. Likewise, the bounds \eqref{eq:general_two_three_four_body_bound_non_renormalised} immediately follow from Proposition \ref{prop:general_two_body_potential_bound} and Corollaries \ref{cor:general_three_body_potential_bound} and \ref{cor:general_four_body_potential_bound}.

Lastly, using arguments similar to those of the proof of Proposition~\ref{prop:general_two_body_potential_bound}, we will prove the following bound, which will be useful when proving \eqref{eq:general_four_body_bound_renormalised}; the proof is given in Section \ref{subsec:proof_non_renormalasied_estimates_3B_4B}.

\begin{lemma}[Five-body estimate]
	\label{lemma:general_five_body_potential_bound}
	Let $\ell_1,\ell_2 > 0$ such that $\ell_1 \geq \ell_2$. Define
	\begin{equation*}
		\cV_\fiveB^{(\ell_1)} = \int_{\R^{15}}\mathds{1}(\vert x - y\vert \leq \ell_1)\mathds{1}(\vert x - z\vert \leq \ell_1)\mathds{1}(\vert x - u\vert \leq \ell_1)\mathds{1}(\vert x - v\vert \leq \ell_1)a_x^*a_y^*a_z^*a_u^*a_v^*a_xa_ya_za_ua_v
	\end{equation*}
	and
	\begin{equation*}
		\cV_\fourB^{(\ell_1,\ell_2)} = \int_{\R^{12}}\mathds{1}(\vert x - y\vert \leq \ell_2)\mathds{1}(\vert x - z\vert \leq \ell_1)\mathds{1}(\vert x - u\vert \leq \ell_1)a_x^*a_y^*a_z^*a_u^*a_xa_ya_za_u.
	\end{equation*}
	Then, there exists a universal constant $C > 0$ such that the inequality
	\begin{equation}
		\label{eq:general_five_body_potential_bound}
		\cV_\fiveB^{(\ell_1)} \leq C\Big(\dfrac{\ell_1}{\ell_2}\Big)^{12}\cN + C\cN\Big(\dfrac{\ell_1}{\ell_2}\Big)^3\cV_\fourB^{(\ell_1,\ell_2)}
	\end{equation}
	holds on $\cF$.
\end{lemma}

\subsection{Localisation into boxes}

\label{subsec:localisation_boxes}

Let $\Lambda_1 = \left[-1/2,1/2\right)^3$ and define
\begin{equation}
	\label{eq:effective_two_body_boxes_definition}
	\Lambda_r^{(\ell)} = r + \ell\Lambda_1,
\end{equation}
the box centred $r$ with side length $\ell$. For a given number of particle $M$, define also
\begin{equation}
	\label{eq:effective_two_body_number_particle_in_box_definition}
	M_r^{(\ell)} = \sum_{i=1}^M\mathds{1}_{\Lambda_r^{(\ell)}}(x_i), \quad M_{r,j}^{(\ell)} = \sum_{\substack{i=1\\i\neq j}}^M\mathds{1}_{\Lambda_r^{(\ell)}}(x_i) \quad \textmd{and} \quad M_{r,jk}^{(\ell)} = \sum_{\substack{i=1\\i\neq j,k}}^M\mathds{1}_{\Lambda_r^{(\ell)}}(x_i),
\end{equation}
for any configuration $(x_1,\dots,x_M)\in\R^{3M}$. Though all three quantities depend on $(x_1,\dots,x_M)$, we omit this from the notation for readability. The quantity $M_r^{(\ell)}$ denotes the number of particles in the box $\Lambda_r^{(\ell)}$; $M_{r,j}^{(\ell)}$ and $M_{r,jk}^{(\ell)}$ are the same without counting the $j$-th and $k$-th particles. Further, define
\begin{equation*}
	M_{r,s}^{(\ell,\ell')} = \sum_{i = 1}^M\1_{\Lambda_r^{(\ell)}}(x_i)\1_{\Lambda_s^{(\ell')}}(x_i);
\end{equation*}
this counts the number of particles that are contained in both boxes $\Lambda_r^{(\ell)}$ and $\Lambda_s^{(\ell')}$.

In what follows, we will often distinguish between $M_r^{(\ell)} \geq C$ and $M_r^{(\ell)} < C$ for some given threshold $C$. For readability, we shall write this using the Iverson bracket $[\cdot]$, defined by
\begin{equation*}
	[A] = \left\{
	\begin{aligned}
		1 &\quad \textmd{if $A$ is true,}\\
		0 &\quad \textmd{otherwise.}
	\end{aligned}
	\right.
\end{equation*}
Namely, we write $[M_r^{(\ell)} \geq C]$ instead of $\mathds{1}(M_r^{(\ell)} \geq C)$. To avoid any confusion with commutators, we use the Iverson bracket only in this section, where no commutators are computed.

Define the projections
\begin{equation}
	\label{eq:projection_constant_function_box}
	p_r^{(\ell)} = \frac{1}{\ell^3}\vert\mathds{1}_{\Lambda_r^{(\ell)}}\rangle\langle\mathds{1}_{\Lambda_r^{(\ell)}}\vert \quad \textmd{and} \quad q_r^{(\ell)} = \mathds{1}_{\Lambda_r^{(\ell)}} - p_r^{(\ell)}.
\end{equation}
The former is the projection on constant functions in the box $\Lambda_r^{(\ell)}$ and the latter is the projection on functions orthogonal to constant functions in the same box.

An important element in the proof of Proposition~\ref{prop:general_two_body_potential_bound} is the following convexity estimate.

\begin{lemma}
	\label{lemma:convexity_estimate_number_particles_box}
	Let $\ell_1,\ell_2,\ell_3 > 0$ with $\ell_1 \geq \ell_2$ and $\ell_1 \geq \ell_3$, and $\alpha,\beta \geq 1$. Let $M$ be a nonnegative integer. Then, there exist nonnegative constants $C_\beta$ and $C_{\alpha,\beta}$ that depend respectively only on $\beta$ and $\alpha,\beta$, such that
	\begin{equation}
		\label{eq:convexity_estimate_number_particles_box}
		\sum_{r\in\ell_2\mathbb{Z}^3}(M_r^{(\ell_1)})^\alpha[M_r^{(\ell_1)} \geq C_\beta^{(\ell_1,\ell_3)}] \leq C_{\alpha,\beta}\Big(\frac{\ell_1}{\ell_2}\Big)^3\Big(\frac{\ell_1}{\ell_3}\Big)^{3(\alpha - 1)}\sum_{s\in\ell_3\mathbb{Z}^3}(M_s^{(\ell_3)})^\alpha[M_s^{(\ell_3)} \geq \beta]
	\end{equation}
	holds on $\mathfrak{H}^M$, where $C_\beta^{(\ell_1,\ell_3)}$ is given by
	\begin{equation}
		\label{eq:convexity_estimate_constant}
		C_{\beta}^{(\ell_1,\ell_3)} = C_\beta\Big(\dfrac{\ell_1}{\ell_3}\Big)^3.
	\end{equation}
\end{lemma}

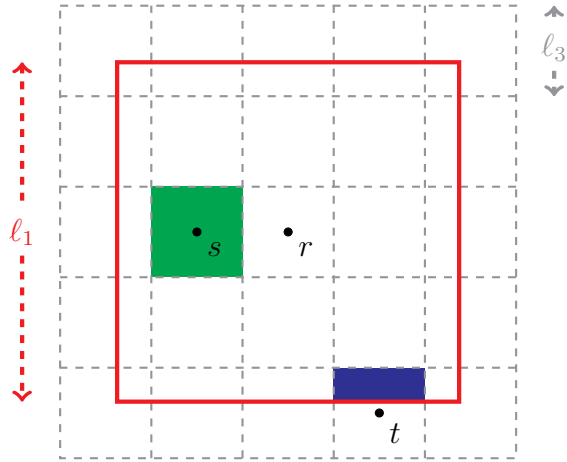
\begin{figure}[ht]
	\centering
	\begin{tikzpicture}
	\tikzmath{\w = 6; \wi = \w - \w/4; \se = \w/5;}
	
	\fill[Green] (-\w/2 + \se,-\w/2 + 2*\se) rectangle (-\w/2 + 2*\se,-\w/2 + 3*\se);
	
	\fill[Blue] (-\w/2 + 3*\se,-\wi/2) rectangle (-\w/2 + 4*\se,-\w/2 + \se);
	
	\draw[Gray, thick, dashed] (-\w/2,\w/2) -- (\w/2,\w/2) -- (\w/2,-\w/2) -- (-\w/2,-\w/2) -- cycle;
	
	\draw[Gray, thick, dashed] (\w/2,-\w/2 + \se) -- (-\w/2,-\w/2 + \se);
	\draw[Gray, thick, dashed] (\w/2,-\w/2 + 2*\se) -- (-\w/2,-\w/2 + 2*\se);
	\draw[Gray, thick, dashed] (\w/2,-\w/2 + 3*\se) -- (-\w/2,-\w/2 + 3*\se);
	\draw[Gray, thick, dashed] (\w/2,-\w/2 + 4*\se) -- (-\w/2,-\w/2 + 4*\se);
	\draw[Gray, thick, dashed] (-\w/2 + \se,-\w/2) -- (-\w/2 + \se,\w/2);
	\draw[Gray, thick, dashed] (-\w/2 + 2*\se,-\w/2) -- (-\w/2 + 2*\se,\w/2);
	\draw[Gray, thick, dashed] (-\w/2 + 3*\se,-\w/2) -- (-\w/2 + 3*\se,\w/2);
	\draw[Gray, thick, dashed] (-\w/2 + 4*\se,-\w/2) -- (-\w/2 + 4*\se,\w/2);
	
	\draw[Red, ultra thick] (-\wi/2,-\wi/2) -- (-\wi/2,\wi/2) -- (\wi/2,\wi/2) -- (\wi/2,-\wi/2) -- cycle;
	
	\filldraw[black] (0,0) circle (1.5pt) node[anchor=north west]{\large $r$};
	\filldraw[black] (-\se,0) circle (1.5pt) node[anchor=north west]{\large $s$};
	\filldraw[black] (\se,-2*\se) circle (1.5pt) node[anchor=north west]{\large $t$};
	\node[color=Red](a) at (-\w/2 - 0.5,0) {\large $\ell_1$};
	\draw[color=Red, ultra thick, dashed,<-](-\w/2 - 0.5,\wi/2) -- (a);
	\draw[color=Red, ultra thick, dashed,->](a) -- (-\w/2 -0.5,-\wi/2);
	\node[color=Gray](b) at (\w/2 + 0.5,\wi/2 + 0.2) {\large $\ell_3$};
	\draw[color=Gray, ultra thick, dashed,<-](\w/2 + 0.5,\w/2) -- (b);
	\draw[color=Gray, ultra thick, dashed,->](b) -- (\w/2 + 0.5,\w/2 - \se);
\end{tikzpicture}
	\caption{Illustration of a coverage of the box ${\color{Red}\Lambda_r^{(\ell_1)}}$ by the smaller boxes ${\color{Gray}\Lambda_s^{(\ell_3)}}$. Some of the {\color{gray}grey} boxes overlap only partially with the {\color{red}red} one. The quantities {\color{Green}$M_{r,s}^{(\ell_1,\ell_3)}$} and {\color{Blue}$M_{r,t}^{(\ell_1,\ell_3)}$} respectively count the number of particles in the {\color{Green}green} and {\color{Blue}blue} rectangles.}
	\label{fig:coverage_large_box_smaller_boxes}
\end{figure}

We say a few words about Lemma~\ref{lemma:convexity_estimate_number_particles_box} before providing its proof. Firstly, the condition
\begin{equation}
	\label{eq:convexity_estimate_constant_condition}
	M_r^{(\ell_1)} \geq C_\beta^{(\ell_1,\ell_3)}
\end{equation}
states that the box $\Lambda_r^{(\ell_1)}$ contains sufficiently many particles so that, if we consider any coverage of $\Lambda_r^{(\ell_1)}$ by disjoint smaller boxes $\Lambda_s^{(\ell_3)}$ (see Figure~\ref{fig:coverage_large_box_smaller_boxes}), and distribute the particles evenly across these smaller boxes, then each of them contains at least $\beta$ particles. Essentially, this ensures that, for any given configuration, there are few boxes $\Lambda_s^{(\ell_3)}$ that contain less than $\beta$ particles. This allows us to prove the bound
\begin{equation*}
	M_r^{(\ell_1)} \leq C\sum_{s\in\ell_3\mathbb{Z}^3}M_{r,s}^{(\ell_1,\ell_3)}[M_{r,s}^{(\ell_1,\ell_3)} \geq \beta],
\end{equation*}
under the condition \eqref{eq:convexity_estimate_constant_condition}. By a convexity argument, this bound implies
\begin{equation*}
	(M_r^{(\ell_1)})^\alpha \leq C\Big(\frac{\ell_1}{\ell_3}\Big)^{3(\alpha - 1)}\sum_{s\in\ell_3\mathbb{Z}^3}(M_{r,s}^{(\ell_1,\ell_3)})^\alpha[M_{r,s}^{(\ell_1,\ell_3)} \geq \beta],
\end{equation*}
again under the condition \eqref{eq:convexity_estimate_constant_condition}. To obtain \eqref{eq:convexity_estimate_number_particles_box}, we sum over $r\in\ell_2\Z^3$. The overlap between the different boxes $\Lambda_r^{(\ell_1)}$ is accounted for by the factor $(\ell_1/\ell_2)^3$ in \eqref{eq:convexity_estimate_number_particles_box}.

\begin{proof}[Proof of Lemma~\ref{lemma:convexity_estimate_number_particles_box}]
	First, we use that
	\begin{equation*}
		\big\{\Lambda_r^{(\ell_1)}\cap\Lambda_s^{(\ell_3)}:s\in\ell_3\Z^3\big\}
	\end{equation*}
	is a coverage of $\Lambda_r^{(\ell_1)}$ to write
	\begin{equation}
		\label{eq:coverage_large_box_smaller_boxes}
		M_r^{(\ell_1)} = \sum_{s\in\ell_3\mathbb{Z}^3}M_{r,s}^{(\ell_1,\ell_3)}
	\end{equation}
	(see Figure~\ref{fig:coverage_large_box_smaller_boxes}). Next, we define
	\begin{equation*}
		D_r^{(\ell_1,\ell_3)} = \#\big\{s\in\ell_3\mathbb{Z}^3: \Lambda_r^{(\ell_1)}\cap\Lambda_s^{(\ell_3)}\neq \emptyset\big\},
	\end{equation*}
	which counts the number of boxes $\Lambda_s^{(\ell_3)}$ that intersect (at least partially) with $\Lambda_r^{(\ell_1)}$. Then, we multiplying both sides of the identity \eqref{eq:coverage_large_box_smaller_boxes} by $[M_r^{(\ell_1)} \geq C_\beta^{(\ell_1,\ell_3)}]$ and we split the right-hand side between $M_{r,s}^{(\ell_1,\ell_3)} \geq \beta$ and $M_{r,s}^{(\ell_1,\ell_3)} < \beta$ to get
	\begin{equation}
		\label{eq:convexity_estimate_number_particles_box_bound_1}
		M_r^{(\ell_1)}[M_r^{(\ell_1)} \geq C_\beta^{(\ell_1,\ell_3)}] \leq \beta D_r^{(\ell_1,\ell_3)}[M_r^{(\ell_1)} \geq C_\beta^{(\ell_1,\ell_3)}] + \sum_{s\in\ell_3\mathbb{Z}^3}M_{r,s}^{(\ell_1,\ell_3)}[M_{r,s}^{(\ell_1,\ell_3)} \geq \beta],
	\end{equation}
	But, by definition of $C_\beta^{(\ell_1,\ell_3)}$ and since $D_r^{(\ell_1,\ell_3)}$ is bounded by $C(\ell_1/\ell_3)^3$, we have
	\begin{equation*}
		\beta D_r^{(\ell_1,\ell_3)}[M_r^{(\ell_1)} \geq C_\beta^{(\ell_1,\ell_3)}] \leq \frac{C\beta}{C_\beta}M_r^{(\ell_1)}[M_r^{(\ell_1)} \geq C_\beta^{(\ell_1,\ell_3)}].
	\end{equation*}
	Therefore, taking the constant $C_\beta$ in \eqref{eq:convexity_estimate_constant} large enough, we can shift this term to the left of the equation \eqref{eq:convexity_estimate_number_particles_box_bound_1} to obtain
	\begin{equation}
		\label{eq:many_particles_large_box_bounded_smaller_box}
		M_r^{(\ell_1)}[M_r^{(\ell_1)} \geq C_\beta^{(\ell_1,\ell_3)}] \leq C\sum_{s\in\ell_3\mathbb{Z}^3}M_{r,s}^{(\ell_1,\ell_3)}[M_{r,s}^{(\ell_1,\ell_3)} \geq \beta],
	\end{equation}
	for some $C$ that depends only on $\beta$.
	
	Now, we raise \eqref{eq:many_particles_large_box_bounded_smaller_box} to the power $\alpha$, and use a convexity argument to bound the right-hand side. Namely, by convexity of $x \mapsto x^\alpha$, Jensen's inequality implies
	\begin{equation*}
		\bigg(\frac{1}{D_r^{(\ell_1,\ell_3)}}\sum_{s\in\ell_3\mathbb{Z}^3}M_{r,s}^{(\ell_1,\ell_3)}[M_{r,s}^{(\ell_1,\ell_3)} \geq \beta]\bigg)^\alpha \leq \frac{1}{D_r^{(\ell_1,\ell_3)}}\sum_{s\in\ell_3\mathbb{Z}^3}(M_{r,s}^{(\ell_1,\ell_3)})^\alpha[M_{r,s}^{(\ell_1,\ell_3)} \geq \beta].
	\end{equation*}
	Thus, using again that $D_r^{(\ell_1,\ell_3)} \leq C(\ell_1/\ell_3)^3$, we deduce
	\begin{equation}
		\label{eq:many_particles_large_box_bounded_smaller_box_power}
		(M_r^{(\ell_1)})^\alpha[M_r^{(\ell_1)} \geq C_\beta^{(\ell_1,\ell_3)}] \leq C\Big(\frac{\ell_1}{\ell_3}\Big)^{3(\alpha - 1)}\sum_{s\in\ell_3\mathbb{Z}^3}(M_{r,s}^{(\ell_1,\ell_3)})^\alpha[M_{r,s}^{(\ell_1,\ell_3)} \geq \beta].
	\end{equation}
	Finally, we sum over $r\in\ell_2\Z^3$ and use that, for any $s$,
	\begin{equation*}
		\#\big\{r\in\ell_2\Z^3:\Lambda_r^{(\ell_1)}\cap\Lambda_s^{(\ell_3)} \neq \emptyset \big\} \leq C\Big(\dfrac{\ell_1}{\ell_2}\Big)^3,
	\end{equation*}
	to obtain \eqref{eq:convexity_estimate_number_particles_box}.
\end{proof}

In addition to Lemma~\ref{lemma:convexity_estimate_number_particles_box}, we will also need the following two bounds to prove Proposition~\ref{prop:general_two_body_potential_bound}.

\begin{lemma}
	Let $\ell_1, \ell_2 > 0$ such that $\ell_2 \leq \ell_1/\sqrt{6}$. Let $M$ be a nonnegative integer. Then,
	\label{lemma:three_body_estimate_box}
	\begin{equation}
		\label{eq:three_body_estimate_box}
		\sum_{\substack{1\leq i,j,k\leq M\\ i\neq j\neq k\neq i}}\mathds{1}(\vert(x_i - x_j,x_i - x_k)\vert \leq \ell_1) \geq \sum_{r\in\ell_2\mathbb{Z}^3}M_r^{(\ell_2)}(M_r^{(\ell_2)} - 1)(M_r^{(\ell_2)} - 2)
	\end{equation}
	holds on $\mfH^M$.
\end{lemma}

\begin{lemma}
	\label{lemma:gse_three_body_potential}
	Let $\ell_1 \geq \ell_2 > 0$ and let $M$ be an integer such that
	\begin{equation}
		\label{eq:gse_three_body_potential_condition_number_particles}
		M \leq \Big(\dfrac{\ell_1}{\ell_2}\Big)^3.
	\end{equation}
	Then, there exist universal constants $C,C' > 0$ such that
	\begin{equation*}
		C\sum_{i = 1}^M-\Delta_i + \sum_{\substack{1\leq i,j,k\leq M\\ i\neq j\neq k\neq i}}\ell_2^{-2}\mathds{1}(|(x_i - x_j, x_i - x_k)| \leq \ell_2) \geq C'\frac{\ell_2^4}{\ell_1^6}M^3
	\end{equation*}
	holds on $L_\sym^2((\Lambda_0^{(\ell_1)})^M)$.
\end{lemma}

\begin{proof}[Proof of Lemma~\ref{lemma:three_body_estimate_box}]
	Using the bounds
	\begin{equation*}
		\mathds{1}(\vert(x_i - x_j, x_i - x_k)\vert\leq \ell_1) \geq \mathds{1}(\vert x_i - x_j\vert \leq \ell_1/\sqrt{2})\mathds{1}(\vert x_i - x_k\vert \leq \ell_1/\sqrt{2})
	\end{equation*}
	and
	\begin{equation*}
		1 \geq \sum_{r\in\ell_2\mathbb{Z}^3}\mathds{1}_{(\Lambda_r^{(\ell_2)})^3}(x_i,x_j,x_k),
	\end{equation*}
	and the fact that two particles located in a box of side length $\ell_2$ are at a distance at most $\sqrt{3}\ell_2 \leq \ell_1/\sqrt{2}$ from one another, we obtain
	\begin{equation*}
		\mathds{1}(\vert(x_i - x_j, x_i - x_k)\vert\leq \ell_1) \geq \sum_{r\in\ell_2\mathbb{Z}^3}\mathds{1}_{(\Lambda_r^{(\ell_2)})^3}(x_i,x_j,x_k).
	\end{equation*}
	Summing over $i,j,k$ yields
	\begin{align*}
		\sum_{\substack{1\leq i,j,k\leq M\\ i\neq j\neq k\neq i}}\mathds{1}(\vert(x_i - x_j, x_i - x_k)\vert\leq \ell_1) &\geq \sum_{r\in\ell_2\mathbb{Z}^3}\sum_{\substack{1\leq i,j,k\leq M\\ i\neq j\neq k\neq i}}\mathds{1}_{(\Lambda_r^{(\ell_2)})^3}(x_i,x_j,x_k)\\
		&= \sum_{r\in\ell_2\mathbb{Z}^3}M_r^{(\ell_2)}(M_r^{(\ell_2)} - 1)(M_r^{(\ell_2)} - 2),
	\end{align*}
	which is precisely \eqref{eq:three_body_estimate_box}.
\end{proof}

\begin{proof}[Proof of Lemma \ref{lemma:gse_three_body_potential}]
	This follows directly from an application of \cite[Proposition 4]{NamRicTri22b} to the three-body potential
	\begin{equation*}
		(x,y,z) \mapsto \varepsilon\ell_2^{-2}\mathds{1}(|\mathcal{M}(x - y, x - z)| \leq \ell_3),
	\end{equation*}
	for some $\varepsilon > 0$ small enough. The smallness of $\varepsilon$ and condition \eqref{eq:gse_three_body_potential_condition_number_particles} on the number of particles ensure that we are in a dilute regime, i.e. that the error from \cite[Proposition 4]{NamRicTri22b} is small.
\end{proof}

\subsection{Proof of Proposition~\ref{prop:general_two_body_potential_bound}}

\label{sec:general_two_body_potential_bound}

Thanks to the relations
\begin{equation*}
	\d{}\Gamma(-\Delta) = \bigoplus_{M\geq 0}\bigg(\sum_{i=1}^M-\Delta_i\bigg), \quad \cV_\twoB^{(\ell_2)} = \bigoplus_{M\geq 0}\bigg(\sum_{\substack{1\leq i,j\leq M\\ i\neq j}}\mathds{1}_{\{|x_i-x_j|\leq \ell_2\}}\bigg)
\end{equation*}
and
\begin{equation*}
	\cV_\threeB^{(\ell_3)} = \bigoplus_{M\geq0}\bigg(\sum_{\substack{1\leq i,j,k\leq M\\ i\neq j\neq k\neq i}}\mathds{1}_{\left\{|(x_i-x_j,x_i-x_k)|\leq \ell_3\right\}}\bigg),
\end{equation*}
proving \eqref{eq:general_two_body_potential_bound} on $\mathcal{F}$ is equivalent to proving
\begin{multline}
	\label{eq:general_two_body_potential_bound_equivalent_statement}
	\sum_{\substack{1\leq i,j\leq M\\ i\neq j}}\mathds{1}_{\{|x_i-x_j|\leq \ell_2\}} \leq C\Big(\frac{\ell_2}{\ell_1}\Big)^3M + C\ell_2^2\Big (1 + \dfrac{\ell_2}{\ell_3}\Big)\Big(\frac{\ell_1}{\ell_3}\Big)^3\sum_{i=1}^M-\Delta_i\\
	+ C\Big(1 + \Big(\frac{\ell_2}{\ell_3}\Big)^3\Big)\Big(\frac{\ell_1}{\ell_3}\Big)^3\sum_{\substack{1\leq i,j,k\leq M\\ i\neq j\neq k\neq i}}\mathds{1}_{\{|(x_i-x_j,x_i-x_k)|\leq \ell_3\}}
\end{multline}
on $\mathfrak{H}^M$ for any nonnegative integer $M$.
	
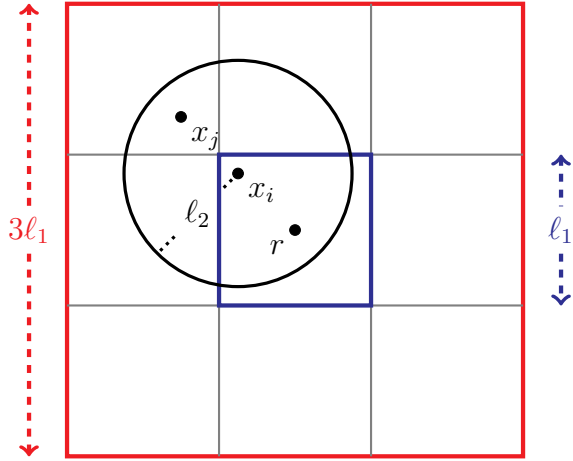
\begin{figure}[ht]
	\centering
	\begin{tikzpicture}
	\tikzmath{\w = 6; \se = \w/3;}
	
	\draw[gray, thick] (-\w/2,-\w/2 + \se) -- (\w/2,-\w/2 + \se);
	\draw[gray, thick] (-\w/2,-\w/2 + 2*\se) -- (\w/2,-\w/2 + 2*\se);
	\draw[gray, thick] (-\w/2 + \se,-\w/2) -- (-\w/2 + \se,\w/2);
	\draw[gray, thick] (-\w/2 + 2*\se,-\w/2) -- (-\w/2 + 2*\se,\w/2);
	
	\draw[Red, ultra thick] (-\w/2,-\w/2) -- (-\w/2,\w/2) -- (\w/2,\w/2) -- (\w/2,-\w/2) -- cycle;
	
	\draw[Blue, ultra thick] (-\w/2 + \se,-\w/2 + \se) -- (-\w/2 + \se,-\w/2 + 2*\se) -- (-\w/2 + 2*\se,-\w/2 + 2*\se) -- (-\w/2 + 2*\se,-\w/2 + \se) -- cycle;
	
	\filldraw[black] (0,0) circle (2pt) node[anchor=north east]{\large $r$};
	\draw[color=black, very thick](-\w/8,\w/8) circle (\w/4);
	\node[color=black](a) at ({-\w/8 - \w/(2*4*sqrt(2))},{\w/8- \w/(2*4*sqrt(2))}) {\large $\ell_2$};
	\draw[color=black, very thick,dotted](-\w/8,\w/8) -- (a);
	\draw[color=black, very thick,dotted](a) -- ({-\w/8 - \w/(4*sqrt(2))},{\w/8- \w/(4*sqrt(2))});
	
	\filldraw[black] (-\w/8,\w/8) circle (2pt) node[anchor=north west]{\large $x_i$};
	\filldraw[black] (-\w/4,\w/4) circle (2pt) node[anchor=north west]{\large $x_j$};
	
	\node[color=Red](a) at (-\w/2 - 0.5,0) {\large $3\ell_1$};
	\draw[color=Red, ultra thick, dashed,<-](-\w/2 - 0.5,\w/2) -- (a);
	\draw[color=Red, ultra thick, dashed,->](a) -- (-\w/2 - 0.5,-\w/2);
	\node[color=Blue](b) at (\w/2 + 0.5,0) {\large $\ell_1$};
	\draw[color=Blue, ultra thick, dashed,<-](\w/2 + 0.5,\se/2) -- (b);
	\draw[color=Blue, ultra thick, dashed,->](b) -- (\w/2 + 0.5,-\se/2);
\end{tikzpicture}
	\caption{Illustration of the boxes ${\color{Blue}\Lambda_r^{(\ell_1)}}$ and ${\color{Red}\Lambda_r^{(3\ell_1)}}$. Because the interaction has range $\ell_2$, which is less than $\ell_1$, particles located in the box $\Lambda_r^{(\ell_1)}$ can only interact with particles located in the box $\Lambda_r^{(3\ell_1)}$.}
	\label{fig:fig02}
\end{figure}

Let us rewrite the left-hand side of  \eqref{eq:general_two_body_potential_bound_equivalent_statement}. Firstly, we decompose the space $\R^3$ into boxes of side length $\ell_1$. Namely, we write
\begin{equation}
	\label{eq:effective_two_body_potential_decomposition_full_space_boxes}
	\R^3 = \bigcup_{r\in \ell_1\mathbb{Z}^3}\Lambda_r^{(\ell_1)},
\end{equation}
where we recall that $\Lambda_r^{(\ell_1)}$ was defined in \eqref{eq:effective_two_body_boxes_definition}. Hence,
\begin{align*}
	\mathds{1}_{\{|x_i-x_j|\leq \ell_2\}} = \sum_{r\in \ell_1\mathbb{Z}^3}\mathds{1}_{\{|x_i-x_j|\leq \ell_2\}}\mathds{1}_{\Lambda_r^{(\ell_1)}}(x_i) &= \sum_{r\in \ell_1\mathbb{Z}^3}\mathds{1}_{\{|x_i-x_j|\leq \ell_2\}}\mathds{1}_{\Lambda_r^{(\ell_1)}}(x_i)\mathds{1}_{\Lambda_r^{(3\ell_1)}}(x_j)\\
	&\leq \sum_{r\in \ell_1\mathbb{Z}^3}\mathds{1}_{\{|x_i-x_j|\leq \ell_2\}}\mathds{1}_{\Lambda_r^{(3\ell_1)}}(x_i)\mathds{1}_{\Lambda_r^{(3\ell_1)}}(x_j). \numberthis \label{eq:effective_two_body_potential_rewritten1}
\end{align*}
In the first equality, we used \eqref{eq:effective_two_body_potential_decomposition_full_space_boxes}. In the second, we used that, since the two-body interactions have range $\ell_2$, which is less than $\ell_1$, any particle located in a box $\Lambda_r^{(\ell_1)}$ can only interact with particles located in the larger box $\Lambda_r^{(3\ell_1)}$ (see Figure~\ref{fig:fig02}). In the third inequality we used the trivial inclusion $\Lambda_r^{(\ell_1)}\subset \Lambda_r^{(3\ell_1)}$. Note that, since the $r$'s we consider are in $\ell_1\mathbb{Z}^3$, the boxes $\Lambda_r^{(3\ell_1)}$ are not disjoint and because of the overlap we have
\begin{equation}
	\label{eq:box_recombination}
	\sum_{r\in \ell_1\mathbb{Z}^3}M_r^{(3\ell_1)} = 27\sum_{r\in \ell_1\mathbb{Z}^3}M_r^{(\ell_1)} = 27M,
\end{equation}
where we recall that $M_r^{(\ell_1)}$ was defined in \eqref{eq:effective_two_body_number_particle_in_box_definition}. We now distinguish between boxes that contain many particles, and the others. Namely, we distinguish between $M_r^{(3\ell_1)}$ smaller or larger than a threshold $C_{\ell_1,\ell_3}$ defined by
\begin{equation*}
	C_{\ell_1,\ell_3} = C\Big(\frac{\ell_1}{\ell_3}\Big)^3,
\end{equation*}
for some universal constant $C$ taken slightly larger than the constant $C_\beta$ in Lemma~\ref{lemma:convexity_estimate_number_particles_box} with $\beta = 3$. Define also
\begin{equation*}
	\cI_1 = \sum_{r\in \ell_1\mathbb{Z}^3}[M_r^{(3\ell_1)} \geq C_{\ell_1,\ell_3}]\sum_{\substack{1\leq i,j\leq M\\ i\neq j}}\mathds{1}_{\{|x_i-x_j|\leq \ell_2\}}\mathds{1}_{(\Lambda_r^{(3\ell_1)})^2}(x_i,x_j)
\end{equation*}
and
\begin{equation*}
	\cI_2 = \sum_{r\in \ell_1\mathbb{Z}^3}[M_r^{(3\ell_1)} < C_{\ell_1,\ell_3}]\sum_{\substack{1\leq i,j\leq M\\ i\neq j}}\mathds{1}_{\{|x_i-x_j|\leq \ell_2\}}\mathds{1}_{(\Lambda_r^{(3\ell_1)})^2}(x_i,x_j).
\end{equation*}
Then, using \eqref{eq:effective_two_body_potential_rewritten1} we get
\begin{equation*}
	\sum_{\substack{1\leq i,j\leq M\\ i\neq j}}\mathds{1}_{\{|x_i - x_j|\leq \ell_2\}} \leq \cI_1 + \cI_2.
\end{equation*}
Next, we bound $\cI_1$ and $\cI_2$ separately.
\\

\noindent
\textbf{Analysis of $\cI_1$.} It follows directly from the definition \eqref{eq:effective_two_body_number_particle_in_box_definition} that
\begin{equation*}
	\sum_{\substack{1\leq i,j\leq M\\ i\neq j}}\mathds{1}_{(\Lambda_r^{(3\ell_1)})^2}(x_i,x_j) = M_r^{(3\ell_1)}(M_r^{(3\ell_1)} - 1) \leq (M_r^{(3\ell_1)})^2.
\end{equation*}
Therefore,
\begin{align*}
	\cI_1 &\leq \sum_{r\in\ell_1\mathbb{Z}^3}[M_r^{(3\ell_1)} \geq C_{\ell_1,\ell_3}](M_r^{(3\ell_1)})^2 \leq C\Big(\frac{\ell_3}{\ell_1}\Big)^3\sum_{r\in\ell_1\mathbb{Z}^3}[M_r^{(3\ell_1)} \geq C_{\ell_1,\ell_3}](M_r^{(3\ell_1)})^3.
\end{align*}
Applying Lemma~\ref{lemma:convexity_estimate_number_particles_box} with $\alpha = 2$ and $\beta = 3$, we find
\begin{equation*}
	\cI_1 \leq C\Big(\frac{\ell_1}{\ell_3}\Big)^3\sum_{s\in\tilde{\ell}_3\mathbb{Z}^3}(M_s^{(\tilde{\ell}_3)})^3[M_s^{(\tilde{\ell}_3)} \geq 3] \leq C\Big(\frac{\ell_1}{\ell_3}\Big)^3\sum_{s\in\tilde{\ell}_3\mathbb{Z}^3}M_s^{(\tilde{\ell}_3)}(M_s^{(\tilde{\ell}_3)} - 1)(M_s^{(\tilde{\ell}_3)} - 2),
\end{equation*}
where we defined $\tilde{\ell}_3 = \ell_3/\sqrt{6}$. Thanks to Lemma~\ref{lemma:three_body_estimate_box}, we finally obtain
\begin{equation}
	\label{eq:two_body_effective_potential_I1_bound2}
	\cI_1 \leq C\Big(\frac{\ell_1}{\ell_3}\Big)^3\sum_{\substack{1\leq i,j,k\leq M\\ i\neq j\neq k\neq i}}\mathds{1}_{\left\{\vert(x_i - x_j, x_i - x_k)\vert \leq \ell_3\right\}}.
\end{equation}
\\

\noindent
\textbf{Analysis of $\cI_2$.} To bound $\cI_2$, we wish to use the Sobolev inequality in the box $\Lambda_r^{(3\ell_1)}$. However, if we do so directly we get a large error term due to the boundary of $\Lambda_r^{(3\ell_1)}$. To avoid that, we apply the inequality only on the part of $\cI_2$ that acts on functions orthogonal to constant functions. For this, we use the resolution of the identity $\mathds{1}_{\Lambda_{r}^{(3\ell_1)}}(x_i) = p_{r,i}^{(3\ell_1)} + q_{r,i}^{(3\ell_1)}$ and the Cauchy--Schwarz inequality to bound
\begin{equation}
	\label{eq:two_body_effective_potential_few_particles_decompositon}
	\sum_{\substack{1\leq i,j\leq M\\ i\neq j}} \mathds{1}_{\{|x_i-x_j|\leq \ell_2\}}\mathds{1}_{(\Lambda_r^{(3\ell_1)})^2}(x_i,x_j)[M_r^{(3\ell_1)} < C_{\ell_1,\ell_3}] \leq \mathcal{J}_{1,r} + \mathcal{J}_{2,r} + \mathcal{J}_{3,r},
\end{equation}
with $\mathcal{J}_{1,r}, \mathcal{J}_{2,r}$ and  $\mathcal{J}_{3,r}$ given by
\begin{align*}
	\mathcal{J}_{1,r} &= 2\sum_{\substack{1\leq i,j\leq M\\ i\neq j}} q_{r,i}^{(3\ell_1)}\mathds{1}_{\{|x_i-x_j|\leq \ell_2\}}\mathds{1}_{\Lambda_r^{(3\ell_1)}}(x_j)[M_r^{(3\ell_1)} < C_{\ell_1,\ell_3}]q_{r,i}^{(3\ell_1)},\\
	\mathcal{J}_{2,r} &= 2\sum_{\substack{1\leq i,j\leq M\\ i\neq j}} p_{r,i}^{(3\ell_1)}\mathds{1}_{\{|x_i-x_j|\leq \ell_2\}}\mathds{1}_{\Lambda_r^{(3\ell_1)}}(x_j)[3 \leq M_r^{(3\ell_1)} < C_{\ell_1,\ell_3}]p_{r,i}^{(3\ell_1)}
\end{align*}
and
\begin{equation*}
	\mathcal{J}_{3,r} = 2\sum_{\substack{1\leq i,j\leq M\\ i\neq j}}p_{r,i}^{(3\ell_1)}\mathds{1}_{\{|x_i-x_j|\leq \ell_2\}}\mathds{1}_{\Lambda_r^{(3\ell_1)}}(x_j)[M_r^{(3\ell_1)} = 2]p_{r,i}^{(3\ell_1)},
\end{equation*}
and where $p_{r,i}^{(3\ell_1)}$ and $q_{r,i}^{(3\ell_1)}$ denote the projections defined in \eqref{eq:projection_constant_function_box} acting on the $i$-th variable.
\\

\noindent
\textit{Analysis of $\mathcal{J}_{1,r}$.} First, we use the inequality $[M_r^{(3\ell_1)} < C_{\ell_1,\ell_3}] \leq [M_{r,i}^{(3\ell_1)} < C_{\ell_1,\ell_3}]$ and the fact that $q_{r,i}^{(3\ell_1)}$ commutes with $M_{r,i}^{(3\ell_1)}$ to write
\begin{equation*}
	q_{r,i}^{(3\ell_1)} \mathds{1}_{\{|x_i-x_j|\leq \ell_2\}}[M_r^{(3\ell_1)} < C_{\ell_1,\ell_3}] q_{r,i}^{(3\ell_1)}  \leq \big(q_{r,i}^{(3\ell_1)} \mathds{1}_{\{|x_i-x_j|\leq \ell_2\}} q_{r,i}^{(3\ell_1)}\big) [M_{r,i}^{(3\ell_1)} < C_{\ell_1,\ell_3}],
\end{equation*}
where we recall that $M_{r,i}^{(3\ell_1)}$ is defined in \eqref{eq:effective_two_body_number_particle_in_box_definition}. Then, we apply the Hölder inequality and the Poincaré--Sobolev inequality (see e.g. \cite[Theorem 8.12]{LiebLos01}) to deduce
\begin{equation*}
	q_{r,i}^{(3\ell_1)} \mathds{1}_{\{|x_i-x_j|\leq \ell_2\}}[M_r^{(3\ell_1)} < C_{\ell_1,\ell_3}] q_{r,i}^{(3\ell_1)} \leq C \ell^{2}_2 \nabla_i^* \mathds{1}_{\Lambda_{r}^{(3\ell_1)}}(x_i) \nabla_i [M_{r,i}^{(3\ell_1)} < C_{\ell_1,\ell_3}],
\end{equation*}
which implies
\begin{equation*}
	\mathcal{J}_{1,r} \leq C \ell^2_2  \sum_{i = 1}^M [M_{r,i}^{(3\ell_1)} < C_{\ell_1,\ell_3}]M_{r,i}^{(3\ell_1)} \nabla_i^* \mathds{1}_{\Lambda_{r}^{(3\ell_1)}} (x_i) \nabla_i \leq C \ell^2_2 C_{\ell_1,\ell_3}  \sum_{i = 1}^M \nabla_i^* \mathds{1}_{\Lambda_{r}^{(3\ell_1)}} (x_i) \nabla_i.
\end{equation*}
Finally, recombining the boxes using \eqref{eq:box_recombination} yields
\begin{equation}
	\label{eq:two_body_effective_potential_J1_bound}
	\sum_{r\in\ell_1\mathbb{Z}^3}\mathcal{J}_{1,r} \leq C\ell_2^2\Big(\frac{\ell_1}{\ell_3}\Big)^3\sum_{i=1}^M-\Delta_i.
\end{equation}
\\

\noindent
\textit{Analysis of $\mathcal{J}_{2,r}$.} Using that $p_{r,i}^{(3\ell_1)}$ commutes with $M_{r,i}^{(3\ell_1)}$ and that $M_r^{(3\ell_1)} = M_{r,i}^{(3\ell_1)} + 1$ when $x_i\in\Lambda_r^{(3\ell_1)}$, we find
\begin{align*}
	\mathcal{J}_{2,r} 
	&\leq  C\Big(\frac{\ell_2}{\ell_1}\Big)^3 \sum_{\substack{1\leq i,j\leq M\\ i\neq j}} p_{r,i}^{(3\ell_1)}\1_{\Lambda_r^{(3\ell_1)}}(x_j)[3 \leq M_r^{(3\ell_1)} \leq C_{\ell_1,\ell_3}] \\
	&\leq  C\Big(\frac{\ell_2}{\ell_1}\Big)^3[3 \leq M_r^{(3\ell_1)} \leq C_{\ell_1,\ell_3}] \Big(M_r^{(3\ell_1)}\Big)^2.
\end{align*}
After an application of Young's inequality for products, this becomes
\begin{equation*}
	\cJ_{2,r} \leq C\Big(\frac{\ell_2}{\ell_1}\Big)^3[3 \leq M_r^{(3\ell_1)} \leq C_{\ell_1,\ell_3}]\Big(M_{r}^{(3\ell_1)} + \Big(M_{r}^{(3\ell_1)}\Big)^3\Big).
\end{equation*}
Now, we use Lemma~\ref{lemma:gse_three_body_potential} to obtain
\begin{multline*}
	\mathcal{J}_{2,r} \leq C\Big(\frac{\ell_2}{\ell_1}\Big)^3 M_{r}^{(3\ell_1)} + C\frac{\ell_2^3}{\ell_3}\Big(\frac{\ell_1}{\ell_3}\Big)^3\sum_{i = 1}^M \nabla_i^* \mathds{1}_{\Lambda_r^{(3\ell_1)}}(x_i)\nabla_i\\
	+ C\Big(\frac{\ell_1}{\ell_3}\Big)^3\Big(\frac{\ell_2}{\ell_3}\Big)^3\sum_{\substack{1\leq i,j,k\leq M\\ i\neq j\neq k\neq i}}\mathds{1}_{\Lambda_r^{(3\ell_1)}}(x_i)\mathds{1}_{\left\{\left|(x_i - x_j,x_i - x_k)\right| \leq \ell_3\right\}}.
\end{multline*}
Lastly, we sum over $r\in\ell_1\mathbb{Z}^3$ and recombine the boxes with \eqref{eq:box_recombination} to get
\begin{multline}
	\label{eq:two_body_effective_potential_J2_bound2}
	\sum_{r\in\ell_1\mathbb{Z}^3}\mathcal{J}_{2,r} \leq C\Big(\frac{\ell_2}{\ell_1}\Big)^3M + \frac{\ell_2^3}{\ell_3}\Big(\frac{\ell_1}{\ell_3}\Big)^3\sum_{i = 1}^M-\Delta_i\\
	+ C\Big(\frac{\ell_1}{\ell_3}\Big)^3\Big(\frac{\ell_2}{\ell_3}\Big)^3\sum_{\substack{1\leq i,j,k\leq M\\ i\neq j\neq k\neq i}}\mathds{1}_{\{|(x_i-x_j,x_i-x_k)| \leq\ell_3\}}.
\end{multline}
\\

\noindent
\textit{Analysis of $\mathcal{J}_{3,r}$.} Following the same reasoning as for $\cJ_{2,r}^{(3\ell_1)}$ gives
\begin{equation*}
	\mathcal{J}_{3,r} \leq C\Big(\frac{\ell_2}{\ell_1}\Big)^3 \sum_{\substack{1\leq i,j\leq M\\ i\neq j}} p_{r,i}^{(3\ell_1)}\1_{\Lambda_r^{(3\ell_1)}}(x_j)[M_r^{(3\ell_1)} = 2] \leq C\Big(\frac{\ell_2}{\ell_1}\Big)^3 \sum_{i=1}^M\mathds{1}_{\Lambda_{r}^{(3\ell_1)}}(x_i),
\end{equation*}
which implies
\begin{equation}
	\label{eq:two_body_effective_potential_J3_bound}
	\sum_{r\in\ell_1\mathbb{Z}^3}\mathcal{J}_{3,r} \leq C\Big(\frac{\ell_2}{\ell_1}\Big)^3M.
\end{equation}

Combining \eqref{eq:two_body_effective_potential_I1_bound2}--\eqref{eq:two_body_effective_potential_J3_bound} yields \eqref{eq:general_two_body_potential_bound}, which concludes the proof of Proposition \ref{prop:general_two_body_potential_bound}.

\subsection{Proof of Corollaries~\ref{cor:general_three_body_potential_bound}~and~\ref{cor:general_four_body_potential_bound}~and~lemma~\ref{lemma:general_five_body_potential_bound}}

\label{subsec:proof_non_renormalasied_estimates_3B_4B}

\begin{proof}[Proof of Corollary~\ref{cor:general_three_body_potential_bound}]
	We prove \eqref{eq:general_three_body_potential_bound} by showing
	\begin{multline*}
		\sum_{\substack{1\leq i,j,k\leq M\\ i\neq j\neq k\neq i}}\mathds{1}_{\{\vert x_i - x_j\vert \leq \ell_2\}}\mathds{1}_{\{\vert x_i - x_k\vert \leq \ell_1\}} \leq C\Big(\frac{\ell_1}{\ell_3}\Big)^3\Big(\frac{\ell_2}{\ell_3}\Big)^3M + C\ell_2^2\Big(1 + \dfrac{\ell_2}{\ell_3}\Big)\Big(\dfrac{\ell_1}{\ell_3}\Big)^3\sum_{i=1}^M-\Delta_i\\
		+ C\Big(1 + \Big(\frac{\ell_2}{\ell_3}\Big)^3\Big)\Big(\frac{\ell_1}{\ell_3}\Big)^3\sum_{\substack{1\leq i,j,k\leq M\\ i\neq j\neq k\neq i}}\mathds{1}_{\left\{\vert(x_i - x_j,x_i - x_k)\vert \leq \ell_3\right\}}
	\end{multline*}
	on $\mfH^M,$ for any nonnegative integer $M$. Dividing the space $\R^3$ into boxes of side length $\ell_2$ and using that $\ell_2$ is less than $\ell_1$, we find
	\begin{equation*}
		\sum_{\substack{1\leq i,j,k\leq M\\ i\neq j\neq k\neq i}}\mathds{1}_{\{\vert x_i - x_j\vert \leq \ell_2\}}\mathds{1}_{\{\vert x_i - x_k\vert \leq \ell_1\}} \leq \sum_{r\in\ell_2\mathbb{Z}^3}M_r^{(3\ell_2)}(M_r^{(3\ell_2)} - 1)M_r^{(3\ell_1)}.
	\end{equation*}
	Define
	\begin{equation*}
		C_{\ell_1,\ell_3} = C\Big(\frac{\ell_1}{\ell_3}\Big)^3,
	\end{equation*}
	for some constant $C$ taken sufficiently large; in particular larger than the constant $C_\beta$ in Lemma~\ref{lemma:convexity_estimate_number_particles_box} with $\beta = 3$. Define as well
	\begin{equation*}
		\cI_1 = \sum_{r\in\ell_2\mathbb{Z}^3}M_r^{(3\ell_2)}(M_r^{(3\ell_2)} - 1)M_r^{(3\ell_1)}[M_r^{(3\ell_1)} < C_{\ell_1,\ell_3}]
	\end{equation*}
	and
	\begin{equation*}
		\cI_2 = \sum_{r\in\ell_2\mathbb{Z}^3}M_r^{(3\ell_2)}(M_r^{(3\ell_2)} - 1)M_r^{(3\ell_1)}[M_r^{(3\ell_1)} \geq C_{\ell_1,\ell_3}].
	\end{equation*}
	Thus,
	\begin{equation*}
		\sum_{\substack{1\leq i,j,k\leq M\\ i\neq j\neq k\neq i}}\mathds{1}_{\{\vert x_i - x_j\vert \leq \ell_2\}}\mathds{1}_{\{\vert x_i - x_k\vert \leq \ell_1\}} \leq \cI_1 + \cI_2.
	\end{equation*}
	We bound $\cI_1$ and $\cI_2$ separately.
	\\
	
	\noindent
	\textbf{Analysis of $\cI_1$.} By definition of $M_r^{(3\ell_2)}$, we can find some constant $C$ large enough (and independent of $\ell_2$) such that
	\begin{equation*}
		M_r^{(3\ell_2)}(M_r^{(3\ell_2)} - 1) \leq \sum_{\substack{1\leq i,j\leq M\\ i\neq j}}\mathds{1}_{\{\vert x_i - x_j\vert \leq C\ell_2\}}\mathds{1}_{\Lambda_r^{(3\ell_2)}}(x_i).
	\end{equation*}
	Moreover, for all $x\in\R^3$,
	\begin{equation*}
		\sum_{r\in\ell_2\Z^3}\mathds{1}_{\Lambda_r^{(3\ell_2)}}(x) = 27.
	\end{equation*}
	Hence, $\cI_1$ is bounded by
	\begin{equation*}
		\cI_1 \leq C\Big(\dfrac{\ell_1}{\ell_3}\Big)^3\sum_{\substack{1\leq i,j\leq M\\ i\neq j}}\mathds{1}_{\{\vert x_i - x_j\vert \leq C\ell_2\}}.
	\end{equation*}
	Thus, we may apply Proposition~\ref{prop:general_two_body_potential_bound} to obtain
	\begin{equation*}
		\cI_1 \leq C\Big(\dfrac{\ell_2}{\ell_3}\Big)^3M + C\ell_2^2\Big(1 + \dfrac{\ell_2}{\ell_3}\Big)\Big(\dfrac{\ell_1}{\ell_3}\Big)^6\sum_{i = 1}^M-\Delta_i + C\Big(1 + \Big(\dfrac{\ell_2}{\ell_3}\Big)^3\Big)\Big(\dfrac{\ell_1}{\ell_3}\Big)^6\sum_{\substack{1 \leq i,j,k\leq M\\ i\neq j \neq k \neq i}}\mathds{1}_{\left\{\vert(x_i - x_j,x_i - x_k)\vert \leq \ell_3\right\}}.
	\end{equation*}
	\\
	
	\noindent
	\textbf{Analysis of $\cI_2$.} We make a distinction between $M_r^{(3\ell_2)} < C_{\ell_2,\ell_3}$ and $M_r^{(3\ell_2)} \geq C_{\ell_2,\ell_3}$, where $C_{\ell_2,\ell_3}$ is given by
	\begin{equation*}
		C_{\ell_2,\ell_3} = C\Big(\frac{\ell_1}{\ell_3}\Big)^3,
	\end{equation*}
	for some $C$ large enough. More specifically, we define
	\begin{equation*}
		\cJ_1 = \sum_{r\in\ell_2\mathbb{Z}^3}(M_r^{(3\ell_2)})^2[M_r^{(3\ell_2)} < C_{\ell_2,\ell_3}]M_r^{(3\ell_1)}[M_r^{(3\ell_1)} \geq C_{\ell_1,\ell_3}]
	\end{equation*}
	and
	\begin{equation*}
		\cJ_2 = \sum_{r\in\ell_2\mathbb{Z}^3}(M_r^{(3\ell_2)})^2[M_r^{(3\ell_2)} \geq C_{\ell_2,\ell_3}]M_r^{(3\ell_1)}[M_r^{(3\ell_1)} \geq C_{\ell_1,\ell_3}],
	\end{equation*}
	and write
	\begin{equation*}
		\cI_2 \leq \cJ_1 + \cJ_2.
	\end{equation*}
	On the one hand, by definition of $C_{\ell_1,\ell_3}$ and $C_{\ell_2,\ell_3}$, we have
	\begin{equation*}
		\cJ_1 \leq C\Big(\dfrac{\ell_2}{\ell_1}\Big)^6\sum_{r\in\ell_2\Z^3}(M_r^{(3\ell_1)})^3[M_r^{(3\ell_1)} \geq C_{\ell_1,\ell_3}].
	\end{equation*}
	Using Lemmas~\ref{lemma:convexity_estimate_number_particles_box}~and~\ref{lemma:three_body_estimate_box}, we get
	\begin{equation*}
		\mathcal{J}_1 \leq C\Big(\dfrac{\ell_1}{\ell_3}\Big)^3\Big(\dfrac{\ell_2}{\ell_3}\Big)^3\sum_{r\in\tilde{\ell}_3\mathbb{Z}^3}(M_s^{(\tilde{\ell}_3)})^3[M_s^{(\tilde{\ell}_3)} \geq 3] \leq C\Big(\frac{\ell_1}{\ell_3}\Big)^3\Big(\frac{\ell_2}{\ell_3}\Big)^3\sum_{\substack{1\leq i,j,k\leq M\\ i\neq j\neq k\neq i}}\mathds{1}_{\left\{\vert(x_i - x_j, x_i - x_k)\vert \leq \ell_3\right\}},
	\end{equation*}
	where we intermediately defined $\tilde{\ell}_3 = \ell_3/\sqrt{6}$. On the other hand, thanks to Young's inequality for products, we obtain
	\begin{equation*}
		\cJ_2 \leq \dfrac{2\varepsilon^{1/2}}{3}\sum_{r\in\ell_2\mathbb{Z}^3}(M_r^{(3\ell_2)})^3[M_r^{(3\ell_2)} \geq C_{\ell_2,\ell_3}] + \dfrac{1}{3\varepsilon}\sum_{r\in\ell_2\mathbb{Z}^3}(M_r^{(3\ell_1)})^3[M_r^{(3\ell_1)} \geq C_{\ell_1,\ell_3}],
	\end{equation*}
	for all $\varepsilon > 0$. Again, we can use Lemmas~\ref{lemma:convexity_estimate_number_particles_box}~and~\ref{lemma:three_body_estimate_box} to show
	\begin{equation*}
		\cJ_2 \leq C\Big(\frac{\ell_1}{\ell_3}\Big)^3\Big(\frac{\ell_2}{\ell_3}\Big)^3\sum_{\substack{1\leq i,j,k\leq M\\ i\neq j\neq k\neq i}}\mathds{1}_{\left\{\vert(x_i - x_j, x_i - x_k)\vert \leq \ell_3\right\}}.
	\end{equation*}
	This finishes the proof of Corollary~\ref{cor:general_three_body_potential_bound}.
\end{proof}

\begin{proof}[Proof of Corollary~\ref{cor:general_four_body_potential_bound}]
	We prove \eqref{eq:general_four_body_potential_bound} by showing
	\begin{multline*}
		\sum_{\substack{1\leq i,j,k,\ell\leq M\\ i\neq j\neq k\neq i\\ \ell\neq i,j,k}}\mathds{1}_{\{\vert x_i - x_j\vert \leq \ell_2\}}\mathds{1}_{\{\vert x_i - x_k\vert \leq \ell_1\}}\mathds{1}_{\{\vert x_i - x_\ell\vert \leq \ell_1\}} \leq C\Big(\frac{\ell_1}{\ell_3}\Big)^3\Big(\frac{\ell_2}{\ell_3}\Big)^3M\\
		+ C\ell_2^2\Big(1 + \dfrac{\ell_2}{\ell_3}\Big)\Big(\frac{\ell_1}{\ell_3}\Big)^9\sum_{i = 1}^M-\Delta_i + C\Big(\frac{\ell_1}{\ell_3}\Big)^3\Big(\frac{\ell_2}{\ell_3}\Big)^3\Big(M + \Big(\frac{\ell_1}{\ell_3}\Big)^6\Big)\sum_{\substack{1\leq i,j,k\leq M\\ i\neq j\neq k\neq i}}\mathds{1}_{\left\{\vert(x_i - x_j,x_i - x_k)\vert \leq \ell_3\right\}}
	\end{multline*}
	on $\mfH^M,$ for any nonnegative integer $M$. Dividing the space $\R^3$ into boxes of side length $\ell_2$, as was done in the proof of Proposition~\ref{prop:general_two_body_potential_bound}, and using that $\ell_1$ is larger than $\ell_2$, we find
	\begin{equation*}
		\sum_{\substack{1\leq i,j,k,\ell\leq M\\ i\neq j\neq k\neq i\\ \ell\neq i,j,k}}\mathds{1}_{\{\vert x_i - x_j\vert \leq \ell_2\}}\mathds{1}_{\{\vert x_i - x_k\vert \leq \ell_1\}}\mathds{1}_{\{\vert x_i - x_\ell\vert \leq \ell_1\}} \leq \sum_{r\in\ell_2\mathbb{Z}^3}(M_r^{(3\ell_2)})^2M_r^{(3\ell_1)}(M_r^{(3\ell_1)} - 1).
	\end{equation*}
	Define
	\begin{equation*}
		C_{\ell_2,\ell_3} = C\Big(\frac{\ell_2}{\ell_3}\Big)^3,
	\end{equation*}
	for some constant $C$ large enough. Define as well
	\begin{equation*}
		\cI_1 = \sum_{r\in\ell_2\mathbb{Z}^3}[M_r^{(3\ell_2)} < C_{\ell_2,\ell_3}](M_r^{(3\ell_2)})^2M_r^{(3\ell_1)}(M_r^{(3\ell_1)} - 1)
	\end{equation*}
	and
	\begin{equation*}
		\cI_2 = \sum_{r\in\ell_2\mathbb{Z}^3}[M_r^{(3\ell_2)} \geq C_{\ell_2,\ell_3}](M_r^{(3\ell_2)})^2M_r^{(3\ell_1)}(M_r^{(3\ell_1)} - 1),
	\end{equation*}
	Thus,
	\begin{equation*}
		\sum_{\substack{1\leq i,j,k,\ell\leq M\\ i\neq j\neq k\neq i\\ \ell\neq i,j,k}}\mathds{1}_{\{\vert x_i - x_j\vert \leq \ell_2\}}\mathds{1}_{\{\vert x_i - x_k\vert \leq \ell_1\}}\mathds{1}_{\{\vert x_i - x_\ell\vert \leq \ell_1\}} \leq \cI_1 + \cI_2.
	\end{equation*}
	Let us deal with $\cI_1$ and $\cI_2$ separately.
	\\
	
	\noindent
	\textbf{Analysis of $\cI_1$.} By definition of the $M_r^{(3\ell_1)}$, we can find some universal constant $C$ large enough so that
	\begin{equation*}
		M_r^{(3\ell_1)}(M_r^{(3\ell_1)} - 1) \leq \sum_{\substack{1\leq i,j\leq M\\ i\neq j}}\mathds{1}_{\{\vert x_i - x_j\vert \leq C\ell_1\}}\mathds{1}_{\Lambda_r^{(3\ell_1)}}(x_i).
	\end{equation*}
	Thus, $\cI_1$ is bounded by
	\begin{equation*}
		\cI_1 \leq C\Big(\frac{\ell_2}{\ell_3}\Big)^6\sum_{\substack{1\leq i,j\leq M\\ i\neq j}}\mathds{1}_{\left\{\vert x_i - x_j\vert \leq C\ell_1\right\}}\sum_{r\in\ell_2\mathbb{Z}^3}\mathds{1}_{\Lambda_r^{(3\ell_1)}}(x_i).
	\end{equation*}
	Using that
	\begin{equation*}
		\#\left\{r\in\ell_2\Z^3: x\in\Lambda_r^{(3\ell_1)}\right\} \leq C\Big(\frac{\ell_1}{\ell_2}\Big)^3,
	\end{equation*}
	for all $x\in\R^3$,	we obtain
	\begin{equation*}
		\cI_1 \leq C\Big(\frac{\ell_1}{\ell_3}\Big)^3\Big(\frac{\ell_2}{\ell_3}\Big)^3\sum_{\substack{1\leq i,j\leq M\\ i\neq j}}\mathds{1}_{\left\{\vert x_i - x_j\vert \leq C\ell_1\right\}}.
	\end{equation*}
	In other words, $\cI_1$ is bounded from above by a two-body potential of range $C\ell_1$. Finally, an application of Proposition~\ref{prop:general_two_body_potential_bound} yield
	\begin{equation}
		\label{eq:four_body_effective_potential_I1_bound}
		\cI_1 \leq C\Big(\frac{\ell_2}{\ell_3}\Big)^3\Big(\frac{\ell_1}{\ell_3}\Big)^3M + C\frac{\ell_2^3}{\ell_3}\Big(\frac{\ell_1}{\ell_3}\Big)^9\sum_{i=1}^M-\Delta_i + C\Big(\frac{\ell_2}{\ell_3}\Big)^3\Big(\frac{\ell_1}{\ell_3}\Big)^9\sum_{\substack{1\leq i,j,k \leq M\\ i\neq j\neq k\neq i}}\mathds{1}_{\left\{\vert(x_i - x_j,x_i - x_k)\vert\leq \ell_3\right\}}.
	\end{equation}
	\\
	
	\noindent
	\textbf{Analysis of $\cI_2$.} We split between $M_r^{(3\ell_1)} < C_{\ell_1,\ell_3}$ and $M_r^{(3\ell_1)} \geq C_{\ell_1,\ell_3}$, where $C_{\ell_1,\ell_3}$ is given by
	\begin{equation*}
		C_{\ell_1,\ell_3} = C\Big(\frac{\ell_1}{\ell_3}\Big)^3,
	\end{equation*}
	for some $C$ large enough. Namely, we define
	\begin{equation*}
		\cJ_1 = \sum_{r\in\ell_2\mathbb{Z}^3}(M_r^{(3\ell_2)})^2[M_r^{(3\ell_2)} \geq C_{\ell_2,\ell_3}](M_r^{(3\ell_1)})^2[M_r^{(3\ell_1)} < C_{\ell_1,\ell_3}]
	\end{equation*}
	and
	\begin{equation*}
		\cJ_2 = \sum_{r\in\ell_2\mathbb{Z}^3}(M_r^{(3\ell_2)})^2[M_r^{(3\ell_2)} \geq C_{\ell_2,\ell_3}](M_r^{(3\ell_1)})^2[M_r^{(3\ell_1)} \geq C_{\ell_1,\ell_3}],
	\end{equation*}
	and write
	\begin{equation*}
		\cI_2 \leq \cJ_1 + \cJ_2.
	\end{equation*}
	On the one hand, by definition of $C_{\ell_1,\ell_3}$ and $C_{\ell_2,\ell_3}$, we have
	\begin{equation*}
		\cJ_1 \leq C\Big(\frac{\ell_1}{\ell_3}\Big)^6\Big(\frac{\ell_3}{\ell_2}\Big)^3\sum_{r\in\ell_2\mathbb{Z}^3}(M_r^{(3\ell_2)})^3[M_r^{(3\ell_2)} \geq C_{\ell_2,\ell_3}].
	\end{equation*}
	Thanks to Lemmas~\ref{lemma:convexity_estimate_number_particles_box}~and~\ref{lemma:three_body_estimate_box}, this implies
	\begin{equation*}
		\mathcal{J}_1 \leq C\Big(\frac{\ell_1}{\ell_3}\Big)^6\Big(\frac{\ell_2}{\ell_3}\Big)^3\sum_{r\in\tilde{\ell}_3\mathbb{Z}^3}(M_s^{(\tilde{\ell}_3)})^3[M_s^{(\tilde{\ell}_3)} \geq 3] \leq C\Big(\frac{\ell_1}{\ell_3}\Big)^6\Big(\frac{\ell_2}{\ell_3}\Big)^3\sum_{\substack{1\leq i,j,k\leq M\\ i\neq j\neq k\neq i}}\mathds{1}_{\left\{\vert(x_i - x_j, x_i - x_k)\vert \leq \ell_3\right\}},
	\end{equation*}
	where we defined $\tilde{\ell}_3 = \ell_3/\sqrt{6}$. On the other hand, using the estimate $M_r^{(3\ell_1)} \leq M$ and Young's inequality for products, we can write
	\begin{equation*}
		\mathcal{J}_2 \leq CM\sum_{r\in\ell_2\mathbb{Z}^3}\Big(\varepsilon(M_r^{(3\ell_2)})^3[M_r^{(3\ell_2)} \geq C_{\ell_2,\ell_3}] + \varepsilon^{-1}(M_r^{(3\ell_1)})^3[M_r^{(3\ell_1)} \geq C_{\ell_1,\ell_3}]\Big),
	\end{equation*}
	Applying Lemmas~\ref{lemma:convexity_estimate_number_particles_box}~and~\ref{lemma:three_body_estimate_box} just as for $\cJ_1$, we obtain
	\begin{equation*}
		\mathcal{J}_2 \leq CM\Big(\frac{\ell_1}{\ell_3}\Big)^3\Big(\frac{\ell_2}{\ell_3}\Big)^3\sum_{\substack{1\leq i,j,k\leq M\\ i\neq j\neq k\neq i}}\mathds{1}_{\left\{\vert(x_i - x_j, x_i - x_k)\vert \leq \ell_3\right\}}.
	\end{equation*}
	Summing up, we have shown that
	\begin{equation}
		\label{eq:four_body_effective_potential_I2_bound}
		\cI_2 \leq C\Big(\frac{\ell_1}{\ell_3}\Big)^3\Big(\frac{\ell_2}{\ell_3}\Big)^3\Big(M + \Big(\frac{\ell_1}{\ell_3}\Big)^3\Big)\sum_{\substack{1\leq i,j,k\leq M\\ i\neq j\neq k\neq i}}\mathds{1}_{\left\{\vert(x_i - x_j,x_i - x_k)\vert \leq \ell_3\right\}}.
	\end{equation}
	
	Combining \eqref{eq:four_body_effective_potential_I1_bound} and \eqref{eq:four_body_effective_potential_I2_bound} concludes the proof of Corollary~\ref{cor:general_four_body_potential_bound}.
\end{proof}

\begin{proof}[Proof of Lemma~\ref{lemma:general_five_body_potential_bound}]
	We show
	\begin{multline*}
		\sum_{1\leq i,j,k,\ell,m \leq M}\mathds{1}_{\{\vert x_i - x_j\vert \leq \ell_1\}}\mathds{1}_{\{\vert x_i - x_k\vert \leq \ell_1\}}\mathds{1}_{\{\vert x_i - x_\ell\vert \leq \ell_1\}}\mathds{1}_{\{\vert x_i - x_m\vert \leq \ell_1\}}\\
		\leq C\Big(\dfrac{\ell_1}{\ell_2}\Big)^{12}M + CM\Big(\dfrac{\ell_1}{\ell_2}\Big)^3\sum_{1\leq i,j,k,\ell\leq M}\1_{\{\vert x_i - x_j\vert \leq \ell_2\}}\1_{\{\vert x_i - x_k\vert \leq \ell_1\}}\1_{\{\vert x_i - x_\ell\vert \leq \ell_1\}}
	\end{multline*}
	on $\mfH^M$, for any nonnegative integer $M$. We divide the space $\R^3$ into boxes of side length $\ell_1$ to find
	\begin{equation*}
		\sum_{1\leq i,j,k,\ell,m \leq M}\mathds{1}_{\{\vert x_i - x_j\vert \leq \ell_1\}}\mathds{1}_{\{\vert x_i - x_k\vert \leq \ell_1\}}\mathds{1}_{\{\vert x_i - x_\ell\vert \leq \ell_1\}}\mathds{1}_{\{\vert x - x_m\vert \leq \ell_1\}} \leq \sum_{r\in\ell_1\Z^3}(M_r^{(3\ell_1)})^5.
	\end{equation*}
	To estimate the right-hand side we start with the bound
	\begin{equation*}
		\sum_{r \in \ell_1\Z^3}(M_r^{(3\ell_1)})^5 \leq C\sum_{r \in \tilde{\ell}_1\Z^3}(M_r^{(\tilde{\ell}_1)})^5 + C,
	\end{equation*}
	where $\tilde{\ell}_1 = \ell_1/\sqrt{3}$. The reasons behind the choice of $\tilde{\ell}_1$ is that any two particles located in a box of side length $\tilde{\ell}_1$ are at a distance at most $\ell_1$ from one another. To prove this bound, we make a distinction between $M_r^{(3\ell_1)} \geq C$ and $M_r^{(3\ell_1)} < C$ for some $C$ large enough, and apply Lemma~\ref{lemma:convexity_estimate_number_particles_box} in the former case. Next, we split between $M_r^{(3\ell_1)} < C_{\ell_1,\ell_2}$ and $M_r^{(3\ell_1)} \geq C_{\ell_1,\ell_2}$ with $C_{\ell_1,\ell_2}$ given by
	\begin{equation*}
		C_{\ell_1,\ell_2} = C\Big(\dfrac{\ell_1}{\ell_2}\Big)^3,
	\end{equation*}
	for some $C$ large enough. Estimating the term $M_r^{(3\ell_1)} < C_{\ell_1,\ell_2}$ rather crudely, we find
	\begin{equation*}
		\sum_{r \in \ell_1\Z^3}(M_r^{(3\ell_1)})^5 \leq C\Big(\dfrac{\ell_1}{\ell_2}\Big)^{12}M + CM\sum_{r \in \tilde{\ell}_1\Z^3}(M_r^{(\tilde{\ell}_1)})^4[M_r^{(\tilde{\ell}_1)} \geq C_{\ell_1,\ell_2}],
	\end{equation*}
	Thanks to \eqref{eq:many_particles_large_box_bounded_smaller_box_power}, we deduce
	\begin{equation*}
		\sum_{r \in \ell_1\Z^3}(M_r^{(3\ell_1)})^5 \leq C\Big(\dfrac{\ell_1}{\ell_2}\Big)^{12}M + CM\Big(\dfrac{\ell_1}{\ell_2}\Big)^3\sum_{r \in \tilde{\ell}_1\Z^3}(M_r^{(\tilde{\ell}_1)})^2\sum_{s \in \tilde{\ell}_2\Z^3}(M_{r,s}^{(\tilde{\ell}_1,\tilde{\ell}_2)})^2[M_{r,s}^{(\tilde{\ell}_1,\tilde{\ell}_2)} \geq 4].
	\end{equation*}
	Following an analogous proof to that of Lemma~\ref{lemma:three_body_estimate_box}, we can bound the last term by
	\begin{equation*}
		CM\Big(\dfrac{\ell_1}{\ell_2}\Big)^3\sum_{1\leq i,j,k,\ell\leq M}\1_{\{\vert x_i - x_j\vert \leq \ell_2\}}\1_{\{\vert x_i - x_k\vert \leq \ell_1\}}\1_{\{\vert x_i - x_\ell\vert \leq \ell_1\}},
	\end{equation*}
	which concludes the proof of Lemma~\ref{lemma:general_five_body_potential_bound}.
\end{proof}

\subsection{A mixed $L^p$-norm regularity estimate}

\label{sec:mixed_norm_regularity_estimate}

\begin{lemma}[Regularity for mixed norm]
	\label{lemma:mixed_norm_regularity_estimate}
	Let $d \geq 1$ and let $1 \leq p,q,r \leq \ii$. Let $V:\R^d\times\R^d \rightarrow \R$ be nonnegative, radially symmetric and decreasing with compact support and of class $L^r(\R^{2d})$. Define the mixed norm
	\begin{equation*}
		\Vert V\Vert_{L^p,L^q} = \Vert\Vert V(x,y)\Vert_{L^p_x}\Vert_{L^q_y}.
	\end{equation*}
	Then, the norm $\Vert V\Vert_{L^p,L^q}$ is finite if
	\begin{equation*}
		r > \dfrac{2pq}{p + q}.
	\end{equation*}
\end{lemma}

By taking $p = 1, q = \ii$, this lemma implies that, for $V$ satisfying Assumption \ref{ass:potential}, the function $x\in\R^3\mapsto\Vert V(x,\cdot)\Vert_{L^1(\R^3)}$ is bounded.

\begin{proof}[Proof of Lemma \ref{lemma:mixed_norm_regularity_estimate}]
	The layer cake representation of the potential $V$ (see \cite[Theorem 1.13]{LiebLos01}) reads
	\begin{equation*}
		V(\bx) = \int_0^\ii\dd{t}\1(V(\bx) > t).
	\end{equation*}
	Moreover, because $V$ is radially symmetric decreasing with compact support, there exists a bounded nonincreasing function $t:[0,\ii) \mapsto R_t$ such that
	\begin{equation}
		\label{eq:layer_cake}
		V(\bx) = \int_0^\ii\dd{t}\1(\vert (\bx)\vert \leq R_t).
	\end{equation}
	In addition, the $L^p$ norm of $V$ can be expressed as
	\begin{equation}
		\label{eq:layer_cake_Lp_norm}
		\Vert V\Vert_p^p = p\int_0^\ii\dd{t}t^{p - 1}\vert\{\bx\in\R^{2d}: V(\bx) > t\}\vert = p\vert\bbB^{2d}\vert\int_0^\ii\dd{t}t^{p - 1}R_t^{2d},
	\end{equation}
	for any $p \geq 1$. Here, $\vert\bbB^{2d}\vert$ denotes the volume of the unit ball in dimension $2d$.
	
	Suppose that both $p$ and $q$ are finite. We use the Minkowski inequality a first time to bound, for a.e. $y\in\R^d$,
	\begin{equation*}
		\Big(\int_{\R^d}\d{x}V(x,y)^p\Big)^{1/p} \leq \int_0^\ii\d{t}\Big(\int_{\R^d}\d{x}\1(\vert (x,y)\vert \leq R_t)\Big)^{1/p} \leq C\int_0^\ii\d{t}R_t^{d/p}\1(\vert y\vert \leq R_t).
	\end{equation*}
	Then, we use the Minkowski inequality a second time to deduce
	\begin{equation*}
		\Vert V\Vert_{L^p,L^q} \leq C\int_0^\infty\d{t} R_t^{d/p}\Big(\int_{\R^d}\d{y}\1(\vert y\vert \leq R_t)\Big)^{1/q} \leq C\int_0^\ii\d{t}R_t^{d/p + d/q}.
	\end{equation*}
	Finally, an application of the Hölder inequality gives
	\begin{equation}
		\label{eq:mixed_norm_hoelder}
		\Vert V\Vert_{L^p,L^q} \leq C_r\Big(\int_0^\ii\d{t}(1 + t)^{r - 1}R_t^{2d}\Big)^{\frac{p + q}{2pq}},
	\end{equation}
	for some finite constant $C_r$ that depends only on $p,q,r$ and $d$. Alternatively, this can be derived using one-dimensional weighted Hardy inequalities (see e.g. \cite[Theorem 7]{BennettGro06}). In view of \eqref{eq:layer_cake_Lp_norm}, the right-hand side is a finite constant, which concludes the proof of Lemma \ref{lemma:mixed_norm_regularity_estimate} for finite $p$ and $q$. Since the proof is essentially the same when either $p$ or $q$ is infinite, we leave it to the reader.
\end{proof}

\section{Estimates with renormalised quantities: proof of Propositions~\ref{prop:general_two_body_bound_renormalised}~and~\ref{prop:general_three_and_four_body_bound_renormalised}}

\label{sec:effective_interaction_potentials_renormalised}

As explained at the beginning of Section~\ref{sec:effective_interaction_potentials_non_renormalised}, the estimates from Propositions \ref{prop:general_two_body_bound_renormalised} and \ref{prop:general_three_and_four_body_bound_renormalised} that do not involve renormalised operators follow directly from Proposition \ref{prop:general_two_body_potential_bound} and Corollaries \ref{cor:general_three_body_potential_bound} and \ref{cor:general_four_body_potential_bound}. To prove the estimates \eqref{eq:general_two_body_bound_renormalised} and \eqref{eq:general_four_body_bound_renormalised}, we conjugate the estimates \eqref{eq:general_two_body_bound_non_renormalised} and \eqref{eq:general_two_three_four_body_bound_non_renormalised} with the cubic transform $e^{K_t}$ from Section~\ref{subsec:removing_correlations}. Recall that the cubic operator $K_t$ is defined by
\begin{equation*}
	K_t = -\dfrac{1}{6}\int_{\R^9}(k_t(x,y,z)b_x^*b_y^*b_z^*\Theta_N - \hc),
\end{equation*}
with $\Theta_N$ given by
\begin{equation*}
	\Theta_N = \sqrt{1 - \dfrac{\cN_+}{N}}\sqrt{1 - \dfrac{\cN_+ + 1}{N}}\sqrt{1 - \dfrac{\cN_+ + 2}{N}}
\end{equation*}
and $k_t$ being the cubic kernel 
\begin{equation*}
	k_t(x,y,z) = N^{3/2}\omega_{\lambda,N}(x,y,z)\varphi_t(x)\varphi_t(y)\varphi_t(z),
\end{equation*}
where $\omega_{\lambda,N}$ is the scattering solution defined in Section~\ref{subsec:scattering_energy}. Note that the modified creation and annihilation operators $b_x^*,b_x$ and the factor $\Theta_N$ ensure that $K_t$ maps $\cF_+^{\leq N}(t)$ to itself, and so does $e^{K_t}$.

\subsection{Proof of Proposition~\ref{prop:general_two_body_bound_renormalised}}

Since $e^{K_t}$ is a unitary operator, we can conjugate the estimate \eqref{eq:general_two_body_bound_non_renormalised} to obtain
\begin{equation}
	\label{eq:cubic_conjugation_two_body_estimate}
	e^{K_t}\cV_\twoB e^{-K_t} \leq C_\varepsilon e^{K_t}\cN_+e^{-K_t} + \varepsilon e^{K_t}(\cK + \cV_N)e^{-K_t},
\end{equation}
for all $\varepsilon > 0$. Then, the claim follows from the bounds
\begin{equation}
	\label{eq:number_excitations_conjugated_cubic}
	e^{K_t}\cN_+e^{-K_t} \leq C(\cN^\ren + 1),
\end{equation}
\begin{equation}
	\label{eq:hamiltonian_conjugated_cubic}
	e^{K_t}(\cK + \cV_N)e^{-K_t} \leq C(\cN^\ren + \cK^\ren + \cV^\ren + N^{1/2})
\end{equation}
and
\begin{equation}
	\label{eq:effective_two_body_conjugated_cubic}
	e^{K_t}\cV_\twoB e^{-K_t} \geq C\cV_\twoB - CN^{-1/2}(\cN^\ren + \cK^\ren + \cV^\ren + N).
\end{equation}
\\

\noindent
\textbf{Analysis of $e^{K_t}\cN_+e^{-K_t}$.} Thanks to the Duhamel formula, we can write
\begin{equation*}
	e^{K_t}\cN_+e^{-K_t} = \cN_+ + \int_0^1\d{}se^{sK_t}\left[K_t,\cN_+\right]e^{-sK_t}.
\end{equation*}
Then, we use the CCR \ref{eq:creation_annihilation_op_modified_ccr}, the identity \eqref{eq:number_excitation_op_int_form} and the fact that $\Theta_N$ commutes with $\cN_+$, to expand the commutator as
\begin{equation*}
	\left[K_t,\cN_+\right] = -\dfrac{1}{6}\int_{\R^{12}}k_t(x,y,z)\left[b_x^*b_y^*b_z^*,b_{x'}^*b_{x'}\right]\Theta_N + \hc = \cN^\ren - \cN_+.
\end{equation*}
Next, we apply Lemma~\ref{lemma:number_operator_renormalised} to deduce
\begin{equation*}
	e^{K_t}\cN_+e^{-K_t} \leq C(\cN^\ren + 1) + C\int_0^1\d{}se^{sK_t}\cN_+e^{-sK_t}.
\end{equation*}
Thus, an application of the Grönwall lemma gives \eqref{eq:number_excitations_conjugated_cubic}. The same analysis can be used to show
\begin{equation}
	\label{eq:number_excitations_conjugated_cubic_partial_power}
	e^{sK_t}\cN_+e^{-sK_t} \leq C(\cN^\ren + 1),
\end{equation}
for all $s \in (0,1)$.
\\

\noindent
\textbf{Analysis of $e^{K_t}\cK e^{-K_t}$.} Using the Duhamel formula, we find
\begin{equation*}
	e^{K_t}\cK e^{-K_t} = \cK + \int_0^1\dd se^{sK_t}\left[K_t,\cK\right]e^{-sK_t}.
\end{equation*}
Thanks to the CCR \eqref{eq:creation_annihilation_op_modified_ccr}, the commutator can be rewritten as
\begin{equation*}
	\left[K_t,\cK\right] = \dfrac{1}{2}\int_{\R^9}(-\Delta_1 k_t)(x,y,z)b_x^*b_y^*b_z^*\Theta_N + \hc
\end{equation*}
This was already analysed in the proof of Lemma~\ref{lemma:generator_estimate_kinetic_cubic_interaction} (see \eqref{eq:kinetic_term_renormalised_first_term_rewritten} and above), where it was shown that
\begin{equation*}
	[K_t,\cK] = \cI_\Delta + \cE_\Delta^{(1)},
\end{equation*}
with $\cI_\Delta$ given by
\begin{equation*}
	\cI_\Delta = \dfrac{N^{3/2}}{3}\int_{\R^9}(V_Nf_N)(x,y,z)\varphi_t(x)\varphi_t(y)\varphi_t(z)b_x^*b_y^*b_z^*\Theta_N + \hc
\end{equation*}
and for some $\cE_\Delta^{(1)}$ satisfying
\begin{equation*}
	\pm\cE_\Delta^{(1)} \leq C\cK + CN^{-1/2}(\cN_+ + \cV_N + N).
\end{equation*}
Then, applying the Duhamel formula a second time, we deduce
\begin{equation*}
	e^{K_t}\cK e^{-K_t} = \cK + \cI_\Delta + \int_0^1\dd s\int_0^s\dd u e^{uK_t}[K_t,\cI_\Delta] e^{-uK_t} + \int_0^1\d{}se^{sK_t}\cE_\Delta^{(1)} e^{-sK_t}.
\end{equation*}
We claim that
\begin{equation*}
	[K_t,\cI_\Delta] = \dfrac{N^3}{3}\int_{\R^9}(V_N\omega_Nf_N)(x,y,z)\vert\varphi_t(x)\vert^2\vert\varphi_t(y)\vert^2\vert\varphi_t(z)\vert^2 + \cE_\Delta^{(2)},
\end{equation*}
for some $\cE_\Delta^{(2)}$ that satisfies
\begin{equation*}
	\pm\cE_\Delta^{(2)} \leq C\cN_+ + CN^{-1/2}(\cK + \cV_N + N).
\end{equation*}
The proof is essentially just a combination of the proof of \eqref{eq:kinetic_term_renormalised_second_term_rewritten} and of the analysis of the error term \eqref{eq:I_3_2_def}. Gathering the previous estimates and using \eqref{eq:kinetic_term_renormalised_first_term_rewritten} and \eqref{eq:kinetic_term_renormalised_second_term_rewritten}, we obtain
\begin{equation}
	\label{eq:cubic_conjugation_kinetic_operator_gronwall}
	e^{K_t}\cK e^{-K_t} = \cK^\ren + \cE_\Delta^{(3)} + \int_0^1\dd se^{-sK_t}\cE_\Delta^{(4)}e^{sK_t},
\end{equation}
for some $\cE_\Delta^{(3)}$ and $\cE_\Delta^{(4)}$ that satisfy
\begin{equation*}
	\pm\cE_\Delta^{(3)} \leq C\cN^\ren + + N^{-1/2}(\cK^\ren + \cV^\ren + N)
\end{equation*}
and
\begin{equation*}
	\pm\cE_\Delta^{(4)} \leq C\cN_+ + N^{-1/2}(\cK + \cV_N).
\end{equation*}
\\

\noindent
\textbf{Analysis of $e^{K_t}\cV_Ne^{-K_t}$.} Using Duhamel's formula a first time, we get
\begin{equation*}
	e^{K_t}\cV_N e^{-K_t} = \cV_N + \int_0^1\d{}se^{sK_t}\left[K_t,\cV_N\right]e^{-sK_t}.
\end{equation*}
Thanks to the CCR \eqref{eq:creation_annihilation_op_modified_ccr}, we expand the commutator in the right to find
\begin{equation*}
	\left[K_t,\cV_N\right] = \cI_6^{(1)} + \cI_6^{(2)} + \cI_6^{(3)},
\end{equation*}
with $\cI_6^{(1)},\cI_6^{(2)}$ and $\cI_6^{(3)}$ given by
\begin{equation*}
	\cI_6^{(1)} = \dfrac{1}{6}\int_{\R^9}(V_Nk_t)(x,y,z)b_x^*b_y^*b_z^*\Theta_N + \hc,
\end{equation*}
\begin{equation*}
	\cI_6^{(2)} = \dfrac{1}{2}\int_{\R^{12}}V_N(x,y,z)k_t(x,y,z')b_x^*b_y^*b_z^*b_{z'}^*b_z\Theta_N + \hc
\end{equation*}
and
\begin{equation*}
	\cI_6^{(3)} = \dfrac{1}{4}\int_{\R^{15}}V_N(x,y,z)k_t(x,y',z')b_x^*b_y^*b_z^*b_{y'}^*b_{z'}^*b_yb_z\Theta_N + \hc
\end{equation*}
The terms $\cI_6^{(2)}$ and $\cI_6^{(3)}$ were  already analysed, up to small variations, in the proof of Lemma~\ref{lemma:generator_commutator_estimates} (below \eqref{eq:kinetic_term_commutator_Q_ren_error}), where it was shown that
\begin{equation*}
	\left[K_t,\cV_N\right] = \cI_6^{(1)} + \cE_6^{(1)},
\end{equation*}
with $\cE_6^{(1)}$ satisfying
\begin{equation*}
	\pm\cE_6^{(1)} \leq C(\cN_+ + \cK + \cV_N).
\end{equation*}
In fact, the proof is simpler in the present case we do not need an error in terms of $\cK^\ren$ and $\cV^\ren$, which was the case in Lemma~\ref{lemma:generator_commutator_estimates} to get the rate $N^{-1/2}$ in Theorem~\ref{th:main_result}. Next, we apply Duhamel's formula a second time to deduce
\begin{equation*}
	e^{sK_t}\cI_6^{(1)}e^{-sK_t} = \cI_6^{(1)} + \int_0^s\d{}ue^{uK_t}[K_t,\cI_6^{(1)}]e^{-uK_t}.
\end{equation*}
Using the CCR \eqref{eq:creation_annihilation_op_modified_ccr}, we expand
\begin{multline*}
	[K_t,\cI_6^{(1)}] = \dfrac{1}{6}\int_{\R^9}(V_Nk_t)(x,y,z)(Q_t^{\otimes 3}\overline{k_t})(x,y,z) + \hc\\
	+ \dfrac{1}{2}\int_{\R^{12}}(V_Nk_t)(x,y,z)(Q_t\otimes Q_t\otimes \1\overline{k_t})(x,y,z')b_z^*b_{z'} + \hc\\
	+ \dfrac{1}{4}\int_{\R^{15}}(V_Nk_t)(x,y,z)(Q_t\otimes\1\otimes\1\overline{k_t})(x,y',z')b_y^*b_z^*b_{y'}b_{z'} + \hc
\end{multline*}
Again, this was already studied in the proof of Lemma~\ref{lemma:generator_commutator_estimates_cubic_derivative} (below \eqref{eq:I_3_2_def}) and following the same steps gives
\begin{equation*}
	[K_t,\cI_6^{(1)}] = \dfrac{1}{3}\int_{\R^9}(V_N\omega_N^2)(x,y,z)\vert\varphi_t(x)\vert^2\vert\varphi_t(y)\vert^2\vert\varphi_t(z)\vert^2 + \cE_6^{(2)},
\end{equation*}
for some $\cE_6^{(2)}$ that satisfies
\begin{equation*}
	\pm\cE_6^{(2)} \leq C\cN_+ + CN^{-1/2}(\cK^\ren + \cV^\ren + N).
\end{equation*}
Summing up, we have shown that
\begin{equation}
	\label{eq:cubic_conjugation_interaction_gronwall}
	e^{K_t}\cV_Ne^{-K_t} = \cV^\ren + \int_0^1\d{}se^{sK_t}\cE_6e^{-sK_t},
\end{equation}
for some $\cE_6$ that satisfies
\begin{equation*}
	\pm\cE_6 \leq C(\cN_+ + \cK + \cV_N).
\end{equation*}
\\

\noindent
\textbf{Bounding $e^{K_t}(\cK + \cV_N)e^{-K_t}$.} Combining \eqref{eq:cubic_conjugation_kinetic_operator_gronwall} and \eqref{eq:cubic_conjugation_interaction_gronwall}, and using \eqref{eq:number_excitations_conjugated_cubic_partial_power}, we find
\begin{equation*}
	e^{K_t}(\cK + \cV_N)e^{-K_t} \leq C(\cN^\ren + \cK^\ren + \cV^\ren + N^{1/2}) + \int_0^1\d{}se^{sK_t}(\cK + \cV_N)e^{-sK_t}.
\end{equation*}
Finally, an application of the Grönwall lemma gives \eqref{eq:hamiltonian_conjugated_cubic}. Following a similar analysis, we can show that, for all $s \in (0,1)$,
\begin{equation}
	\label{eq:kinetic_plus_interaction_bdd_renormalised_op_partial_power}
	e^{sK_t}(\cK + \cV_N)e^{-sK_t} \leq C(\cK^\ren + \cV^\ren + N).
\end{equation}
\\

\noindent
\textbf{Analysis of $e^{K_t}\cV_\twoB e^{-K_t}$.} The Duhamel formula implies
\begin{equation*}
	e^{K_t}\cV_\twoB e^{-K_t} = \cV_\twoB + \int_0^1\d{}se^{sK_t}[K_t,\cV_\twoB]e^{-sK_t}.
\end{equation*}
Thanks to the CCR \eqref{eq:creation_annihilation_op_modified_ccr}, we can expand the commutator in the right-hand side as
\begin{align*}
	[K_t,\cV_\twoB] &= \cI_\twoB^{(1)} + \cI_\twoB^{(2)},
\end{align*}
with $\cI_\twoB^{(1)}$ and $\cI_\twoB^{(2)}$ given by
\begin{equation*}
	\cI_\twoB^{(1)} = N^{1/2}\int_{\R^9}\Vert V(N^{1/2}(x - y),\cdot)\Vert_{L^1}k_t(x,y,x')b_x^*b_y^*b_{x'}^* + \hc
\end{equation*}
and
\begin{equation*}
	\cI_\twoB^{(2)} = N^{1/2}\int_{\R^{12}}\Vert V(N^{1/2}(x - y),\cdot)\Vert_{L^1}k_t(y,x',y')b_x^*b_y^*b_{x'}^*b_{y'}^*b_x + \hc
\end{equation*}
To bound $\cI_\twoB^{(1)}$, write
\begin{equation*}
	\cI_\twoB^{(1)} = N^{1/2}\int_{\R^6}\Vert V(N^{1/2}(x - y),\cdot)\Vert_{L^1}b_x^*b_y^*\int_{\R^3}\d{}x'k_t(x,y,x')b_{x'}^* + \hc
\end{equation*}
Then, it is easy to see that the Cauchy--Schwarz inequality and the bounds \eqref{eq:truncated_scattering_equation_Lp_estimates} and \eqref{eq:truncated_scattering_equation_LpLinf_estimates} imply
\begin{equation*}
	\pm\cI_\twoB^{(1)} \leq \cV^\twoB + CN^{-1/2}(\cN_+ + N).
\end{equation*}
To bound $\cI_\twoB^{(2)}$, we define
\begin{equation*}
	B_y = \int_{\R^6}\d{}y'\d{}z'k_t(y,x',y')b_{x'}^*b_{y'}^*
\end{equation*}
and use the Cauchy--Schwarz inequality as follows:
\begin{equation*}
	\pm\cI_\twoB^{(2)} \leq \cV_\twoB + \varepsilon CN^{1/2}\int_{\R^6}\Vert V(N^{1/2}(x - y),\cdot)\Vert_{L^1}b_x^*B_yB_y^*b_x.
\end{equation*}
Using the CCR \eqref{eq:creation_annihilation_op_modified_ccr} and the estimates \eqref{eq:truncated_scattering_equation_Lp_estimates} and \eqref{eq:truncated_scattering_equation_LpLinf_estimates}, we deduce
\begin{equation*}
	\pm\cI_\twoB^{(2)} \leq \cV_\twoB + CN^{-1}(\cN_+ + \tilde{\cV}_\twoB + \tilde{\cV}_\threeB),
\end{equation*}
where
\begin{equation*}
	\tilde{\cV}_\twoB = \int_{\R^6}\1(\vert x - y\vert \leq \lambda)b_x^*b_y^*b_xb_y
\end{equation*}
and
\begin{equation*}
	\tilde{\cV}_\threeB = \int_{\R^9}\1(\vert x - y\vert \leq \lambda)\1(\vert x - z\vert \leq \lambda)b_x^*b_y^*b_z^*b_xb_yb_z.
\end{equation*}
An application of Proposition \ref{prop:general_two_body_potential_bound} and Corollary \ref{cor:general_three_body_potential_bound} gives
\begin{equation*}
	\pm\cI_\twoB^{(2)} \leq \cV_\twoB + CN^{-1/2}(\cN_+ + \cK + \cV_N).
\end{equation*}
Gathering the previous bounds and using \eqref{eq:number_excitations_conjugated_cubic_partial_power} and \eqref{eq:kinetic_plus_interaction_bdd_renormalised_op_partial_power}, we obtain
\begin{equation*}
	e^{K_t}\cV_\twoB e^{-K_t} \geq \cV_\twoB - CN^{-1/2}(\cN^\ren + \cK^\ren + \cV^\ren + N) - 2\int_0^1\dd se^{sK_t}\cV_\twoB e^{-sK_t}.
\end{equation*}
Thus, the inequality \eqref{eq:effective_two_body_conjugated_cubic} follows from an application of the Grönwall lemma.

This concludes the proof of Proposition~\ref{prop:general_two_body_bound_renormalised}.

\subsection{Proof of Proposition~\ref{prop:general_three_and_four_body_bound_renormalised}}

We begin by conjugating \eqref{eq:general_two_three_four_body_bound_non_renormalised} with $e^{K_t}$ to obtain
\begin{equation*}
	e^{K_t}\cV_{\textmd{4B}}e^{-K_t} \leq Ce^{K_t}(\cN_+ + \cK + \cV_N)e^{-K_t}.
\end{equation*}
The right-hand side was already analysed in \eqref{eq:number_excitations_conjugated_cubic} and \eqref{eq:hamiltonian_conjugated_cubic}, and we thus only need to treat the left-hand side. To shorten some notation, define
\begin{equation*}
	\tilde{V}(x) = \Vert V(x,\cdot)\Vert_{L^1(\R^3)}.
\end{equation*}
Using the Duhamel formula, we have
\begin{equation*}
	e^{K_t}\cV_\fourB e^{-K_t} = \cV_\fourB + \int_0^1\d{}se^{sK_t}\left[K_t,\cV_\fourB\right]e^{-sK_t}.
\end{equation*}
Using the CCR \eqref{eq:creation_annihilation_op_modified_ccr}, we expand
\begin{equation*}
	\left[K_t,\cV_\fourB\right] = \sum_{i = 1}^7\cI_\fourB^{(i)},
\end{equation*}
where
\begin{align*}
	\cI_\fourB^{(1)} &= \dfrac{24}{\sqrt{N}}\int\tilde{V}(N^{1/2}(x - y))\mathds{1}_{\{\vert x - x'\vert \leq \lambda\}}\mathds{1}_{\{\vert x -  y'\vert \leq \lambda\}}\overline{k_t(x,y,x')}b_{y'}^*b_xb_yb_{x'}b_{y'} + \hc,\\
	\cI_\fourB^{(2)} &= \dfrac{24}{\sqrt{N}}\int\tilde{V}(N^{1/2}(x - y))\mathds{1}_{\{\vert x - x'\vert \leq \lambda\}}\mathds{1}_{\{\vert x -  y'\vert \leq \lambda\}}\overline{k_t(x,x',y')}b_{y}^*b_xb_yb_{x'}b_{y'} + \hc,\\
	\cI_\fourB^{(3)} &= \dfrac{12}{\sqrt{N}}\int\tilde{V}(N^{1/2}(x - y))\mathds{1}_{\{\vert x - x'\vert \leq \lambda\}}\mathds{1}_{\{\vert x -  y'\vert \leq \lambda\}}\overline{k_t(x,y,x'')}b_{x'}^*b_{y'}^*b_{x''}b_xb_yb_{x'}b_{y'} + \hc,\\
	\cI_\fourB^{(4)} &= \dfrac{12}{\sqrt{N}}\int\tilde{V}(N^{1/2}(x - y))\mathds{1}_{\{\vert x - x'\vert \leq \lambda\}}\mathds{1}_{\{\vert x -  y'\vert \leq \lambda\}}\overline{k_t(x',y',x'')}b_x^*b_y^*b_{x''}b_xb_yb_{x'}b_{y'} + \hc,\\
	\cI_\fourB^{(5)} &= \dfrac{12}{\sqrt{N}}\int\tilde{V}(N^{1/2}(x - y))\mathds{1}_{\{\vert x - x'\vert \leq \lambda\}}\mathds{1}_{\{\vert x -  y'\vert \leq \lambda\}}\overline{k_t(x,x',x'')}b_y^*b_{y'}^*b_{x''}b_xb_yb_{x'}b_{y'} + \hc,\\
	\cI_\fourB^{(6)} &= \dfrac{6}{\sqrt{N}}\int\tilde{V}(N^{1/2}(x - y))\mathds{1}_{\{\vert x - x'\vert \leq \lambda\}}\mathds{1}_{\{\vert x -  y'\vert \leq \lambda\}}\overline{k_t(x',x'',y'')}b_x^*b_y^*b_{y'}^*b_{x''}b_{y''}b_xb_yb_{x'}b_{y'} + \hc,
\end{align*}
and
\begin{equation*}
	\cI_\fourB^{(7)} = \dfrac{6}{\sqrt{N}}\int\tilde{V}(N^{1/2}(x - y))\mathds{1}_{\{\vert x - x'\vert \leq \lambda\}}\mathds{1}_{\{\vert x -  y'\vert \leq \lambda\}}\overline{k_t(x,x'',y'')}b_y^*b_{x'}^*b_{y'}^*b_{x''}b_{y''}b_xb_yb_{x'}b_{y'} + \hc
\end{equation*}
We explain how to bound  $\cI_\fourB^{(6)}$ and $\cI_\fourB^{(7)}$, and we leave the rest to the reader because it is done in a similar way. We start with $\cI_\fourB^{(6)}$. Define
\begin{equation*}
	A_{x'} = \int_{\R^6}\d{}x''\d{}y''\overline{k_t(x',x'',y'')}b_{x''}b_{y''}
\end{equation*}
and write
\begin{equation*}
	\cI_\fourB^{(6)} = \dfrac{6}{\sqrt{N}}\int_{\R^{12}}\tilde{V}(N^{1/2}(x - y))\mathds{1}_{\{\vert x - x'\vert \leq \lambda\}}\mathds{1}_{\{\vert x -  y'\vert \leq \lambda\}}b_x^*b_y^*b_{y'}^*A_{x'}b_xb_yb_{x'}b_{y'} + \hc
\end{equation*}
Then, we use the Cauchy--Schwarz inequality to split
\begin{equation*}
	\pm\cI_\fourB^{(6)} \leq \dfrac{C}{\sqrt{N}}\int_{\R^{12}}\tilde{V}(N^{1/2}(x - y))\mathds{1}_{\{\vert x - x'\vert \leq \lambda\}}\mathds{1}_{\{\vert x -  y'\vert \leq \lambda\}}b_x^*b_y^*b_{y'}^*A_{x'}A_{x'}^*b_xb_yb_{y'} + C\cV_\fourB.
\end{equation*}
Next, we use the commutation relations \eqref{eq:creation_annihilation_op_modified_ccr} to put the first term in normal order, and apply the estimates \eqref{eq:truncated_scattering_equation_Lp_estimates} and \eqref{eq:truncated_scattering_equation_LpLinf_estimates} to deduce
\begin{multline*}
	\pm N^{-1/2}\int_{\R^{12}}\tilde{V}(N^{1/2}(x - y))\mathds{1}_{\{\vert x - x'\vert \leq \lambda\}}\mathds{1}_{\{\vert x -  y'\vert \leq \lambda\}}b_x^*b_y^*b_{y'}^*A_{x'}A_{x'}^*b_xb_yb_{y'}\\
	\leq CN^{-3/2}\cV_\threeB + CN^{-1}\cV_\fourB + CN^{-1}\cV_\fiveB,
\end{multline*}
where $\cV_\threeB$ and $\cV_\fourB$ are defined as in Proposition~\ref{prop:general_three_and_four_body_bound_renormalised}, and $\cV_\fiveB$ is given by
\begin{equation*}
	\cV_\fiveB = N^{-1/2}\int_{\R^{15}}\tilde{V}(N^{1/2}(x - y))\mathds{1}_{\{\vert x - z\vert \leq \lambda\}}\mathds{1}_{\{\vert x -  u\vert \leq \lambda\}}\mathds{1}_{\{\vert x -  v\vert \leq \lambda\}}b_x^*b_y^*b_z^*b_u^*b_v^*b_xb_yb_zb_ub_v.
\end{equation*}
Bounding $\mathds{1}(\vert x - v\vert \leq \lambda)$ by $1$, integrating over $v$ using \eqref{eq:number_excitation_op_int_form} and $\cN_+ \leq N$ on $\cF_+^{\leq N}$, we obtain
\begin{equation*}
	\cV_\fiveB \leq CN\cV_\fourB.
\end{equation*}
Hence, may apply \eqref{eq:general_two_three_four_body_bound_non_renormalised} to deduce
\begin{equation*}
	\pm\cI_\fourB^{(6)} \leq CN^{-3/2}(\cN_+ + \cK + \cV_N) + C\cV_\fourB.
\end{equation*}

Regarding $\cI_\fourB^{(7)}$, following the same strategy gives
\begin{equation*}
	\pm\cI_\fourB^{(7)} \leq C\cV_\fourB + CN^{-2}(\tilde{\cV}_\threeB + \tilde{\cV}_\fourB + \tilde{\cV}_\fiveB),
\end{equation*}
where $\tilde{\cV}_\threeB, \tilde{\cV}_\fourB$ and $\tilde{\cV}_\fiveB$ are given by
\begin{equation*}
	\tilde{\cV}_\threeB = \int_{\R^9}\mathds{1}_{\{\vert x - y\vert \leq C\lambda\}}\mathds{1}_{\{\vert x - z\vert \leq C\lambda\}}b_x^*b_y^*b_z^*b_xb_yb_z,
\end{equation*}
\begin{equation*}
	\tilde{\cV}_\fourB = \int_{\R^{12}}\mathds{1}_{\{\vert x - y\vert \leq C\lambda\}}\mathds{1}_{\{\vert x - z\vert \leq C\lambda\}}\mathds{1}_{\{\vert x - u\vert \leq C\lambda\}}b_x^*b_y^*b_z^*b_u^*b_xb_yb_zb_u
\end{equation*}
and
\begin{equation*}
	\tilde{\cV}_\fiveB = \int_{\R^{15}}\mathds{1}_{\{\vert x - y\vert \leq C\lambda\}}\mathds{1}_{\{\vert x - z\vert \leq C\lambda\}}\mathds{1}_{\{\vert x - u\vert \leq C\lambda\}}\mathds{1}_{\{\vert x - v\vert \leq C\lambda\}}b_x^*b_y^*b_z^*b_u^*b_v^*b_xb_yb_zb_ub_v.
\end{equation*}
An application of Corollaries \ref{cor:general_three_body_potential_bound}~and~\ref{cor:general_four_body_potential_bound} and of Lemma~\ref{lemma:general_five_body_potential_bound} implies
\begin{equation*}
	\pm\cI_\fourB^{(7)} \leq CN^{-1}(\cK + \cV_N) + C(\cN_+ + \cV_\fourB).
\end{equation*}

Bounding $\cI_\fourB^{(1)}, \dots,\cI_\fourB^{(5)}$ similarly as above, we finally find
\begin{equation*}
	e^{K_t}\cV_\fourB e^{-K_t} \geq \cV_\fourB - CN^{-1}\int_0^1\d{}se^{sK_t}(\cK + \cV_N)e^{-sK_t} - C\int_0^1\d{}se^{sK_t}(\cN_+ + \cV_\fourB)e^{-sK_t}.
\end{equation*}
The conclusion of the proof is the same as for Proposition~\ref{prop:general_two_body_bound_renormalised}.

\subsection{Bound for a quadratic term: proof of \eqref{eq:bound_L21}}

\label{subsec:quadratic_term_estimate}

Recall that
\begin{equation*}
	\mathcal{L}_2^{(1)} = \Theta_2^{(1)}\dfrac{N^2}{2}\int_{\R^9} V_N(x,y,z)|\varphi_t(z)|^2 \overline{\varphi_t(x) \varphi_t(y)} b_xb_y + \hc
\end{equation*}
Let us show
\begin{equation}
	\label{eq:quadratic_term_lower_bound_non_renormalised}
	\pm\cL_2^{(1)} \leq \varepsilon\cK + C\cN_+ + C_\varepsilon N^{1/2},
\end{equation}
for all $\varepsilon > 0$. Then we shall conjugate this inequality by $e^{K_t}$ to deduce the bound
\begin{equation}
	\label{eq:quadratic_term_lower_bound_renormalised}
	\pm\cL_2^{(1)} \leq \varepsilon\cK^\ren + \varepsilon\cV^\ren + C\cN^\ren + C_\varepsilon N^{1/2},
\end{equation}
for all $\varepsilon > 0$.

\begin{proof}[Proof of \eqref{eq:quadratic_term_lower_bound_non_renormalised}]
	Using that $\Theta_2^{(1)}$ satisfies \eqref{eq:theta_coef_estimate_almost_1}, we can show that
	\begin{equation*}
		\mathcal{L}_2^{(1)} = \dfrac{N^2}{2}\int_{\R^9} V_N(x,y,z)|\varphi_t(z)|^2 \overline{\varphi_t(x) \varphi_t(y)}b_xb_y + \hc + \cE_2^{(1)},
	\end{equation*}
	for some $\cE_2^{(1)}$ such that
	\begin{equation*}
		\pm\cE_2^{(1)} \leq C\cN_+.
	\end{equation*}	
	Next, define the kernel
	\begin{equation*}
		\eta_t(x,y) = \dfrac{N^2}{2}\int_{\R^3}\d{}z V_N(x,y,z)|\varphi_t(z)|^2 \overline{\varphi_t(x) \varphi_t(y)},
	\end{equation*}
	as well as the quadratic Hamiltonian
	\begin{equation*}
		\bbH = \cN + \varepsilon\d{}\Gamma(-\Delta) - \int_{\R^6}\left(\eta_t(x,y)a_xa_y + \hc\right).
	\end{equation*}
	Then, an application of \cite[Theorem 2]{NamNapSol16} yields
	\begin{equation}
		\label{eq:quadratic_hamiltonian_abstract_lower_bound}
		\bbH \geq -\tr{\left(\eta_t(1 + \varepsilon(-\Delta))^{-1}\eta_t^*\right)}/2.
	\end{equation}
	Using that $(1 + \varepsilon(-\Delta))^{-1} \leq \varepsilon^{-1}(-\Delta)^{-1}$ in the quadratic form sense, and the green function of the Laplacian (see e.g. \cite[Theorem 6.20]{LiebLos01}), we can bound the right-hand side by
	\begin{equation*}
		\tr{\left(\eta_t(1 + \varepsilon(-\Delta))^{-1}\eta_t^*\right)} \leq C\varepsilon^{-1}N^{1/2}.
	\end{equation*}
	Injecting this into \eqref{eq:quadratic_hamiltonian_abstract_lower_bound} and restricting the quadratic form domain to $\cF_+^{\leq N}(t)$ proves the claim.
\end{proof}

\begin{proof}[Proof of \eqref{eq:quadratic_term_lower_bound_renormalised}]
	To prove \eqref{eq:quadratic_term_lower_bound_renormalised}, we conjugate \eqref{eq:quadratic_term_lower_bound_non_renormalised} with the cubic transform $e^{K_t}$, which gives
	\begin{equation*}
		\pm e^{K_t}\cL_2^{(1)}e^{-K_t} \leq Ce^{K_t}\cN_+e^{-K_t} + C_\varepsilon N^{1/2} + \varepsilon e^{K_t}\cK e^{-K_t}.
	\end{equation*}
	The right-hand side is bounded using \eqref{eq:number_excitations_conjugated_cubic} and \eqref{eq:hamiltonian_conjugated_cubic} (note that since $e^{K_t}\cV_Ne^{-K_t}$ is nonnegative it can be added for free), and, therefore, we only need to take care of the left-hand side. Thanks to the Duhamel formula, we have
	\begin{equation*}
		e^{K_t}\cL_2^{(1)}e^{-K_t} = \cL_2^{(1)} + \int_0^1\dd{s}e^{sK_t} [K_t,\cL_2^{(1)}]e^{-sK_t}.
	\end{equation*}
	The analysis of the commutator in the right-hand side is rather straightforward and similar to that of the error terms in Section~\ref{subsec:generator_commutator_estimates_errors}, so we leave it to the reader. We find
	\begin{equation*}
		\pm[K_t,\cL_2^{(1)}] \leq C\cN_+ + CN^{-1/2}(\cK + \cV_N + N).
	\end{equation*}
	Then, we may conclude the proof of \eqref{eq:quadratic_term_lower_bound_renormalised} just as we did in the proof of Proposition~\ref{prop:general_two_body_bound_renormalised}.
\end{proof}

\appendix

\section{Properties of the Gross--Pitaevskii equation}

\label{sec:prop_gp_equation}

In the following proposition, we collect some important properties of the solution to the Gross--Pitaevskii equation \eqref{eq:gross_pitaevskii_equation}, which were shown in \cite{ColliandKeeStaTak08}.

\begin{proposition}
	Let $\varphi_0$ be normalised and such that $\cE^\GP[\varphi_0] < \infty,$ where $\cE^\GP$ is the energy functional defined in \eqref{eq:gross_pitaevskii_energy}. Then, there exists a unique global solution $t\mapsto\varphi_t\in C\left(\R,H^1(\R^3)\right)$ to \eqref{eq:gross_pitaevskii_equation} with initial data $\varphi_0$. The solution $\varepsilon_t$ is normalised and has energy $\cE^\GP[\varphi_t] = \cE^\GP[\varphi_0]$, which implies
	\begin{equation}
		\label{eq:gross_pitaevskii_solution_estimate_H1_L6}
		\sup_{t\in\R}\Vert\varphi_t\Vert_{H^1} \leq C \quad \textmd{and} \quad \sup_{t\in\R}\Vert\varphi_t\Vert_{L^6} \leq C,
	\end{equation}
	for some constant $C > 0$ that only depends on $\cE^\GP[\varphi_0]$. Moreover, if $\varphi_0\in H^6(\R^3)$, then $\varphi_t \in H^6(\R^3)$ and the uniform bounds
	\begin{equation}
		\label{eq:gross_pitaevskii_solution_estimate_H4}
		\sup_{t\in\R^3}\Vert\varphi_t\Vert_{H^6} \leq C,
	\end{equation}
	\begin{equation}
		\label{eq:gross_pitaevskii_solution_estimate_Linf}
		\sup_{t\in\R}\Vert\varphi_t\Vert_{L^\infty} \leq C, \quad \sup_{t\in\R}\Vert\partial_t\varphi_t\Vert_{L^\infty} \leq C \quad \textmd{and} \quad \sup_{t\in\R}\Vert\partial_t^2\varphi_t\Vert_{L^\infty} \leq C,
	\end{equation}
	hold for some constant $C > 0$ that depends only on $\Vert\varphi_0\Vert_{H^6}$.
\end{proposition}

\begin{proof}
	The global existence and uniqueness of $\varphi_t$, as well as the higher order regularity estimate \eqref{eq:gross_pitaevskii_solution_estimate_H4} were shown in \cite{ColliandKeeStaTak08}. The estimates \eqref{eq:gross_pitaevskii_solution_estimate_Linf} directly follow from the the Gross--Pitaevskii equation \eqref{eq:gross_pitaevskii_equation} and the Sobolev embedding theorem (see \cite[Theorem 8.8]{LiebLos01}).
\end{proof}

\section{Useful operator bounds}

\label{app:op_bounds}

The following lemma is a simple adaptation of \cite[Lemma 24]{NamRicTri23}.

\begin{lemma}
	\label{lemma:three_body_kernel_second_quantisation_estimates}
	Let $V$ satisfy Assumption~\ref{ass:potential} and let $\omega_{\lambda,N}$ and $\varepsilon_\lambda$ be as in Lemma~\ref{lemma:truncated_three_body_scattering_solution}. Let $\varphi\in L^\infty(\R^3)$. Define
	\begin{equation*}
		T_\lambda(f)(z) = \varphi(z)\int_{\R^6}\d{}x\d{}yN^{3/2}\omega_{\lambda,N}(x,y,z)\varphi(x)\varphi(y)f(x,y) + \hc,
	\end{equation*}
	\begin{equation*}
		U_\lambda(f)(z) = \varphi(z)\int_{\R^6}\d{}x\d{}yN^{3/2}(\nabla_x\omega_{\lambda,N}(x,y,z))\cdot(\nabla\varphi)(x)\varphi(y)f(x,y) + \hc
	\end{equation*}
	and
	\begin{equation*}
		S_\lambda(f)(z) = \varphi(z)\int_{\R^6}\d{}x\d{}yN^{-1/2}\varepsilon_{\lambda}(x,y,z)\varphi(x)\varphi(y)f(x,y) + \hc
	\end{equation*}
	Then, there exists a constant $C > 0$ depending only on $V$ such that
	\begin{equation}
		\label{eq:three_body_kernel_op_norm_estimate}
		\Vert T_\lambda\Vert_{\op} \leq C\Vert\varphi\Vert_\ii^3\lambda^{1/2}N^{-1/2} \quad \textmd{and} \quad \Vert S_\lambda\Vert_{\op} \leq C\Vert\varphi\Vert_\ii^3\lambda^{-3/2}N^{-1/2}.
	\end{equation}
	Moreover, for all $\varepsilon > 0$ small, there exists a constant $C_\varepsilon > 0$ depending only on $V$ and $\varepsilon$ such that 
		\begin{equation}
		\label{eq:three_body_kernel_gradient_op_norm_estimate}
		\Vert U_\lambda\Vert_{\op} \leq C_\varepsilon\Vert\nabla\varphi\Vert_\ii\Vert\varphi\Vert_\ii^2\lambda^{\varepsilon}N^{-1/2 + \varepsilon}.
	\end{equation}
	Consequently,
	\begin{equation}
		\label{eq:three_body_truncated_kernel_second_quantisation_estimate}
		\d{}\Gamma(T_\lambda^*T_\lambda) \leq C\Vert\varphi\Vert_\ii^6\lambda N^{-1}\cN_+, \quad \d{}\Gamma(T_\lambda T_\lambda^*) \leq C\Vert\varphi\Vert_\ii^6\lambda N^{-1}\cN_+^2
	\end{equation}
	\begin{equation}
		\label{eq:three_body_truncated_kernel_second_quantisation_estimate_gradient}
		\d{}\Gamma(U_\lambda^*U_\lambda) \leq C_\varepsilon\Vert\nabla\varphi\Vert_\ii^2\Vert\varphi\Vert_\ii^4\lambda^\varepsilon N^{-1 + \varepsilon}\cN_+, \quad \d{}\Gamma(U_\lambda U_\lambda^*) \leq C_\varepsilon\Vert\nabla\varphi\Vert_\ii^2\Vert\varphi\Vert_\ii^4\lambda^\varepsilon N^{-1 + \varepsilon}\cN_+^2,
	\end{equation}
	\begin{equation}
		\label{eq:three_body_truncated_kernel_second_quantisation_estimate_error}
		\d{}\Gamma(S_\lambda^*S_\lambda) \leq C\Vert\varphi\Vert_\ii^6\lambda^{-3} N^{-1}\cN_+ \quad \textmd{and} \quad \d{}\Gamma(S_\lambda S_\lambda^*) \leq C\Vert\varphi\Vert_\ii^6\lambda^{-3} N^{-1}\cN_+^2
	\end{equation}
	hold on $\cF_+$, for all $\varepsilon > 0$ small.
\end{lemma}

\begin{proof}
	Using \eqref{eq:truncated_scattering_equation_pointwise_estimate}, we can write
	\begin{equation*}
		\vert T_\lambda(f)(z)\vert \leq C\Vert\varphi\Vert_\infty^3 N^{-1/2}\int_{\R^3}\d{}y\mathds{1}_{\{\vert y - z\vert \leq C\lambda\}}\left(\int_{\R^3}\d{}x\dfrac{\vert f(x,y)\vert}{(\vert x - y\vert^2 + \vert y - z\vert^2)^2}\right).
	\end{equation*}
	Moreover, thanks to the Cauchy--Schwarz inequality we have
	\begin{align*}
		\int_{\R^3}\d{}x\dfrac{\vert f(x,y)\vert}{(\vert x - y\vert^2 + \vert y - z\vert^2)^2} &\leq \left(\int_{\R^3}\dfrac{\d{}u}{(\vert u\vert^2 + \vert y - z\vert^2)^4}\right)^{1/2}\left(\int_{\R^3}\d{}x\vert f(x,y)\vert^2\right)^{1/2}\\
		&\leq C\dfrac{1}{\vert y - z\vert^{5/2}}\Vert f(\cdot,y)\Vert_{L^2(\R^3)}.
	\end{align*}
	Hence, using Young's inequality, we obtain
	\begin{align*}
		\langle g,T(f)\rangle_{L^2(\R^3)} &\leq CN^{-1/2}\Vert\varphi_t\Vert_\infty^3\int_{\R^6}\d{}y\d{}z \vert g(z)\vert\dfrac{\mathds{1}_{\{\vert y - z\vert \leq C\lambda\}}}{\vert y - z\vert^{5/2}}\Vert f(\cdot,y)\Vert_{L^2(\R^3)}\\
		&\leq C\lambda^{1/2}N^{-1/2}\Vert\varphi_t\Vert_\ii^3\Vert g\Vert_2\Vert f\Vert_{L^2(\R^6)},
	\end{align*}
	for all $g\in L^2(\R^3),$ which implies $\Vert T_\lambda\Vert_\op \leq C\Vert\varphi_t\Vert_\ii^3\lambda^{1/2} N^{-1/2}$. To prove \eqref{eq:three_body_kernel_gradient_op_norm_estimate}, we use \eqref{eq:truncated_scattering_equation_pointwise_estimate} to deduce
	\begin{align*}
		\vert U_\lambda(f)(z)\vert &\leq C\Vert\nabla\varphi\Vert_\infty\Vert\varphi\Vert_\infty^2N^2\int_{\R^3}\d{}y\1(\vert y - z\vert\leq C\lambda)\int_{\R^3}\d{}x\dfrac{\vert f(x,y)\vert}{(1 + N^{1/2}(\vert x - y\vert^2 + \vert y - z\vert^2))^{5/2}}\\
		&\leq C\Vert\nabla\varphi\Vert_\infty\Vert\varphi\Vert_\infty^2N^{-1/2 + \varepsilon}\int_{\R^3}\d{}y\1(\vert y - z\vert\leq C\lambda)\int_{\R^3}\d{}x\dfrac{\vert f(x,y)\vert}{(\vert x - y\vert^2 + \vert y - z\vert^2)^{5/2 - \varepsilon}}.
	\end{align*}
	Then, we proceed just as above to obtain \eqref{eq:three_body_kernel_gradient_op_norm_estimate}. The estimate on $\Vert S_\lambda\Vert_\op$ is proven similarly using \eqref{eq:truncated_scattering_equation_error_pointwise_estimate}. The estimates \eqref{eq:three_body_truncated_kernel_second_quantisation_estimate}, \eqref{eq:three_body_truncated_kernel_second_quantisation_estimate_gradient} and \eqref{eq:three_body_truncated_kernel_second_quantisation_estimate_error} follow immediately.
\end{proof}

\section{Analysis of the prefactors in \eqref{eq:hamiltonian_conjugation_excitation_map}}

\label{app:theta_coefficients}

The coefficients $\Theta_j$ and $\Theta_j^{(i)}$ in \eqref{eq:hamiltonian_conjugation_excitation_map} are given by
\begin{align*}
	\Theta_0 &=  \left(1-\dfrac{\cN_+}{N}\right) \left(1-\dfrac{\cN_++1}{N}\right)\left(1-\dfrac{\cN_++2}{N}\right),\\
	\Theta_1 &=  \sqrt{1- \dfrac{\cN_+}{N}} \left(1-\dfrac{\cN_+ + 1}{N}\right)\left(1-\dfrac{\cN_+ + 2}{N}\right),\\
	\Theta_2^{(1)} &= \left(1-\dfrac{\cN_+}{N}\right)\sqrt{1- \dfrac{\cN_++1}{N}} \sqrt{1- \dfrac{\cN_++2}{N}},\\
	\Theta_2^{(2)} &= \Theta_2^{(3)} = \left(1-\dfrac{\cN_+}{N}\right)\left(1-\dfrac{\mathcal{N}_+ + 1}{N}\right), \\
	\Theta_3^{(1)} &= \sqrt{1-\dfrac{\mathcal{N}_+}{N}}\sqrt{1-\dfrac{\mathcal{N}_++1}{N}}\sqrt{1-\dfrac{\mathcal{N}_++2}{N}},\\
	\Theta_3^{(2)} &= \Theta_3^{(3)} = \sqrt{1-\dfrac{\mathcal{N}_+}{N}}\left(1-\dfrac{\cN_+ - 1}{N}\right),\\
	\Theta_4^{(1)} &= \sqrt{1-\dfrac{\mathcal{N}_+}{N}}\sqrt{1-\dfrac{\mathcal{N}_+ - 1}{N}},\\
	\Theta_4^{(2)} &= \Theta_4^{(3)} = 1 - \dfrac{\mathcal{N}_+}{N}, \quad \Theta_5 = \sqrt{1 -\dfrac{\mathcal{N}_+}{N}}.
\end{align*}

To prove \eqref{eq:theta_coef_estimate_almost_1}, we multiply and divide the left-hand side by $\Theta_j^{(i)} + 1$ to deduce
\begin{equation*}
	\vert\Theta_j^{(i)} - 1\vert = \dfrac{\vert(\Theta_j^{(i)})^2 - 1\vert}{\Theta_j^{(i)} + 1} \leq \vert(\Theta_j^{(i)})^2 - 1\vert \leq C_p(\cN_+ + 1)/N.
\end{equation*}

Let us prove \eqref{eq:theta_coef_estimate_difference} for $\Theta_1$. First, we write
\begin{align*}
	\Theta_1(\cN_+ + p) &= \left(1 - \dfrac{\cN_+ + 1 + p}{N}\right)\left(1 - \dfrac{\cN_+ + 2 + p}{N}\right)\sqrt{1 - \dfrac{\cN_+ + p}{N}}\\
	&= \left(1 - \dfrac{\cN_+ + 1 }{N}\right)\left(1 - \dfrac{\cN_+ + 2}{N}\right)\sqrt{1 - \dfrac{\cN_+ + p}{N}}\\
	&\phantom{=} + \left(-\dfrac{2p}{N} + \dfrac{2p\cN_+ + 3p + p^2}{N^2}\right)\sqrt{1 - \dfrac{\cN_+ + p}{N}}.
\end{align*}
Then, using the triangle inequality we deduce
\begin{align*}
	\vert\Theta_1(\cN_+ + p) - \Theta_1(\cN_+)\vert &\leq \left(1 - \dfrac{\cN_+ + 1 }{N}\right)\left(1 - \dfrac{\cN_+ + 2}{N}\right)\left\vert\sqrt{1 - \dfrac{\cN_+ + p}{N}} - \sqrt{1 - \dfrac{\cN_+}{N}}\right\vert + \dfrac{C_p}{N}.
\end{align*}
To bound the first term, we multiply and divide it by $\sqrt{1 - (\cN_+ + p)/N} + \sqrt{1 - \cN_+/N}$ and use
\begin{equation*}
	\left\vert\sqrt{1 - \dfrac{\cN_+ + p}{N}} - \sqrt{1 - \dfrac{\cN_+}{N}}\right\vert = \dfrac{p/N}{\sqrt{1 - \dfrac{\cN_+ + p}{N}} + \sqrt{1 - \dfrac{\cN_+}{N}}} \leq \dfrac{p/N}{\sqrt{1 - \dfrac{\cN_+}{N}}}.
\end{equation*}
Combining this with the estimate
\begin{equation*}
	\Theta_N/\sqrt{1 - \dfrac{\cN_+}{N}} \leq 1
\end{equation*}
gives \eqref{eq:theta_coef_estimate_difference} for $\Theta_1$. The other coefficients $\Theta_j$ and $\Theta_j^{(i)}$ are dealt with similarly.

\printbibliography

\end{document}